\documentclass[12pt]{article}
\usepackage[margin=0.84in]{geometry}
\usepackage{float}
\usepackage{authblk}
\usepackage{natbib}
\usepackage{amsmath}
\usepackage{amssymb}
\usepackage{amstext}
\usepackage{amsthm}
\usepackage{amsfonts}
\usepackage{mathtools}
\usepackage{color}
\usepackage{xcolor}
\usepackage{multirow}
\usepackage{mathrsfs}
\usepackage{listings}
\usepackage[hyphens]{url}

\newtheorem{thm}{Theorem}

\newtheorem{lem}{Lemma}
\newtheorem{cor}{Corollary}
\newtheorem{ass}{Assumption}

\theoremstyle{definition}

\newtheorem{exm}{Example}
\newtheorem{rmk}{Remark}

\DeclareMathOperator*{\argmin}{arg\,min}

\PassOptionsToPackage{hyphens}{url}

\usepackage{chngcntr}

\usepackage{listings}
\usepackage[hidelinks]{hyperref}
\hypersetup{
	colorlinks,
	linkcolor={blue},
	citecolor={blue},
	urlcolor={blue}
}

\mathtoolsset{showonlyrefs=true}

\title{Functional Synthetic Control Methods for Metric Space-Valued Outcomes}
\author{Ryo Okano\thanks{Graduate School of Economics, Hitotsubashi University} \thanks{RIKEN Center for Advanced Intelligence Project} \and Daisuke Kurisu\thanks{Center for Spatial Information Science, The University of Tokyo}}
\date{\today}

\begin{document}

\maketitle
\vspace{-2em}

\begin{abstract}
The synthetic control method (SCM) is a widely used tool for evaluating causal effects of policy changes in panel data settings. Recent studies have extended its framework to accommodate complex outcomes that take values in metric spaces, such as distributions, functions, networks, covariance matrices, and compositional data. 
However, due to the lack of linear structure in general metric spaces, 
theoretical guarantees for estimation and inference within these extended frameworks remain underdeveloped. 
In this study, we propose the functional synthetic control (FSC) method as an extension of the SCM for metric space-valued outcomes. To address challenges arising from the nonlinearlity of metric spaces, we leverage isometric embeddings into Hilbert spaces.
Building on this approach, we develop the FSC and augmented FSC estimators for counterfactual outcomes, with the latter being a bias-corrected version of the former.
We then derive their finite-sample error bounds to establish theoretical guarantees for estimation,  and construct prediction sets based on these estimators to conduct inference on causal effects.
We demonstrate the usefulness of the proposed framework through simulation studies and three empirical applications.

\end{abstract}

\section{Introduction}

The synthetic control method (SCM), originally proposed by \citet{abadie2003economic} and \citet{abadie2010synthetic}, is a widely used tool for evaluating the causal effects of policy changes in panel data settings.
It is designed for situations in which a single unit is exposed to a policy intervention while the other units are not, and their outcomes are observed both before and after the intervention.
In such settings, a weighted average of the control units that closely matches the treated unit’s pre-treatment outcomes---referred to as a synthetic control---is constructed to estimate the counterfactual outcomes of the treated unit.
Because of its simple and precise interpretability, the SCM has been widely applied in empirical research in economics and other disciplines \citep{abadie2021using}. 
Moreover, since its original introduction, the SCM has been extended and modified in various ways \citep{robbins2017framework, amjad2018robust, abadie2021penalized, ben2021augmented, ben2022synthetic}.

While causal inference has traditionally focused on scalar outcomes, or more generally, outcomes situated in a Euclidean space, modern data analysis increasingly encounters data with complex structures. Examples of such data include functions \citep{ wang2016functional}, distributions \citep{petersen2022modeling}, networks \citep{severn2022manifold}, covariance matrices \citep{arsigny2007geometric} and compositional data \citep{greenacre2021compositional}. 
Such data are often inherently non-Euclidean and can be described as elements of a metric space that satisfies certain structural conditions \citep{dubey2024metric}. 

Causal inference for outcomes that lie in a metric space is a rapidly evolving area. 
Recent studies include the estimation of average treatment effects under the unconfoundedness assumption \citep{lin2023causal, kurisu2024geodesic, bhattacharjee2025doubly}, the difference-in-differences \citep{opoku2023differences, zhou2025geodesic, boussim2025compositional}, and the regression discontinuity designs \citep{van2025regression, kurisu2025regression}. For the SCM, \cite{gunsilius2023distributional} introduced a framework for univariate distributional outcomes that employs the Wasserstein metric.  \cite{kurisu2025geodesic} proposed a geodesic synthetic control (GSC) method for outcomes residing in general geodesic metric spaces, encompassing distributions, networks, compositional data, and covariance matrices as examples. Importantly, they also developed the augmented GSC and the geodesic synthetic difference-in-differences, which are extensions of GSC, motivated by the augmented SCM introduced by \cite{ben2021augmented} and the synthetic difference-in-differences introduced by  \cite{arkhangelsky2021synthetic} for scalar outcomes, respectively. 

Although the existing SCM framework has been extended by the aforementioned studies, their estimation methods are not yet fully supported by solid theoretical foundations.
In the case of scalar outcomes rather than metric space-valued outcomes, \citet{abadie2010synthetic} demonstrated that when the pre-treatment fit by synthetic control is perfect,  the SCM estimator can be unbiased, or that its bias can be bounded, under certain data-generating assumptions.
\citet{ben2021augmented} further derived finite-sample error bounds for general weighting estimators, showing that the errors are governed by the quality of the pre-treatment fit and the norm of the weights. This result implies that (i) the SCM yields accurate estimates when it achieves good pre-treatment fit, and (ii) augmentation can further improve estimation accuracy by enhancing the pre-treatment fit.
However, for the extended SCM frameworks, comparable theoretical guarantees remain largely underdeveloped.
Moreover, inference and uncertainty quantification procedures, such as the construction of prediction sets for potential outcomes \citep{chernozhukov2021exact}, are also still being developed for the extended frameworks.
These challenges primarily stem from the lack of linear structure in the spaces where the outcomes reside.

In this paper, we propose the functional synthetic control (FSC) method as an extension of the SCM for metric space-valued outcomes. 
To address the challenges arising from the lack of linear structure, we leverage isometric embeddings of the metric spaces into  Hilbert spaces.
This approach enables the use of linear operations, inner products, and basis expansions in the latent Hilbert spaces without inducing metric distortions.
Such embeddings are feasible for various types of metric spaces, including spaces of distributions, networks, covariance matrices, and functions.
Note that causal inference methods based on isometric embeddings have also been explored in \citet{kurisu2025regression} and \citet{bhattacharjee2025doubly}.

Building on this approach, we first develop the FSC estimator as a direct extension of the SCM estimator, and then introduce the augmented FSC estimator as a bias-corrected version of the FSC estimator. For the latter, following \cite{ben2021augmented}, we focus on augmentation using a ridge regression model. We show that the ridge augmented FSC estimator possesses properties analogous to those of the ridge augmented SCM in \cite{ben2021augmented}: it can be expressed as a weighted average of the control units’ outcomes, and both the pre-treatment fit and the norm of the augmented weights are controlled by the ridge regularization parameter.

After proposing the estimators, we derive finite-sample estimation error bounds for general weighting estimators under several data generating processes, including functional versions of an autoregressive model and a linear factor model. These bounds show that,  similar to the case of scalar outcomes, estimation errors are governed by the quality of the pre-treatment fit and the norm of the weights. 
This implies that the FSC yields accurate estimates when it achieves a good pre-treatment fit.
The bounds further demonstrate that augmentation can improve estimation accuracy by enhancing the quality of the pre-treatment fit.
The bounds also show that the augmentation can improve estimation accuracy by enhancing the pre-treatment fit.
Finally, we extend inference procedures developed for the standard SCM to our setting, including the construction of prediction sets for counterfactual outcomes based on conformal inference \citep{chernozhukov2021exact} and the placebo permutation test \citep{abadie2010synthetic}.

Although our theoretical results may appear parallel to those of \citet{ben2021augmented}, their derivation is nontrivial in several important respects.
First, we consider outcomes taking values in generally infinite-dimensional spaces, whereas \citet{ben2021augmented} focus on scalar outcomes in a finite-dimensional space. To address infinite dimensionality, we employ orthogonal basis expansions and truncate the corresponding coefficient vectors when constructing augmented estimators. Our theoretical analysis explicitly accounts for the approximation error introduced by this truncation.
Second, we assume that the outcomes take values in a metric space that is isometrically embedded not in the entire Hilbert space but in a proper subset thereof. As a result, when mapping estimates from the latent Hilbert space back to the original metric space, a projection onto this subset is required. Our theoretical results incorporate the effect of this projection step on the estimation error.

The remainder of the paper is organized as follows. Section \ref{sec:setup} describes the setting and provides examples of outcomes that our framework can accommodate. Section \ref{sec:method} introduces the FSC and augmented FSC methods, and investigates the properties of the augmented estimator. Section \ref{sec:est_error} presents the finite-sample error bounds and the inference procedures. Section \ref{sec:simulations} reports simulation studies, and Section \ref{sec:empirical} illustrates the proposed framework with three empirical applications: the effect of an abortion legislation on fertility patterns in East Germany, the impact of the collapse of the Soviet Union on mortality, and the effect of the Brexit announcement on UK services trade. 
Finally, Section \ref{sec:conclusion} concludes with a brief discussion. 
Auxiliary results and theoretical proofs are provided in the Appendix.
Code to implement the methods, simulations, and empirical illustrations described in this paper can be found at
\url{https://github.com/RyoOkano21/FSC}. 

\section{Setup}
\label{sec:setup}
\subsection{Outcomes in Metric Space and Causal Effects}
Consider a canonical panel data setting for the synthetic control method (SCM), in which outcomes are observed for $N$ units over $T$ periods. The units are indexed by $i=1, \ldots, N$, and the time periods are indexed by $t=1, \ldots, T$.
In what follows, we fix both $N$ and $T$.
Assume that there exists a time period $T_0$ such that (i) no units receive treatment for $1 \le t \le T_0$, and (ii) only unit $i = 1$ receives treatment for $T_0 + 1 \le t \le T$.

In this study, we assume that the outcomes take values in a metric space $(\mathcal{M}, d)$ that can be isometrically embedded into a separable Hilbert space, in particular, an $L^2$ space or a Euclidean space. 
Specifically, let $\mathcal{H}$ denote either the $L^2$ space $L^2(\mathcal{I})$ of real-valued squared integrable functions on a compact interval $\mathcal{I}$, or the 
Euclidean space $\mathbb{R}^d$. 
We denote the equipped inner product by $\langle \cdot, \cdot \rangle_{\mathcal{H}}$ and define the corresponding norm  by $\|h\|_{\mathcal{H}} = \sqrt{\langle h, h \rangle}_{\mathcal{H}}$.  
When $\mathcal{H}$ is the $L^2$ space, the $L^2$ inner product $\langle f_1, f_2 \rangle_{\mathcal{H}} = \int_{\mathcal{I}}f_1(x)f_2(x)dx$ is equipped. 
When $\mathcal{H}$ is the Euclidean space, the standard inner product is equipped. 
We assume that there exists an isometry from $\mathcal{M}$ to $\mathcal{H}$, that is, a 
map $\Psi: \mathcal{M} \to \mathcal{H}$ such that $d(x, y) = \|\Psi(x) - \Psi(y)\|_{\mathcal{H}}$ for all $x, y \in \mathcal{M}$. Such a map $\Psi$ is automatically a bijection from $\mathcal{M}$ onto its image $\Psi(\mathcal{M}) \subset \mathcal{H}$. 
We further assume that the range $\Psi(\mathcal{M})$ is closed and convex in $\mathcal{H}$. Under these assumptions, 
the metric space $\mathcal{M}$ is a unique geodesic metric space \citep{kurisu2024geodesic}. 
Specifically, for $\alpha, \beta \in \mathcal{M}$, the unique geodesic between them, $\gamma_{\alpha, \beta}: [0, 1] \to \mathcal{M}$, is given by 
\begin{equation}
    \gamma_{\alpha, \beta}(s) = \Psi^{-1}((1-s)\Psi(\alpha) + s\Psi(\beta)), \quad s \in [0, 1]. 
    \label{eq:geodesic_form}
\end{equation}
Examples of metric spaces satisfying the above assumptions are presented in Section \ref{subsec:examples}.

For unit $i = 1, \ldots, N$ at time period $t = 1, \ldots, T$, we denote the potential outcomes under treatment and control by $\nu_{it}^I \in \mathcal{M}$ and $\nu_{it}^N \in \mathcal{M}$, respectively.
The observed outcome $\nu_{it}$ is 
\begin{equation}
    \nu_{it} = 
    \begin{cases}
        \nu_{it}^N & \text{if}\,\ i \ge 2 \,\ \text{or} \,\ t \le T_0, \\ 
        \nu_{it}^I & \text{if}\,\ i =1 \,\ \text{and} \,\ T_0+1 \le t \le T.
    \end{cases}
\end{equation}
We denote the image $\Psi(\mathcal{M})$ by $\mathcal{Y}$, and define the transformed outcomes in the space $\mathcal{Y}$ as
$Y_{it}^I = \Psi(\nu_{it}^I), Y_{it}^N = \Psi(\nu_{it}^N)$ and $Y_{it} = \Psi(\nu_{it})$.

Our estimands of interest are the counterfactual objects $\nu_{1t}^N, t=T_0+1, \ldots ,T$ for the first unit in the post-treatment periods. Once estimates of the counterfactual objects are obtained, 
one can estimate the individual treatment effects of the first unit, which can be flexibly defined according to the user’s specific interests. For example, following \cite{kurisu2025geodesic}, one can define the causal effects as the geodesics between the treatment and control potential outcomes, that is, 
\begin{equation}
    \tau_t
    =
    \gamma_{\nu_{1t}^N, \nu_{1t}^I}
    =
    \gamma_{\nu_{1t}^N, \nu_{1t}}, \quad t=T_0+1, \ldots , T.
\end{equation}
Since a geodesic connecting two points in a geodesic metric space embodies both the shortest distance and directional information between them, $\tau_t$ provide natural extensions for quantifying the treatment effects in geodesic metric spaces \citep{kurisu2024geodesic}.
Another possible option is to define  the differences between the transformed potential outcomes,  
\[
Y_{1t}^I-Y_{1t}^N = Y_{1t} - Y_{1t}^N, \quad t=T_0+1, \ldots, T,
\]
as the causal effects.
For example, when $\mathcal{M}$ is the space of one-dimensional probability distributions equipped with the 2-Wasserstein distance, one can set $\Psi(\nu)=F_\nu^{-1}(\cdot)$, where $F_\nu^{-1}$ is the quantile function of $\nu \in \mathcal{M}$  (see Example \ref{exm:distribution} in the next subsection). In this case, $Y_{1t}^{I} - Y_{1t}^N$ is defined as the difference between the corresponding quantile functions. 
Furthermore, one can consider the magnitudes of the treatment effects, which are defined as the lengths of the geodesics $\tau_t$, 
\[
d_t = d(\nu_{1t}^I, \nu_{1t}^N) = \|Y_{1t}^I-Y_{1t}^N\|_\mathcal{H},
\]
as the causal estimands.

\subsection{Examples}
\label{subsec:examples}
Examples of metric spaces that our approach can handle and that frequently arise in applications are given below. We will revisit some of these examples in the empirical analysis.

\begin{exm}[Functional data]
\label{exm:function}
Data consisting of functions is referred to as functional data
\citep{ramsay2005functional, hsing2015theoretical, wang2016functional}.
They are typically assumed to be situated in the space $L^2(\mathcal{I})$ of real-valued square integrable functions on a compact interval $\mathcal{I}$. This space is equipped with the inner product $\langle f_1, f_2 \rangle_{L^2}=\int_{\mathcal{I}} f_1(x)f_2(x)dx$, which induces the $L^2$ metric, given by $\|f_1-f_2\|_{L^2} = \sqrt{\int_\mathcal{I}\{f_1(x) - f_2(x)\}^2dx}$. 
\end{exm}

\begin{exm}[One-dimensional probability distributions]
\label{exm:distribution}
    Distributional data arise when each data point is regarded as a probability distribution
    \citep{petersen2022modeling, brito2022analysis}.
    Let $\mathcal{W}$ denote the space of probability distributions on a real line, with finite second moments, equipped with the 2-Wasserstein distance
    $d_{W}(\mu_1, \mu_2) = 
    \sqrt{\int_0^1 \{F_{\mu_1}^{-1}(u) - F_{\mu_2}^{-1}(u)\}^2du}$.
   Here, $F_{\mu}^{-1}$ denotes the quantile function of a probability distribution $\mu$.
    The metric space $(\mathcal{W}, d_W)$ is referred to as the Wasserstein space \citep{bigot2020statistical, panaretos2020invitation}, and distributional data are typically assumed to reside within this space. The mapping $\Psi(\mu) = F_{\mu}^{-1}$ is clearly an isometry from $\mathcal{W}$ to the $L^2$ space $L^2([0, 1])$. Moreover, the image $\Psi(\mathcal{W})$ is characterized as the set of all square-integrable, almost everywhere increasing functions on $[0, 1]$ \citep[Remark 2.2,][]{bigot2017geodesic}, and it
    is closed and convex in $L^2([0, 1])$ \citep[Proposition 2.1,][]{bigot2017geodesic}. 
\end{exm}

\begin{exm}[Symmetric positive semidefinite matrices]
\label{exm:covmat}
Let $\text{Sym}_m^{+}$ denote the space of $m \times m$ symmetric positive semidefinite matrices, and $\text{Sym}_m^{++}$ the space of $m \times m$ symmetric positive definite matrices.
These spaces have been studied under various metrics \citep{arsigny2007geometric, dryden2009non, pigoli2014distances, lin2019riemannian}.
The Frobenius metric $d_F$ on $\text{Sym}_m^+$ is defined as $d_F(A,B) = [\text{tr}\{(A-B)'(A-B)\}]^{1/2}$. In this case, the space $\text{Sym}_m^+$ obviously forms a closed convex subset of the Euclidean space $\mathbb{R}^{m^2}$.
For any constant $p > 0$, the power metric $d_{F, p}$ on $\text{Sym}_m^+$ is defined as  $d_{F, p}(A, B) = d_F(A^p, B^p)$. In this case, the mapping $\Psi(A) = A^p$ is an isometry from $\text{Sym}_m^+$ to the Euclidean space $\mathbb{R}^{m^2}$, and the image $\Psi(\text{Sym}_m^+) = \text{Sym}_m^+$ is closed and convex in $\mathbb{R}^{m^2}$.
The Log-Euclidean metric $d_{\mathrm{LE}}$ on $\text{Sym}_m^{++}$ is defined as $d_{\mathrm{LE}}(A,B) = d_F(\log(A),\log(B))$, where $\log(A)$ denotes the matrix logarithm of the matrix $A$.
In this case, the mapping $\Psi(A) =\log(A)$ is an isometry from $\text{Sym}_m^{++}$ to the Euclidean space $\mathbb{R}^{m^2}$, and the image $\Psi(\text{Sym}_m^{++}) = \{B \in \mathbb{R}^{m^2}: B = B'\}$ is closed and convex in $\mathbb{R}^{m^2}$. 
\end{exm}

\begin{exm}[Networks]
Consider a simple, undirected, and weighted network with a set of nodes $\{v_1, \ldots, v_m\}$ and a set of bounded edge weights $ \{w_{pq}: p, q=1, \ldots, m\}$, where  $0 \le w_{pq} \le W$. Such a network can be uniquely represented by its graph Laplacian matrix $L = (l_{pq}) \in \mathbb{R}^{m^2}$, defined as 
\[
l_{pq} = 
\begin{cases}
    -w_{pq} & \text{if}\,\, p \neq q, \\
    \sum_{r \neq p}w_{pr} & \text{if}\,\, p = q, 
\end{cases}
 \quad 
    p, q = 1, \ldots, m.
\]
The space of graph Laplacians is given by
\begin{equation}
    \mathcal{L}_m = \{L = (l_{pq}): L = L', L1_m = 0_m, -W \le l_{pq} \le 0 \,\, \text{for}\,\, p \neq q
    \},
\end{equation}
where $1_m$ and $0_m$ are the $m$-vectors of ones and zeros, respectively. This space provides a natural framework for characterizing network structures \citep{kola:14, severn2022manifold, zhou2022network}. Equipped with the Frobenius metric $d_F$, the space of graph Laplacians $\mathcal{L}_m$ forms a closed and convex subset of the Euclidean space $\mathbb{R}^{m^2}$  \citep[Proposition 1,][]{zhou2022network}.
\end{exm}

\begin{exm}[Compositional data]
Compositional data are non-negative multivariate data that convey relative rather than absolute information \citep{aitchison1982statistical, pawlowsky2011compositional, greenacre2021compositional}.
If researchers are interested in the space of compositional data with strictly positive components, that is,
    \begin{equation}
        \Delta_+^{d-1}
        =
        \left\{y \in \mathbb{R}^d: \sum_{j=1}^d y_j = 1, 
        y_j > 0 \,\ \text{for}\,\  j=1, \dots, d \right\},
    \end{equation}
    then they can endow $\Delta_+^{d-1}$ with the Aitchison metric $d_A$.  This distance is defined by $ d_A(x, y) = \|\Psi(x) - \Psi(y)\|_2$, where $\|\cdot\|_2$ is the Euclidean norm on $\mathbb{R}^d$, and 
    \begin{align}
    \Psi(x) &= \left(\log {\frac{x_1}  {g(x)}},\dots, \log {\frac{x_d} {g(x)}}\right),\quad g(x) = \left(\prod_{j=1}^d x_j\right)^{1/d}.
    \end{align}
    The map $\Psi$ is clearly an isometry from $\Delta_+^{d-1}$ to $\mathbb{R}^d$. Moreover, it can be shown that $\Psi(\Delta_+^{d-1}) = \{y \in \mathbb{R}^d:\sum_{j=1}^d y_j=0\}$, which is a closed convex subset of $\mathbb{R}^d$. 
\end{exm}

\begin{rmk}[Isometric embedding of metric spaces into Hilbert spaces]
There are several general theories concerning the embedding of metric spaces.  
A classical result, due to \cite{Schoenberg1935RemarksTM} and \cite{schoenberg1938metric}, states that a metric space \( (\mathcal{M}, d) \) can be isometrically embedded into a Hilbert space if and only if it has 2-negative type. That is,  
for any \( n \ge 2 \), \( x_1, \ldots, x_n \in \mathcal{M} \) and \( \alpha_1, \ldots, \alpha_n \in \mathbb{R} \) with \( \sum_{i=1}^n \alpha_i = 0 \), it holds that $\sum_{i=1}^n \sum_{j=1}^n \alpha_i \alpha_j d^2(x_i, x_j) \le 0$.
\end{rmk}

\section{Functional Synthetic Control Method and Its Augmentation}
\label{sec:method}
\subsection{Functional Synthetic Control Method}
The standard SCM estimates the missing potential outcomes for the treated unit as weighted averages of the control units' outcomes \citep{abadie2003economic, abadie2010synthetic}. 
The weights are chosen to balance the pre-treatment outcomes.
We extend the standard SCM to our setting as follows.
For simplicity, we omit auxiliary covariates here; their incorporation into our framework is discussed in Section \ref{sec:covariates} in the Appendix.

Let $\Delta^{N-1}$ denote the standard simplex in $\mathbb{R}^{N-1}:$ 
\[
\Delta^{N-1}
=
\left\{\gamma = (\gamma_2, \ldots, \gamma_N) \in \mathbb{R}^{N-1}: \sum_{i=2}^N \gamma_i = 1, \gamma_i \ge 0\,\ \text{for}\,\  i=2, \ldots, N\right\}.
\]
The functional synthetic control (FSC) method chooses weights $\hat{\gamma}^{\text{scm}} = (\hat{\gamma}_i^{\text{scm}})_{i=2}^N \in \Delta^{N-1}$ to balance the pre-treatment outcomes within the space $\mathcal{H}$, i.e., 
\begin{equation}
\hat{\gamma}^{\text{scm}}
\in
\argmin_{\gamma \in \Delta^{N-1}} \sum_{t=1}^{T_0} \left\| Y_{1t} - \sum_{i=2}^N \gamma_i Y_{it} \right\|_{\mathcal{H}}^2.
\label{eq:optim_fscm}
\end{equation}
For $t=T_0+1, \ldots, T$, the 
FSC estimator for the transformed counterfactual outcome $Y_{1t}^N$ is defined as 
\begin{equation}
    \hat{Y}_{1t}^{N, \mathrm{scm}}
    \in
    \sum_{i=2}^N \hat{\gamma}_i^{\mathrm{scm}}Y_{it}.
    \label{eq:fscm_est}
\end{equation}
Note that since the set $\mathcal{Y}$ is convex and $\hat{Y}_{1t}^{N, \mathrm{scm}}$ is a convex combination of $Y_{it}$'s, 
the estimator $ \hat{Y}_{1t}^{N, \mathrm{scm}}$ always lies in $\mathcal{Y}$.
The FSC estimator for the counterfactual outcome $\nu_{1t}^N$ is then defined as 
\begin{equation}
     \hat{\nu}_{1t}^{N, \mathrm{scm}}
    =
    \Psi^{-1}(\hat{Y}_{1t}^{N, \mathrm{scm}}).
    \label{eq:fscm_y}
\end{equation}

\begin{rmk}[Connection with the geodesic synthetic control estimator]
    \cite{kurisu2025geodesic} propose the geodesic synthetic control (GSC) estimator $\hat{\nu}_{1t}^{N, \mathrm{gsc}}$ for $\nu_{1t}^N$. This estimator is defined as a weighted Fr\'{e}chet mean \citep{Fréchet1948} of the control units' outcomes: $\hat{\nu}_{1t}^{N, \mathrm{gsc}} = \argmin_{\nu \in \mathcal{M}}\sum_{i=2}^N \hat{\gamma}_i^{\mathrm{gsc}}d^2(\nu, \nu_{it})$, where the weights are obtained by
    \begin{align}
        \hat{\gamma}^{\mathrm{gsc}} =    \argmin_{\gamma \in \Delta^{N-1}} \sum_{t=1}^{T_0}d^2(\nu_{1t}, \nu_{2:N, t}^{(\gamma)}), \quad 
        \nu_{2:N, t}^{(\gamma)}
        =
        \argmin_{\nu \in \mathcal{M}}
        \sum_{i=2}^N \gamma_i d^2(\nu, \nu_{it}).
    \end{align}
    In our setting, the FSC estimator $\hat{\nu}_{1t}^{N, \mathrm{scm}}$ and  the GSC estimator $\hat{\nu}_{1t}^{N, \mathrm{gsc}}$ are identical. 
    To see this, observe that for any $\gamma \in \Delta^{N-1}$, we have $\nu_{2:N, t}^{(\gamma)} = \Psi^{-1}(\sum_{i=2}^N \gamma_i Y_{it})$, which implies $\sum_{t=1}^{T_0 }d^2(\nu_{1t}, \nu_{2:N, t}^{(\gamma)}) = \sum_{t=1}^{T_0}\|Y_{1t} - \sum_{i=2}^N\gamma_i Y_{it}\|_{\mathcal{H}}^2$. Therefore, the weights $\hat{\gamma}^{\mathrm{scm}}$ and $\hat{\gamma}^{\mathrm{gsc}}$ are identical, and thus the estimators $\hat{\nu}_{1t}^{N, \mathrm{scm}}$ and $\hat{\nu}_{1t}^{N, \mathrm{gsc}}$ coincide.
\end{rmk}

\subsection{Augmented Functional Synthetic Control Method}
In many applications, the pre-treatment fit achieved by the weights $\hat{\gamma}^{\mathrm{scm}}$ is imperfect,  i.e., $Y_{1t} \neq \sum_{i=2}^N \hat{\gamma}_i Y_{it}$ for some $t \in \{1, \ldots, T_0\}$.
In such cases, the FSC estimator $\hat{\nu}_{1t}^{N, \mathrm{scm}}$  may be biased for the true potential outcome $\nu_{1t}^N$. 
To address this issue, we propose the augmented FSC method, a de-biased version of the FSC method following the approach of \cite{ben2021augmented}.

\paragraph{General form.}
Fix a post-treatment period $t \in \{T_0+1, \ldots, T\}$. For $i= 1, \ldots, N$, let $\hat{m}_{it} \in \mathcal{H}$ be an estimate of the control potential outcome $Y_{it}^N$ based on a regression model and the pre-treatment outcomes. We define the augmented FSC estimator for the transformed counterfactual $Y_{1t}^{N}$ as 
\begin{equation}
    \hat{Y}_{1t}^{N, \mathrm{aug}}
    =
   \hat{Y}_{1t}^{N, \mathrm{scm}}
    +
    \underbrace{\left(\hat{m}_{1t} - \sum_{i=2}^N \hat{\gamma}_i^{\mathrm{scm}}\hat{m}_{it}\right)}_{\text{estimate of the bias of $\hat{Y}_{1t}^{N, \mathrm{scm}}$} },
    \label{eq:aug_est}
\end{equation}
where $\hat{Y}_{1t}^{N, \mathrm{scm}}$ is the FSC estimator in \eqref{eq:fscm_est}, and
$\hat{\gamma}_i^{\mathrm{scm}}$'s are the weights obtained in \eqref{eq:optim_fscm}.
Unlike the FSC estimator, the augmented FSC estimator $\hat{Y}_{1t}^{N, \mathrm{aug}}$ does not necessarily lie in the space $\mathcal{Y}$. This motivates us to modify $\hat{Y}_{1t}^{N, \mathrm{aug}}$ by  projecting it  onto $\mathcal{Y}$:
\begin{equation}
    \tilde{Y}_{1t}^{N, \mathrm{aug}}
    =
    \argmin_{y \in \mathcal{Y}}\| y - \hat{Y}_{1t}^{N, \mathrm{aug}}\|_{\mathcal{H}},
    \label{eq:aug_fscm_modif}
\end{equation}
where the existence and uniqueness of the minimizer in \eqref{eq:aug_fscm_modif} are guaranteed by Lemma \ref{lem:proj} in Appendix \ref{sec:projection}.
The augmented FSC estimator for the counterfactual outcome $\nu_{1t}^N$ is then defined as 
\begin{equation}
     \hat{\nu}_{1t}^{N, \mathrm{aug}}
    =
    \Psi^{-1}(\tilde{Y}_{1t}^{N, \mathrm{aug}}).
    \label{eq:afscm_y}
\end{equation}

\begin{rmk}[Connection with the augmented geodesic synthetic control estimator]
As a bias-corrected version of the GSC estimator, \cite{kurisu2025geodesic} propose the augmented GSC estimator. For $i=1, \ldots, N$, let $\hat{\mu}_{it} \in \mathcal{M}$ denote an estimate of the control outcome $\nu_{1t}^N$ based on a regression model and the pre-treatment outcomes. 
The augmented  GSC estimator for $\nu_{1t}^N$ is defined as 
\begin{equation}
    \hat{\nu}_{1t}^{N, \mathrm{agsc}}
    =
    \Gamma_{\hat{\mu}_{0t}, \hat{\mu}_{1t}}(\hat{\nu}_{1t}^{N, \mathrm{gsc}}),
\end{equation}
where $\hat{\mu}_{0t} = \argmin_{\nu \in \mathcal{M}}\sum_{i=2}^N \hat{\gamma}_i^{\mathrm{gsc}}d^2(\nu, \hat{\mu}_{it})$, and $\Gamma_{\hat{\mu}_{0t}, \hat{\mu}_{1t}}(\cdot)$ denotes the geodesic transport map determined by $\hat{\mu}_{0t}$ and $\hat{\mu}_{1t}$ \citep{zhu2025geodesic}. The augmented GSC estimator and the augmented FSC estimator coincide when the geodesic transport map is appropriately specified. Specifically, suppose that $\Psi(\hat{\nu}_{1t}^{N, \mathrm{gsc}}) + \Psi(\hat{\mu}_{1t})-\Psi(\hat{\mu}_{0t}) \in \mathcal{Y}$. If we set $\Gamma_{\hat{\mu}_{0t}, \hat{\mu}_{1t}}(\hat{\nu}_{1t}^{N, \text{gsc}})=\Psi^{-1}(\Psi(\hat{\nu}_{1t}^{N, \text{gsc}}) + \Psi(\hat{\mu}_{1t})-\Psi(\hat{\mu}_{0t}))$, then the two estimators coincide.
\end{rmk}

\begin{rmk}[Modification via rearrangement method]
\label{rmk:rearrangement}
As shown in Example \ref{exm:distribution}, if $\mathcal{M}$ is the Wasserstein space, then $\mathcal{H}$ is the $L^2$ space $L^2([0, 1])$, and $\mathcal{Y}$ is the set of all square-integrable, almost everywhere increasing functions on $[0, 1]$. 
In this case, instead of using the projection in \eqref{eq:aug_fscm_modif}, one can apply the rearrangement method \citep{chernozhukov2009improving} to modify the estimator $\hat{Y}_{1t}^{N, \text{aug}}$. 
According to Proposition 1 in \cite{chernozhukov2009improving}, the increasing rearrangement of $\hat{Y}_{1t}^{N, \text{aug}}$, denoted by $\check{Y}_{1t}^{N, \text{aug}}$, weakly reduces the estimation error: $\|Y_{1t}^N - \check{Y}_{1t}^{N, \text{aug}}\|_{\mathcal{H}} \le \|Y_{1t}^N - \hat{Y}_{1t}^{N, \text{aug}}\|_{\mathcal{H}}$.
\end{rmk}

\paragraph{Ridge augmented FSC method.}
While the definition in \eqref{eq:aug_est} is general, the choice of the estimates $\hat{m}_{it}$ is crucial for understanding the  properties of the method.
In this study, we focus on a specific case where the estimates are obtained using a ridge-regularized linear model, following the approach of \cite{ben2021augmented}.
In what follows, to treat different cases in a unified manner, we represent the Hilbert space $\mathcal{H}$ as the space of real-valued squared integrable functions on a measure space $(\mathcal{X, \mathcal{A}, \mu})$. 
When $\mathcal{H}$ is the $L^2$ space $L^2(\mathcal{I})$, we take $\mathcal{X}$ as the interval $\mathcal{I}$ and $\mu$ as the Lebesgue measure on $\mathcal{I}$. When $\mathcal{H}$ is the Euclidean space $\mathbb{R}^d$, we take $\mathcal{X}$ as the finite set $\{1, \ldots, d\}$ and $\mu$ as the counting measure. 

As a preliminary, we center the pre-treatment outcomes $Y_{it}$. 
For $t=1, \ldots, T_0$, let $\bar{Y}_{t} = (N-1)^{-1}\sum_{i=2}^N Y_{it}$ be the mean of the control units' outcomes. Define 
\[
X_{it} = Y_{it} - \bar{Y}_t, \quad i=1, \ldots, N.
\]
Then, the values $X_{it}, i=2, \ldots, N$ are centered in the sense that $\sum_{i=2}^N X_{it} = 0$.

Fix a post-treatment period $t \in \{T_0+1, \ldots, T\}$ and consider the following estimate of $Y_{it}^N$ based on a linear regression model: 
\begin{equation}
    \hat{m}_{it}(x) = \eta_0(x) + \sum_{s=1}^{T_0}  \langle \eta_s(x, \cdot), X_{is}\rangle_{\mathcal{H}}, \quad x \in \mathcal{X},
    \label{eq:lin_est}
\end{equation}
where ${\eta}_0: \mathcal{X} \to \mathbb{R}$ and $\eta_s: \mathcal{X}^2 \to \mathbb{R}, s=1, \ldots, T_0$ are the intercept and coefficient functions.
Since $\mathcal{H}$ is a separable Hilbert space, it has a countable orthonormal basis $\{\varphi_k\}_{k=1}^\infty$. 
When $\mathcal{H} = L^2(\mathcal{I})$, typical examples include the Fourier basis, B-splines, and wavelets.
If $\mathcal{H} = \mathbb{R}^d$, we take $\varphi_k$ to be the $k$-th standard basis vector for $1 \le k \le d$, and the zero vector for $k > d$. 
Using the orthonormal basis 
$\{\varphi_k\}_{k=1}^\infty$,  for each $x$ and $s$, we expand the function $\eta_s(x, \cdot)$ as $\eta_s(x, y) = \sum_{k=1}^\infty {\theta}_{sk}(x)\varphi_k(y)$.  This yields $\langle {\eta}_s(x, \cdot), X_{is}\rangle_{\mathcal{H}} = \sum_{k=1}^\infty {\theta}_{s k}(x)r_{i,s,k}$, where 
$
r_{i,s,k} = \langle \varphi_k, X_{is} \rangle_{\mathcal{H}}.
$
 Hence, for a large integer $K$,  the estimate \eqref{eq:lin_est} can be approximated as  
 \[
 {\eta}_0(x) + \sum_{s=1}^{T_0}\sum_{k=1}^K {\theta}_{sk}(x)r_{i,s,k} = {\eta}_0(x) + \sum_{s=1}^{T_0} {\theta}_{s}(x)'r_{i s},
 \]
 where ${\theta}_s(x) = ({\theta}_{s1}(x), \ldots, {\theta}_{sK}(x))'$ and 
$r_{is} = (r_{i,s,1}, \ldots, r_{i,s,K})'$.
For each $x \in \mathcal{X}$,
we estimate the intercept ${\eta}_0(x)$ and  coefficients ${\theta}_s(x), s=1, \ldots, T_0$ via ridge regression. 
Specifically, our estimate of $Y_{it}^N$ is 
\begin{equation}
    \hat{m}_{it}^{\text{rid}}(x) = \hat{\eta}_0(x) + \sum_{s=1}^{T_0}\hat{\theta}_s(x)'r_{is}, \quad x \in \mathcal{X},
    \label{eq:ridge_estimate_approx}
\end{equation}
where \begin{align}
     \left\{\hat{\eta}_0(x), \hat{\theta}_1(x), \ldots, \hat{\theta}_{T_0}(x)\right\}  
    & =
    \argmin_{\eta_0, \theta_1, \ldots, \theta_{T_0}} \sum_{i=2}^N\left\{Y_{it}(x) - \left(\eta_0 + \sum_{s=1}^{T_0}\theta_s'r_{is}\right)\right\}^2 + \lambda \sum_{s=1}^{T_0} \|\theta_s\|_2^2. 
\end{align}
Here, 
$\lambda > 0$ is a regularization parameter.
The ridge augmented FSC estimator for $Y_{1t}^N$ is then 
\begin{equation}
    \hat{Y}_{1t}^{N, \text{aug}}
    =
    \hat{Y}_{1t}^{N, \text{scm}}
    +
    \left\{\hat{m}_{1t}^{\text{rid}} - \sum_{i=2}^N \hat{\gamma}_i^{\text{scm}}\hat{m}_{it}^{\text{rid}}\right\},
    \label{eq:est_ridge_y}
\end{equation}
and its modification $\tilde{Y}_{1t}^{N, \text{aug}}$ is defined as in \eqref{eq:aug_fscm_modif}. 
In addition, the ridge augmented FSC estimator $\hat{\nu}_{1t}^{N, \text{aug}}$  for $\nu_{1t}^{N}$ is defined as in
\eqref{eq:afscm_y}.

The next lemma shows that (i) the estimator \eqref{eq:est_ridge_y} is itself a weighting estimator whose weights adjust the FSC weights $\hat{\gamma}^{\text{scm}}$, and (ii) the augmented weights can be characterized as the solution to a penalized SCM problem, where the penalty term penalizes deviations from the FSC weights $\hat{\gamma}^{\text{scm}}$. This result is analogous to Lemma 1 in \cite{ben2021augmented}. 
For $i=1, \ldots, N$, define a $KT_0$-column vector $r_{i\cdot}$ and an $(N-1) \times (KT_0)$ matrix $r_{0\cdot}$ by 
\begin{equation*}
r_{i\cdot} = 
\begin{pmatrix}
    r_{i 1} \\ 
    \vdots \\ 
    r_{iT_0}
\end{pmatrix}, \quad 
    r_{0\cdot} = 
\begin{pmatrix}
    r_{2\cdot}' \\ 
    \vdots \\ 
    r_{N\cdot}'
\end{pmatrix}, \quad 
\end{equation*}
respectively. For any positive integer $p$, let $I_p$ denote the identity matrix of size $p$.

\begin{lem}
For any $t=T_0+1, \ldots, T$, the ridge augmented FSC estimator \eqref{eq:est_ridge_y} is expressed as 
    \begin{equation}
        \hat{Y}_{1t}^{N, \mathrm{aug}} = \sum_{i=2}^N \hat{\gamma}_i^{\mathrm{aug}}Y_{it},
    \end{equation}
    where the weights $\hat{\gamma}^{\mathrm{aug}} = (\hat{\gamma}_i^{\mathrm{aug}})_{i=2}^N$ are given by 
    \begin{equation}
        \hat{\gamma}_i^{\mathrm{aug}} = \hat{\gamma}_i^{\mathrm{scm}} + (r_{1\cdot} - r_{0\cdot}'\hat{\gamma}^{\mathrm{scm}})'(r_{0\cdot}'r_{0\cdot} + \lambda I_{KT_0})^{-1} r_{i \cdot}.
        \label{eq:weight_aug}
    \end{equation}
    Moreover, the weights $\hat{\gamma}^{\mathrm{aug}}$ are the solution to the following constrained optimization problem: 
    \begin{equation}
    \min_{\gamma \in \mathbb{R}^{N-1}} \|r_{1\cdot} - r_{0\cdot}'\gamma \|_2^2 + \lambda \|\gamma - \hat{\gamma}^{\mathrm{scm}}\|_2^2\quad  
        \mathrm{subject \,\ to \,}  \sum_{i=2}^{N}\gamma_i = 1.
        \label{eq:const_optim_prob}
    \end{equation}
    \label{lem:closed_form_RFASCM}
\end{lem}

 From \eqref{eq:weight_aug}, we see that the augmented weights \( \hat{\gamma}^{\text{aug}} \) can take negative values, in contrast to the FSC weights \( \hat{\gamma}^{\text{scm}} \). 
When the pre-treatment fit achieved by \( \hat{\gamma}^{\text{scm}} \) is perfect, i.e., $Y_{1t} = \sum_{i=2}^N \hat{\gamma}_i^{\text{scm}} Y_{it}$ for all $t=1, \ldots, T_0$, we have $r_{1\cdot} = r_{0\cdot}'\hat{\gamma}^{\text{scm}}$, and thus \( \hat{\gamma}^{\text{scm}} \) and \( \hat{\gamma}^{\text{aug}} \) coincide.  
In contrast, when the pre-treatment fit is imperfect, the ridge augmented FSC method may assign negative weights in estimating the counterfactual outcomes \( Y_{1t}^N \).
In other words, it allows extrapolation outside the convex hull of the control units' outcomes \citep{abadie2021using}.
The degree of extrapolation is governed by the regularization parameter \( \lambda \): as \( \lambda \) increases, the adjustment term in \eqref{eq:weight_aug} diminishes, and \( \hat{\gamma}^{\text{aug}} \) approaches \( \hat{\gamma}^{\text{scm}} \).

We next analyze how the pre-treatment fit achieved by the augmented weights,
\[
\sqrt{\sum_{t=1}^{T_0}\left\|Y_{1t} - \sum_{i=2}^N \hat{\gamma}_i^{\mathrm{aug}}Y_{it}\right\|_{\mathcal{H}}^2},
\]
and the norm of the augmented weights,  $\|\hat{\gamma}^{\text{aug}}\|_2$, depend on the parameter $\lambda$. This result will be used in the subsequent estimation error analysis.
To make the dependence on the number of the orthonormal vectors $K$ explicit,  we denote $r_{0\cdot}$ and $\hat{\gamma}^{\text{aug}}$ by $r_{0\cdot}^{(K)}$ and $\hat{\gamma}^{\text{aug}(K)}$, respectively. 
For any positive integer $K$, let $m(K)$ be the rank of the matrix $r_{0\cdot}^{(K)}$, and let $d_{\text{max}}^{(K)}, d_{\text{min}}^{(K)}$ be the maximum and minimum singular values of $r_{0\cdot}^{(K)}$, respectively. Note that since $r_{0\cdot}^{(K)}$ is an $(N-1) \times (KT_0)$ matrix, we have $m(K) \le N-1$ for all $K$.

\begin{ass}
    There exist constants $C_1 > 0$ and $c_1 > 0$ such that $d_{\mathrm{max}}^{(K)} \le C_1$ and $ d_{\mathrm{min}}^{(K)} \ge c_1$ hold for any positive integer $K$.
    \label{ass:singular_values}
\end{ass}

Under this assumption, we obtain the following result. 

\begin{lem}
Suppose Assumption \ref{ass:singular_values} holds. 
Then,  for any positive integer $K$,  the ridge augmented FSC weights $\hat{\gamma}^{\mathrm{aug}(K)}$ with regularization parameter $\lambda > 0$ satisfy the following:
\begin{enumerate}
    \item 
    \begin{equation}
        \sqrt{\sum_{t=1}^{T_0} \left\|Y_{1t} - \sum_{i=2}^N \hat{\gamma}_i^{\mathrm{aug}(K)}Y_{it}\right\|_{\mathcal{H}}^2}
        \le 
        \frac{\sqrt{m(K)}\lambda}{(d_{\mathrm{min}}^{(K)})^2+\lambda} \sqrt{\sum_{t=1}^{T_0} \left\|Y_{1t} - \sum_{i=2}^N \hat{\gamma}_i^{\mathrm{scm}}Y_{it}\right\|_{\mathcal{H}}^2}
        +
        R_{1}^{(K)},
        \label{eq:prefit_bound_aug}
    \end{equation}
     where $R_1^{(K)} \to 0$ as $K \to \infty$.
    \item 
    \begin{equation}
         \|\hat{\gamma}^{\mathrm{aug}(K)}\|_2 \le 
        \|\hat{\gamma}^{\mathrm{scm}}\|_2 +
        \frac{\sqrt{m(K)}d_{\mathrm{max}}^{(K)}}{(d_{\mathrm{min}}^{(K)})^2 + \lambda}\sqrt{\sum_{t=1}^{T_0} \left\|Y_{1t} - \sum_{i=2}^N \hat{\gamma}_i^{\mathrm{scm}}Y_{it}\right\|_{\mathcal{H}}^2} + R_2^{(K)},
        \label{eq:norm_bound_aug}
    \end{equation}
    where $R_2^{(K)} \to 0$ as $K \to \infty$.
\end{enumerate}
    \label{lem:aug_weight_prop}
\end{lem}

From \eqref{eq:prefit_bound_aug} and \eqref{eq:norm_bound_aug}, we observe that the bounds include remainder terms $R_1^{(K)}$ and $ R_2^{(K)}$. These terms arise from truncating the number of orthogonal vectors in $\mathcal{H}$, and they vanish as $K \to \infty$.
While Lemma \ref{lem:aug_weight_prop} is similar to Lemma 3 in \cite{ben2021augmented}, the presence of the remainder terms is specific to our setting.
The main term of the bound in \eqref{eq:prefit_bound_aug} decreases as  $\lambda$ decreases, and in particular, it vanishes as $\lambda \to 0$. This implies that the ridge augmented FSC method can achieve an almost perfect fit when $\lambda$ is close to zero (assuming $K$ is large).
In contrast, the main term of the bound in \eqref{eq:norm_bound_aug} increases as $\lambda$ decreases, indicating that the norm of the augmented weights can become large when $\lambda$ is small.

\section{Finite-Sample Error Bounds and Inference}
\label{sec:est_error}
This section mainly consists of two parts. Section \ref{subsec:error_auto} and \ref{subsec:error_factor} provide finite-sample error bounds for the proposed estimators, and Section \ref{subsec:inference} presents inference procedures.
To derive the error bounds, we consider two data-generating processes for transformed potential outcomes $Y_{it}^N$: (i) an autoregressive model, in which the control potential outcomes for the post-treatment periods are linear in their lagged outcomes, and (ii) a latent factor model, in which the control potential outcomes are linear in a set of latent factors. These assumptions about the data-generating processes are standard in the SCM literature \citep{abadie2010synthetic, ben2021augmented, ben2022synthetic}. 

\subsection{Error Bounds under Autoregressive Model} 
\label{subsec:error_auto}
In this and next subsections, for simplicity, we restrict attention to the case with a single post-treatment time period, i.e., $T = T_0 + 1$.
To clarify our discussion, we introduce a general weighting estimator for $Y_{1T}^N$, $\hat{Y}_{1T}^N = \sum_{i=2}^N \hat{\gamma}_i Y_{iT}$, where the weights $(\hat{\gamma}_i)_{i=2}^N$ are not dependent on the post-treatment outcomes $Y_{1T}, \ldots, Y_{NT}$, and satisfy $\sum_{i=2}^N \hat{\gamma}_i=1$. The ridge augmented FSC estimator $\hat{Y}_{1T}^{N, \text{aug}}$ takes this form, as shown in Lemma \ref{lem:closed_form_RFASCM}. 
We also define its modification,  
$
\tilde{Y}_{1T}^{N}
    =
    \argmin_{y \in \mathcal{Y}} \| y - \hat{Y}_{1T}^{N}\|_{\mathcal{H}},
$
and construct an estimator for $\nu_{1T}^N$ as 
$\hat{\nu}_{1T}^N = \Psi^{-1}(\tilde{Y}_{1T}^{N})$. In what follows, we focus on bounding the estimation error $d(\nu_{1T}^N, \hat{\nu}_{1T}^N)$.

We begin by assuming the following data-generating process.
\begin{ass}[Autoregressive model]
\label{ass_auto}
    For each unit $i=1, \ldots, N$, we assume that the post-treatment control potential outcome $Y_{iT}^N \in \mathcal{Y}$ is generated as
    \begin{equation}
        Y_{iT}^{N}(x)
        =
        \sum_{t=1}^{T_0} \langle\beta_t(x, \cdot), Y_{it}^{N}  \rangle_{\mathcal{H}}
         + \varepsilon_{iT}(x), \quad x \in \mathcal{X}.
         \label{eq:autregressive}
    \end{equation}
    Here, $\beta_t(\cdot, \cdot): \mathcal{X}^2 \to \mathbb{R}, t=1, \ldots, T_0$ are coefficient functions, and $\varepsilon_{iT} \in \mathcal{H}, i=1, \ldots, N$ are independent mean-zero errors. Furthermore, we assume that there exists a constant $\sigma > 0$ such that $\|\varepsilon_{iT}\|_{\mathcal{H}} \le \sigma$ holds almost surely for all $i=1, \ldots, N$.
\end{ass}

Under this assumption, we derive the following finite-sample error bound. 
We define the norm of each coefficient function \( \beta_t \) as
\[
\|\beta_t\|_{\mathcal{H} \times \mathcal{H}} = \sqrt{ \int_{\mathcal{X}} \int_{\mathcal{X}} \beta_t(x, y)^2 \, d\mu(x) \, d\mu(y) }.
\]

\begin{thm}
    Suppose Assumption \ref{ass_auto} holds. Then, for any $\delta > 0$, the generic  estimator $\hat{\nu}_{1T}^N$ satisfies  
    \begin{align}
        d(\nu_{1T}^N, \hat{\nu}_{1T}^N)
        &\le \sqrt{\sum_{t=1}^{T_0} \|\beta_t\|_{\mathcal{H} \times \mathcal{H}}^2} \sqrt{\sum_{t=1}^{T_0} \left\|Y_{1t} - \sum_{i=2}^N \hat{\gamma}_iY_{it}\right\|_{\mathcal{H}}^2} 
        +
        \delta \sigma(1 +\|\hat{\gamma}\|_2)
        \label{eq:est_error_auto}
    \end{align}
    with probability at least $1 - 2e^{-\delta^2/2}$. 
    \label{thm:est_error_auto}
\end{thm}

From \eqref{eq:est_error_auto}, we observe that the estimation error of the generic estimator $\hat{\nu}_{1T}^N$ is governed by the pre-treatment fit in the space $\mathcal{H}$,  $\sqrt{\sum_{t=1}^{T_0} \left\|Y_{1t} - \sum_{i=2}^N \hat{\gamma}_i Y_{it}\right\|_{\mathcal{H}}^2}$,  and the norm of the weights $\|\hat{\gamma}\|_2$.
This finding is analogous to Proposition 1 in \citet{ben2021augmented}.
Note that, in constructing the estimator $\hat{\nu}_{1T}^N = \Psi^{-1}(\tilde{Y}_{1T}^N)$, we project the weighting estimator $\hat{Y}_{1T}^N$ onto the space $\mathcal{Y}$. 
Our result accounts for the effect of this projection step on the overall estimation error.

Combining Theorem \ref{thm:est_error_auto}  with Lemma \ref{lem:aug_weight_prop}, we obtain the following bound for the ridge augmented FSC estimator. We denote this estimator by $\hat{\nu}_{1T}^{N, \text{aug}(K)}$ to make the dependence on the number of the orthonormal vectors $K$ explicit. 

\begin{cor}
    Suppose Assumptions \ref{ass:singular_values} and  \ref{ass_auto} hold. Then, for any positive integer $K$ and $\delta > 0$, the ridge augmented FSC estimator $\hat{\nu}_{1T}^{N, \mathrm{aug}}$ with regularization parameter $\lambda > 0$ satisfies 
    \begin{align}
          d(\nu_{1T}^N, \hat{\nu}_{1T}^{N, \mathrm{aug}(K)}) 
         & \le \frac{\sqrt{m(K)}\lambda}{({d_{\mathrm{min}}^{(K)})^2 + \lambda}}\sqrt{\sum_{t=1}^{T_0} \|\beta_t\|_{\mathcal{H} \times \mathcal{H}}^2} \sqrt{\sum_{t=1}^{T_0} \left\|Y_{1t} - \sum_{i=2}^N \hat{\gamma}_i^{\mathrm{scm}}Y_{it}\right\|_{\mathcal{H}}^2} \\ 
        & \phantom{\le}+ \delta \sigma \left\{ 1+
        \|\hat{\gamma}^{\mathrm{scm}}\|_2 +
        \frac{\sqrt{m(K)}d_{\mathrm{max}}^{(K)}}{(d_{\mathrm{min}}^{(K)})^2 + \lambda}\sqrt{\sum_{t=1}^{T_0} \left\|Y_{1t} - \sum_{i=2}^N \hat{\gamma}_i^{\mathrm{scm}}Y_{it}\right\|_{\mathcal{H}}^2}\right\} 
        + R_3^{(K)}
        \label{eq:est_error_auto_2}
    \end{align}
    with probability at least $1 - 2e^{-\delta^2/2}$. Here, $R_3^{(K)} \to 0$ as $K \to \infty$. 
    \label{cor:error_auto}
\end{cor}
From \eqref{eq:est_error_auto_2}, we see that the regularization parameter $\lambda$ controls a trade-off in the main part of the bound: the first term, which corresponds to the pretreatment fit, increases with $\lambda$, while the second term, which corresponds to the norm of the weights,  decreases with $\lambda$. This observation is parallel to the case of scalar outcomes (cf.  Proposition 1 in \cite{ben2021augmented}).

Let us compare the estimation error of the ridge augmented FSC estimator with that of the FSC estimator. As $\lambda \to \infty$ in \eqref{eq:est_error_auto_2}, we recover a finite-sample error bound for the FSC estimator, whose main part is given by
\begin{equation}
    \sqrt{m(K)}\sqrt{\sum_{t=1}^{T_0} \|\beta_t\|_{\mathcal{H} \times \mathcal{H}}^2} \sqrt{\sum_{t=1}^{T_0} \left\|Y_{1t} - \sum_{i=2}^N \hat{\gamma}_i^{\mathrm{scm}}Y_{it}\right\|_{\mathcal{H}}^2}
    +
    \delta \sigma( 1+  \|\hat{\gamma}^{\mathrm{scm}}\|_2).
    \label{eq:est_error_auto_fscm}
\end{equation}
If the noise level $\sigma$ satisfies
\[
\sigma < \frac{(d_{\mathrm{min}}^{(K)})^2}{\delta d_{\mathrm{max}}^{(K)}}\sqrt{\sum_{t=1}^{T_0} \|\beta_t\|_{\mathcal{H} \times \mathcal{H}}^2},
\]
then the main part of the bound in 
\eqref{eq:est_error_auto_2} is smaller than that in  \eqref{eq:est_error_auto_fscm}. This implies that, when the noise level is sufficiently low, augmentation improves the estimation accuracy.

\subsection{Error Bounds under Latent Factor Model}
\label{subsec:error_factor}
We next assume the following data-generating process.

\begin{ass}[Latent factor model]
\label{ass_latent}
    Suppose that there are $J$ unknown, latent  factors at each time $t=1, \ldots, T$, denoted by $\mu_t = (\mu_{jt})_{j=1}^J$, where each $\mu_{jt}$ is in $\mathcal{H}$.
    For each $x \in \mathcal{X}$,
    we define the vectors of pre-treatment factors, $\mu_t(x) \in \mathbb{R}^J, t=1, \ldots, T_0$, and the matrix $\mu(x) \in \mathbb{R}^{T_0 \times J}$ as 
    \[
    \mu_t(x) = \begin{pmatrix}
    \mu_{1t}(x) \\ 
    \vdots \\ 
    \mu_{Jt}(x)
    \end{pmatrix}, \quad 
    \mu(x) = 
    \begin{pmatrix}
        \mu_1(x)' \\ 
        \vdots \\ 
        \mu_{T_0}(x)'
    \end{pmatrix}.
    \]
    We assume that there exists a constant $M_1 > 0$  such that $|\mu_{jt}(x)| \le M_1$ for all $j, t, x$. Furthermore, 
    we assume that there exists a constant $M_2 > 0$ such that, for any $x$, the minimum eigenvalue of the matrix  $\mu(x)'\mu(x)$, denoted by $\xi_{\text{min}}(x)$, satisfies   $\xi_{\text{min}}(x) \ge M_2$.  
    In addition,  suppose that each unit $i$ has a vector of unknown factor loadings $\phi_i = (\phi_{ij})_{j=1}^J \in \mathbb{R}^J$.  We assume that for each unit $i = 1, \ldots, N$, the control potential outcome at time period $t=1, \ldots, T$ is generated as 
    \begin{equation}
        Y_{it}^{N}
        =
        \sum_{j=1}^{J}  \phi_{ij} \mu_{jt}
         + \varepsilon_{it},
         \label{eq:latent_factor}
    \end{equation}
    where $\varepsilon_{it} \in \mathcal{H}, i=1, \ldots, N, t=1, \ldots, T$ are independent, zero-mean errors.  We further assume that there exists a constant \( \sigma > 0 \) such that \( \| \varepsilon_{it} \|_{\mathcal{H}} \leq \sigma \) almost surely for every \( i \) and \( t \).
\end{ass}

Under this assumption, we derive the following finite-sample error bound for the generic estimator $\hat{\nu}_{1T}^N$. 
For any vector $x$, let $\|x\|_1$ denote its $\ell^1$ norm.  
\begin{thm}
    Suppose Assumption \ref{ass_latent} holds. Then, for any $\delta > 0$, the generic  estimator $\hat{\nu}_{1T}^N$ satisfies
    \begin{align}
        d(\nu_{1T}^N, \hat{\nu}_{1T}^N)
        &\le \frac{M_1^2J^{3/2}}{M_2\sqrt{T_0}} \sqrt{\sum_{t=1}^{T_0} \left\|Y_{1t} - \sum_{i=2}^N \hat{\gamma}_iY_{it}\right\|_{\mathcal{H}}^2} 
        +
        \frac{2\sigma M_1^2 J^{3/2}}{M_2}\|\hat{\gamma}\|_1
        +
        \delta \sigma(1+\|\hat{\gamma}\|_2)
        \label{eq:est_error_factor}
    \end{align}
    with probability at least $1 - 2e^{-\delta^2/2}$. 
    \label{thm:est_error_factor}
\end{thm}

From \eqref{eq:est_error_factor}, we see that the estimation error of the generic estimator $\hat{\nu}_{1T}^N$ is governed by the pre-treatment fit $\sqrt{\sum_{t=1}^{T_0} \left\|Y_{1t} - \sum_{i=2}^N \hat{\gamma}_i Y_{it}\right\|_{\mathcal{H}}^2}$ and the norms of the weights $\|\hat{\gamma}\|_1, \|\hat{\gamma}\|_2$.
This result is analogous to Theorem A.3 in \cite{ben2021augmented}.
As in the autoregressive case, our result accounts for the effect of the projection step involved in constructing the generic estimator $\hat{\nu}_{1T}^N$.

Combining Theorem \ref{thm:est_error_factor} with Lemma \ref{lem:aug_weight_prop}, we  obtain the following bound for the ridge augmented FSC estimator.

\begin{cor}
    Suppose Assumptions \ref{ass:singular_values} and  \ref{ass_auto} hold. Then,  for any positive integer $K$ and $\delta > 0$, the ridge augmented FSC estimator $\hat{\nu}_{1T}^{N, \mathrm{aug}}$ with regularization parameter $\lambda > 0$ satisfies 
\begin{align}
    &d(\nu_{1T}^N, \hat{\nu}_{1T}^{N, \mathrm{aug}(K)}) \\
    &\le \frac{M_1^2J^{3/2}\sqrt{m(K)}\lambda}{M_2\sqrt{T_0}\{(d_{\mathrm{min}}^{(K)})^2 + \lambda\}} \sqrt{\sum_{t=1}^{T_0} \left\|Y_{1t} - \sum_{i=2}^N \hat{\gamma}_i^{\mathrm{scm}}Y_{it}\right\|_{\mathcal{H}}^2} \\ 
    &\phantom{\le}+
    \sigma\left(\frac{2\sqrt{N-1}M_1^2J^{3/2}}{M_2}+\delta\right) \left\{
        \|\hat{\gamma}^{\mathrm{scm}}\|_2 +
        \frac{\sqrt{m(K)}d_{\mathrm{max}}^{(K)}}{(d_{\mathrm{min}}^{(K)})^2 + \lambda}\sqrt{\sum_{t=1}^{T_0} \left\|Y_{1t} - \sum_{i=2}^N \hat{\gamma}_i^{\mathrm{scm}}Y_{it}\right\|_{\mathcal{H}}^2}\right\} \\ 
        &\phantom{\le}+
        \delta \sigma + R_4^{(K)}
        \label{eq:est_error_factor_2}
\end{align}
with probability at least $1 - 2e^{-\delta^2/2}$.
Here, $R_4^{(K)} \to 0$ as $K \to  \infty$.
\label{cor:est_error_factor}
\end{cor}
From \eqref{eq:est_error_factor_2}, we see that,  similarly to the autoregressive model, the regularization parameter $\lambda$ controls a trade-off in the main part of the bound: the first term, which corresponds to the pretreatment fit, increases with $\lambda$, while the second term, which corresponds to the norms of the weights,  decreases with $\lambda$.

Finally, let us compare the estimation error of the ridge augmented FSC estimator with that of the FSC estimator. As $\lambda \to \infty$ in \eqref{eq:est_error_factor_2}, we recover a finite-sample error bound for the FSC estimator, whose main part is given by
\begin{equation}
    \frac{M_1^2J^{3/2}\sqrt{m(K)}}{M_2\sqrt{T_0}} \sqrt{\sum_{t=1}^{T_0} \left\|Y_{1t} - \sum_{i=2}^N \hat{\gamma}_i^{\mathrm{scm}}Y_{it}\right\|_{\mathcal{H}}^2}
    +
    \sigma\left(\frac{2\sqrt{N-1}M_1^2J^{3/2}}{M_2}+\delta\right)\|\hat{\gamma}^{\text{scm}}\|_2 + \delta \sigma.
    \label{eq:est_error_factor_fscm}
\end{equation}
If the noise level $\sigma$ satisfies
\[
\sigma < \frac{(d_{\text{min}}^{(K)})^2M_1^2J^{3/2}}{d_{\text{max}}^{(K)}\sqrt{T_0}(2\sqrt{N-1}M_1^2J^{3/2} + M_2\delta)}, 
\]
then the main part of the bound in \eqref{eq:est_error_factor_2} is smaller than that in \eqref{eq:est_error_factor_fscm}. This implies that, similarly to the autoregressive model, when the noise level is sufficiently low, augmentation improves the estimation accuracy.

\begin{rmk}[Selection of regularization parameter]
\label{rmk:regulatization}
In practice, following \cite{ben2021augmented}, we adopt a cross-validation approach for selecting the regularization hyperparameter $\lambda$. Given $\lambda > 0$, for each $t = 1, \ldots, T_0$, let $\hat{Y}_{1t}^{(-t)} = \sum_{i=2}^N \hat{\gamma}_{i(-t)}^{\text{aug}} Y_{it}$ denote the estimate of $Y_{1t}$, where $\hat{\gamma}_{i(-t)}^{\text{aug}}$ denotes the augmented weights calculated without using the observations from time period $t$.
We then compute the leave-one-out cross-validation (CV) error over the pre-treatment periods:
    \begin{equation}
        \text{CV}(\lambda) = \sum_{t=1}^{T_0} \|Y_{1t} - \hat{Y}_{1t}^{(-t)}\|_{\mathcal{H}}^2.
    \end{equation}
    The hyperparameter $\lambda$ is selected by minimizing the cross-validation error:
    \[
        \hat{\lambda}_{\text{cv}} = \argmin_{\lambda > 0} \text{CV}(\lambda).
    \]
\end{rmk}

\subsection{Inference}
\label{subsec:inference}
In this subsection, we extend inference procedures for the standard SCM to our setting. Specifically, we focus on two procedures:
(i) construction of prediction sets for counterfactural outcomes based on conformal inference approach \citep{chernozhukov2021exact}, and
(ii) the placebo permutation test for the presence of a causal effect \citep{abadie2010synthetic}. 
We illustrate how these procedures can be implemented in our framework.

\paragraph{Conformal inference for counterfactual outcomes.}
Fix a post-treatment period $t$.
Let $\hat{\gamma} = (\hat{\gamma}_i)_{i=2}^N$ denote the FSC weights in \eqref{eq:optim_fscm} or the augmented weights in \eqref{eq:weight_aug}. Our goal is to construct, for each $x \in \mathcal{X}$, a $(1 - \alpha)$ prediction interval $\hat{C}_{1-\alpha}(x)$ for the counterfactual outcome $Y_{1t}^N(x)$ based on the estimate $\hat{Y}_{1t}^N = \sum_{i=2}^N \hat{\gamma}_i Y_{it}$.
Here, $\alpha \in (0, 1)$ is a pre-specified level assumed to satisfy $\alpha > 1/(T_0 + 1)$.

To this end, consider the sharp null hypothesis $H_0: Y_{1t}^N(x) = y_0$ for a given $y_0 \in \mathbb{R}$. A $p$-value $p(y_0)$ for the hypothesis $H_0$ can be computed by comparing the post-treatment residual $y_0 - \sum_{i=2}^N\hat{\gamma}_iY_{it}(x)$ with the pre-treatment residuals $Y_{1s}(x) - \sum_{i=2}^N \hat{\gamma}_i Y_{is}(x), s=1, \ldots, T_0$:
\begin{align}
    p(y_0) &= 
    \frac{1}{T_0+1} \sum_{s=1}^{T_0} 
     1\!\left\{
         \left|y_0 - \sum_{i=2}^N \hat{\gamma}_i Y_{it}(x)\right| 
         \le 
         \left|Y_{1s}^N(x) - \sum_{i=2}^N \hat{\gamma}_i Y_{is}(x)\right|
     \right\} \\ 
     &\phantom{=} +
     \frac{1}{T_0+1} 1\!\left\{
         \left|y_0 - \sum_{i=2}^N \hat{\gamma}_i Y_{it}(x)\right| 
         \le 
         \left|Y_{1t}^N(x) - \sum_{i=2}^N \hat{\gamma}_i Y_{it}(x)\right|\right\} \\ 
         &=
         \frac{1}{T_0+1} \sum_{s=1}^{T_0} 
     1\!\left\{
         \left|y_0 - \sum_{i=2}^N \hat{\gamma}_i Y_{it}(x)\right| 
         \le 
         \left|Y_{1s}(x) - \sum_{i=2}^N \hat{\gamma}_i Y_{is}(x)\right|
     \right\}
     +
     \frac{1}{T_0+1}.
\end{align}
Here, $1\{\cdot\}$ denotes the indicator function. We invert this test to construct the  prediction set $\hat{C}_{1-\alpha}(x)$ for $Y_{1t}^N(x)$: 
\[
\hat{C}_{1 - \alpha}(x)
    =
    \{y \in \mathbb{R}: p(y) \ge \alpha\}.
\]
For any $y \in \mathbb{R}$, 
\begin{align*}
    p(y) \ge \alpha
    &\iff 
    \frac{1}{T_0}\sum_{s=1}^{T_0} 
     1\!\left\{
         \left|y - \sum_{i=2}^N \hat{\gamma}_i Y_{it}(x)\right| 
         \le 
         \left|Y_{1s}(x) - \sum_{i=2}^N \hat{\gamma}_i Y_{is}(x)\right|
     \right\} \ge \frac{(T_0+1)\alpha-1}{T_0} \\ 
     & \iff 
     \left|y - \sum_{i=2}^N\hat{\gamma}_i Y_{it}(x)\right| \le q_{\alpha}(x),
\end{align*}
where $q_{\alpha}(x)$ is the $[1 - \{(T_0+1)\alpha - 1\}/{T_0}]$-quantile of the pre-treatment residuals $|Y_{1s}(x) - \sum_{i=2}^N \hat{\gamma}_i Y_{is}(x)|, s=1, \ldots, T_0$.
Hence, the set $\hat{C}_{1-\alpha}(x)$ is an  interval centered at the estimate $\hat{Y}_{1t}^N(x) =  \sum_{i=2}^N \hat{\gamma}_i Y_{it}(x)$:
\[
\hat{C}_{1 - \alpha}(x)
=
\left[\hat{Y}_{1t}^N(x) -q_{\alpha}(x), \hat{Y}_{1t}^N(x) +q_{\alpha}(x)\right].
\]

From the above procedures, we obtain a pointwise prediction band $\hat{C}_{1 - \alpha}(\cdot)$ for the counterfactual outcome $Y_{1t}^N(\cdot)$. A corresponding prediction band for the causal effect $Y_{1t}^I - Y_{1t}^N$ is then given by $Y_{1t}(\cdot) - \hat{C}_{1-\alpha}(\cdot)$. 
In order to guarantee that the band $\hat{C}_{1 - \alpha}(\cdot)$ contain the estimate $\hat{Y}_{1t}^N(\cdot)$, we use the same weights $\hat{\gamma} = (\hat{\gamma}_i)_{i=2}^N$ to construct the residuals under any sharp null hypothesis. 
This contrasts with the procedures in \cite{chernozhukov2021exact} and \cite{ben2021augmented} for scalar outcomes, where the weights may vary depending on the sharp null hypothesis being considered.

For the case of scalar outcomes, \cite{chernozhukov2021exact} and \cite{ben2021augmented} show that, under several conditions, prediction sets based on the conformal inference approach for the SCM and the augmented SCM are asymptotically valid as the number of pre-treatment periods $T_0$ goes to infinity.
We expect that our procedure also enjoys asymptotic validity, although a rigorous justification is left for future work.

\paragraph{Placebo permutation test.}
Fix a post-treatment time period $t$.
Let $\hat{Y}_{1t}^N$ be the FSC  or augmented FSC estimate for $Y_{1t}^N$. In addition to this estimate, we compute 
$\hat{Y}_{it}^{N}$ for $i=2, \ldots, N$ by applying the same algorithm as used to obtain $\hat{Y}_{1t}^N$, pretending that unit $i$ is the treatment unit. Then the $p$-value for testing the null hypothesis of no causal effect, $H_0: Y_{1t}^I- Y_{1t}^N = 0$, or equivalently, $H_0: d(\nu_{1t}^I, \nu_{1t}^N) =  0$, is given by 
\begin{align}
    &\frac{1}{N}\sum_{i=1}^N 1\left\{\|Y_{1t} - \hat{Y}_{1t}^{N}\|_{\mathcal{H}} \le \|Y_{it} - \hat{Y}_{it}^{N}\|_{\mathcal{H}}\right\}  
    =
    \frac{1}{N}\sum_{i=2}^N 1\left\{\|Y_{1t} - \hat{Y}_{1t}^{N}\|_{\mathcal{H}} \le \|Y_{it} - \hat{Y}_{it}^{N}\|_{\mathcal{H}}\right\}
    +
    \frac{1}{N}.
\end{align}
Analogous procedure is proposed by \cite{gunsilius2023distributional} for the case of distributional outcomes, and by \cite{kurisu2025geodesic} for the case of outcomes in a geodesic metric space. 

\section{Simulations}
\label{sec:simulations}
In this section, we conduct simulation studies to evaluate the finite-sample performance of the proposed and related methods under both autoregressive and latent factor models.
In these simulations, we assume that the outcomes take  values in the $L^2$ space $L^2([0, 1])$ (see Example \ref{exm:function}), and set $N=50, T=10$ and $T_0=9$. 

For the autoregressive model, we first generate the pre-treatment outcomes $Y_{it}, i=1, \ldots, 50, t=1, \ldots, 9$ as 
\begin{equation}
    Y_{it}(x) =  \sum_{\ell = 1}^{10}\ell^{-1.2}U_{\ell t} g_{\ell}(x), \quad x \in [0, 1],
\end{equation}
where $U_{\ell t}$'s are independent random variables uniformly distributed on $[-\sqrt{3}/100, \sqrt{3}/100]$, and 
$g_{1}(x) = 1, g_{\ell}(x) = \sqrt{2}\cos((\ell-1)\pi x)$ for $\ell \ge 2$. The generated pre-treatment outcomes are fixed throughout the remaining simulations for the autoregressive model. 
We then generate post-treatment control outcomes $Y_{i, 10}^N, i=1, \ldots, 50$ as 
\begin{equation}
    Y_{i, 10}^N(x) = \langle \beta_1(x, \cdot), Y_{i9} \rangle_{L^2} + \langle \beta_2(x, \cdot), Y_{i8} \rangle_{L^2} + \langle \beta_3(x, \cdot), Y_{i7} \rangle_{L^2} +
    \varepsilon_i(x), \quad x \in [0, 1],
\end{equation}
where $\beta_1(x, y) = 0.6f(y | x, 0.1), \beta_2(x, y) = 0.3f(y | x, 0.1), \beta_3(x, y) = 0.1f(y | x, 0.1)$, with $f(\cdot | \mu, \sigma)$ denoting the density function of the normal distribution with mean $\mu$ and standard deviation $\sigma$. The error term $\varepsilon_i$ is generated as 
$
\varepsilon_i(x) = \varepsilon_{1i} + \varepsilon_{2i}x^{1/2} + \varepsilon_{3i} x^{1/3} + \varepsilon_{4i} x^{1/4},
$
where $\varepsilon_{i1}, \varepsilon_{i2}, \varepsilon_{i3}, \varepsilon_{i4}$ are independent random variables uniformly distributed on $[-C, C]$. We consider three cases for the value of $C$: $C = 0.05$ (low noise level), $C = 0.2$ (medium noise level), and $C = 1$ (high noise level).

For the latent factor model, we set the number of factors to $J=5$. The factors $\mu_{jt}$, $j = 1, \ldots, 5, t = 1, \ldots, 10$ are specified as $\mu_{11}(x) = 1$ and $\mu_{jt}(x) = \sqrt{2}\cos((j + t)\pi x)$ for $(j, t) \neq (1, 1)$.
The factor loadings $\phi_{ij}, i = 1, \ldots, 50, j = 1, \ldots, 5$ are generated from the normal distribution with mean $0$ and standard deviation $0.1$. The control potential outcomes $Y_{it}^N, i = 1, \ldots, 50, t = 1, \ldots, 10$ are then generated as
\[
Y_{it}^N(x) = \sum_{j=1}^5 \phi_{ij} \mu_{jt}(x) + \varepsilon_{it}(x),
\]
where the error term $\varepsilon_{it}$ is generated as
$
\varepsilon_{it}(x) = \varepsilon_{1it} + \varepsilon_{2it}x^{1/2} + \varepsilon_{3it} x^{1/3} + \varepsilon_{4it} x^{1/4},
$
and  $\varepsilon_{i1t}, \varepsilon_{i2t}, \varepsilon_{i3t}, \varepsilon_{i4t}$ are independent random variables uniformly distributed on $[-C, C]$. We consider three cases for the value of $C$: $C = 0.02$ (low noise level), $C = 0.1$ (medium noise level), and $C = 0.5$ (high noise level).

For each model, we estimate the counterfactual outcome $Y_{1, 10}^N$ using the following  methods: the FSC estimator; the ridge augmented FSC estimators with penalty hyperparameters $\lambda = 100 \hat{\lambda}_{\text{cv}}, \hat{\lambda}_{\text{cv}}, 0.01\hat{\lambda}_{\text{cv}}$, where $\hat{\lambda}_{\text{cv}}$ denotes the optimal value of $\lambda$ selected by the cross validation (see Remark \ref{rmk:regulatization}); and the augmented GSC estimator proposed by \cite{kurisu2025geodesic}.
In our setting, the augmented GSC estimator is defined as
\[
\hat{Y}_{1, {10}}^{N, \text{agsc}}
=
\hat{Y}_{1, 10}^{N, \text{scm}} + \hat{m}_{1, 10} - \sum_{i=2}^{50} \hat{\gamma}_i^{\text{scm}}\hat{m}_{i, 10},
\]
where $\hat{m}_{i, 10}$ denotes a regression estimator of $Y_{i, 10}^N$ which can be implemented by using the geodesic optimal transport (GOT) regression \citep{zhu2025geodesic}. Since the GOT regression reduces to the standard linear regression in our setting, we have 
$
\hat{m}_{i, 10} = \sum_{t=1}^{9} \hat{\alpha}_t Y_{it}, i=1, \ldots, 50,
$
where the coefficients $\hat{\alpha}_1, \ldots, \hat{\alpha}_9 \in \mathbb{R}$ are obtained by solving
\[
\min_{\alpha_1, \ldots \alpha_9 \in \mathbb{R}} \sum_{i=2}^{50}\left\|Y_{i, 10} - \sum_{t=1}^9 \alpha_t Y_{it}\right\|_{L^2}.
\]
For the choice of the basis functions $\{\varphi_k\}_{k=1}^K$ in the ridge augmented FSC, we use cubic B-splines and set $K=50$.
The number of Monte Carlo repetitions for each setup is set to $500$. The performance of the estimates $\hat{Y}_{1, 10}^N$ is evaluated using the $L^2$-distance $\|Y_{1, 10}^N - \hat{Y}_{1, 10}^N\|_{L^2}$. To compute the integrals over $[0,1]$, we approximate them by finite summations over $100$ grid points: $0.01, 0.02, \ldots, 0.99$.

\begin{figure}[t]
    \centering
    \includegraphics[width=0.9\linewidth]{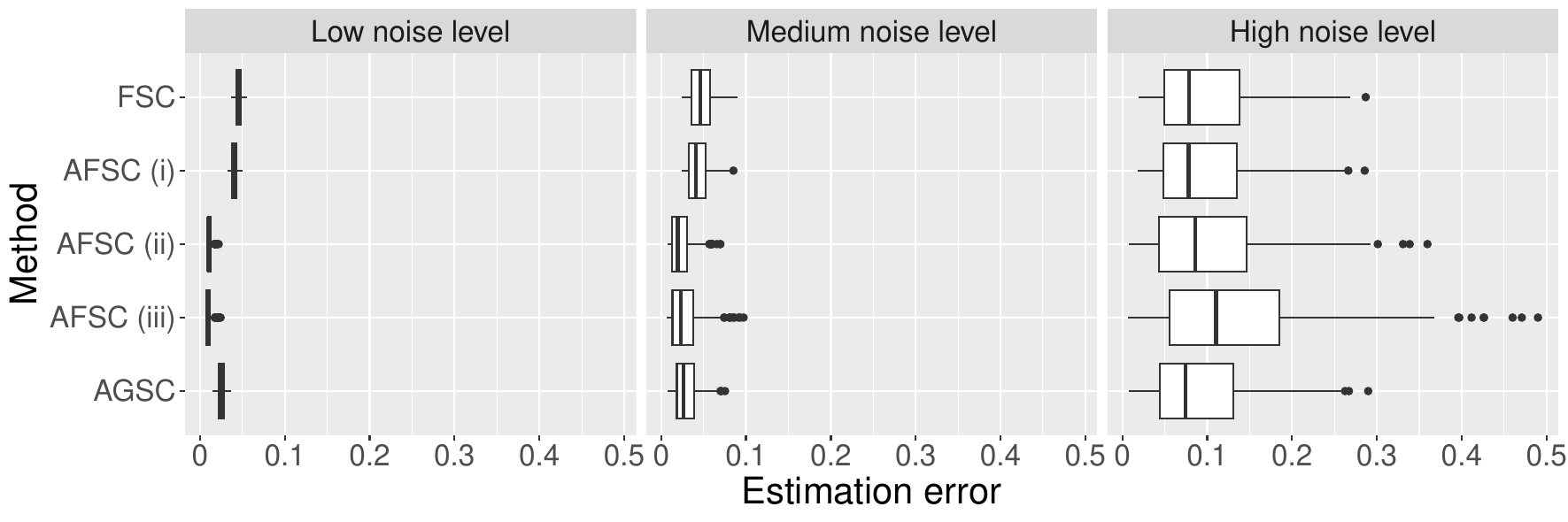}
    \caption{Boxplots of estimation errors under the autoregressive model with low noise (left), medium noise (middle), and high noise (right). AFSC (i), AFSC (ii), and AFSC (iii) denote the ridge augmented FSC with penalty hyperparameters $\lambda = 100\hat{\lambda}_{\text{cv}},\ \hat{\lambda}_{\text{cv}},\ \text{and}\ 0.01\hat{\lambda}_{\text{cv}}$, where $\hat{\lambda}_{\text{cv}}$ is the value of $\lambda$ selected by cross validation.}
    \label{fig:simulation_auto}
\end{figure}

\begin{figure}[t]
    \centering
    \includegraphics[width=0.9\linewidth]{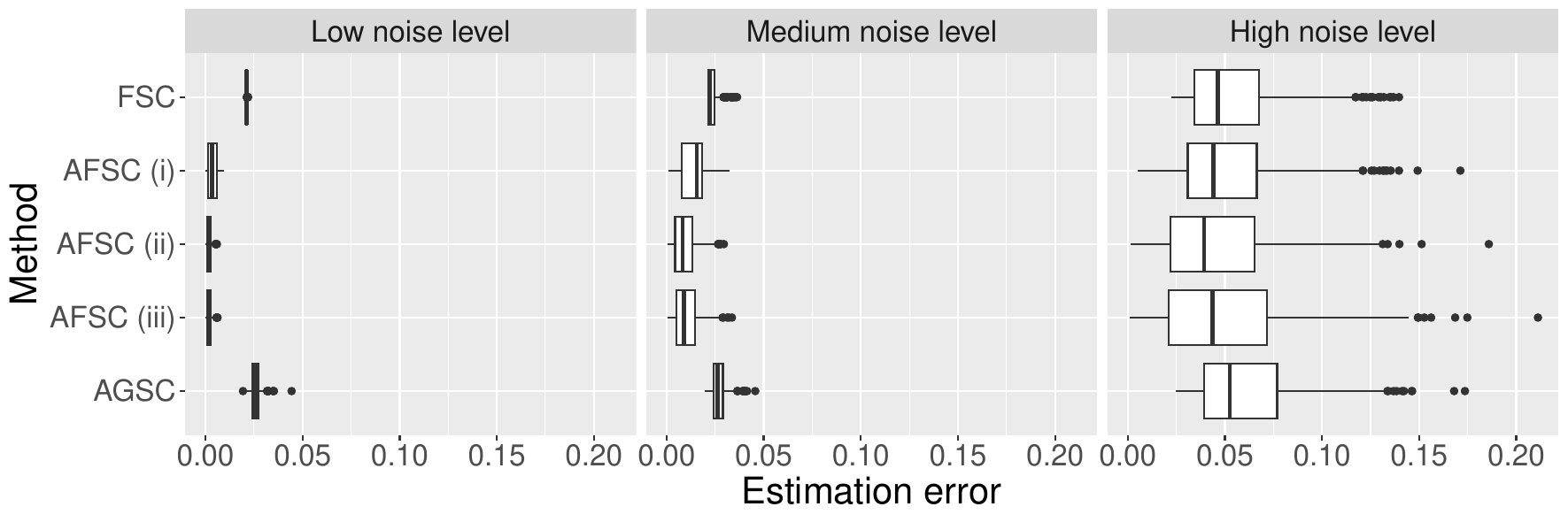}
    \caption{Boxplots of estimation errors under the latent factor model with low noise (left), medium noise (middle), and high noise (right). }
    \label{fig:simulation_factor}
\end{figure}

Figures \ref{fig:simulation_auto} and \ref{fig:simulation_factor} show the results for the autoregressive and latent factor models, respectively. We first compare the performance of the FSC estimator, the augmented FSC estimator with $\lambda = \hat{\lambda}_{\text{cv}}$, and the augmented GSC estimator. As implied by the theoretical results in Section \ref{sec:est_error}, for both models the augmented FSC estimator outperforms the FSC estimator when the noise level is not high. In such cases, the augmented FSC estimator also performs better than the augmented GSC estimator. When the noise level is high, the three estimators exhibit relatively comparable performance.

To examine the dependence of the augmented FSC estimator on its hyperparameter $\lambda$, we next compare the performance of the augmented FSC estimators with $\lambda = 100\hat{\lambda}_{\text{cv}}, \hat{\lambda}_{\text{cv}}$ and $0.01\hat{\lambda}_{\text{cv}}$. 
For both models, the estimator with $\lambda = 0.01\hat{\lambda}_{\text{cv}}$ performs better when the noise level is low, whereas it performs relatively worse when the noise level is high. This observation is also consistent with the following theoretical implications.  When the noise level $\sigma$ is small, the first terms of the error bounds in \eqref{eq:est_error_auto_2} and \eqref{eq:est_error_factor_2} dominate, so small values of $\lambda$ lead to smaller error bounds. In contrast, when the noise level $\sigma$ is high, the second terms of the bounds dominate, so small values of $\lambda$ lead to larger error bounds.

\section{Empirical Illustrations}
\label{sec:empirical}
\subsection{Analysis of the Impact of Abortion Legislation on Fertility Patterns in East Germany}
\label{susec:abortion}
Using the proposed methods, we first reanalyze the impact of the 1972 abortion legislation in East Germany on age-specific fertility rates (ASFRs), which was previously examined by \cite{kurisu2025geodesic}.
We utilize data from the Human Fertility Database (\url{https://www.humanfertility.org}),
which provides annual ASFRs for ages 12 to 55 across various countries.
When regarded as functions over age, the ASFRs constitute functional data residing in the $L^2$ space $L^2([12, 55])$ (see Example \ref{exm:function}). 

In March 1972, East Germany enacted legislation permitting the termination of pregnancies within the first twelve weeks, thereby significantly liberalizing access to abortion. 
We aim to investigate the causal effect of this legislation on the ASFR curves of East Germany. 
We designate East Germany as the treated unit and, subject to data availability, construct a control group consisting of 20 countries that did not experience comparable liberalization of abortion laws (Austria, Belgium, Bulgaria,  Canada, Switzerland, Czechia, Denmark, Spain, Finland, France, Hungary, Ireland, Italy, Japan, the Netherlands, Portugal, Slovakia, Sweden, the United Kingdom (England and Wales), and the United States). 
We define 1956--1971 as the  pre-treatment periods and 1972--1975 as the post-treatment periods. 
Figure~\ref{fig:data_asfr} in the Appendix presents the ASFR curves for East Germany and the control countries from 1956 to 1975.
In this application, we define the differences of the functional outcomes $Y_{1t}^I - Y_{1t}^N, t=1972, 1973, 1974, 1975$ as the causal effects of interest. 
The proposed approach was implemented to obtain the FSC and ridge augmented FSC estimators.
For the choice of basis functions $\{\varphi_{k}\}_{k=1}^K$ in the ridge augmented FSC, we employed cubic B-splines and set $K = 50$. 

Figure~\ref{fig:synthetic_outcomes_asfr} compares the observed ASFR curves of East Germany with the corresponding ASFR curves obtained from the FSC and augmented FSC units. 
To save space, we report results only for the periods from 1964 to 1975; see Figure~\ref{fig:synthetic_asfr_until1963} in the Appendix for the results from 1956 to 1963.
We find that the pre-treatment fit achieved by the FSC method is relatively good, and consequently, the augmentation introduces only minor adjustments in this application. 
The pre-treatment fits $\sqrt{\sum_{t=1}^{T_0} \|Y_{1t} - \sum_{i=2}^N \hat{\gamma}_i Y_{it}\|_{\mathcal{H}}^2}$ by the FSC and augmented FSC are $0.1259$ and $0.0687$, respectively, indicating that the augmentation improves the pre-treatment fit by approximately $45.4\%$.
Table \ref{table:asfr} presents the weights of the control units. The weights from the FSC estimator are sparse, with only Austria, Bulgaria, Switzerland, and Czechia receiving positive weights. 
Although the augmented weights can take negative values, they are broadly similar to the FSC weights.

Figure \ref{fig:prediction_asfr} shows the estimates and 90$\%$ prediction bands for the causal effects ${Y}_{1t}^I - {Y}_{1t}^N, t=1972, 1973, 1974, 1975$, based on the augmented FSC. The results suggest that the legislation led to notable downward shifts in fertility rates across the prime childbearing ages (approximately 20–30 years old), which is consistent with the results reported by \cite{kurisu2025geodesic}.
To assess the statistical significance of these effects, we implemented the placebo permutation test using the augmented FSC. The resulting $p$-values for the null hypothesis of no causal effect in each post-treatment year (1972--1975) are  $0.095, 0.048, 0.048, 0.095$, respectively. 
These findings indicate that the effects are statistically significant at the $0.10$ significance level. 
See Figure~\ref{fig:placebo_asfr} in the Appendix for detailed results from the permutation test.

\begin{figure}[t]
    \centering
    \includegraphics[width=\linewidth]{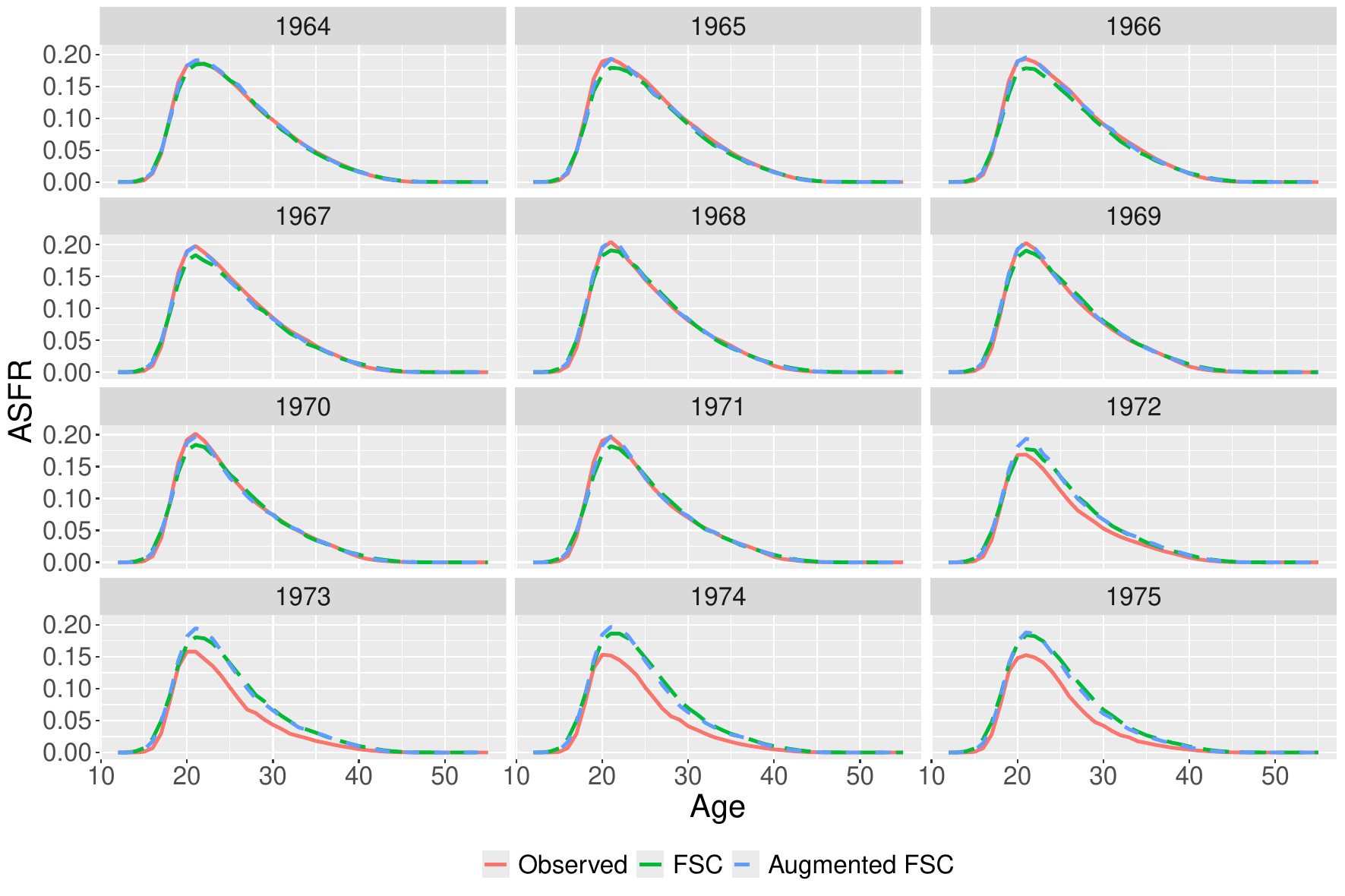}
    \caption{Observed ASFR curves for East Germany and the corresponding ASFR curves for the FSC and ridge-augmented FSC units during pre-treatment periods (1964–1971) and the post-treatment periods (1972–1975).}
    \label{fig:synthetic_outcomes_asfr}
\end{figure}

\begin{figure}[t]
    \centering
    \includegraphics[width=0.8\linewidth]{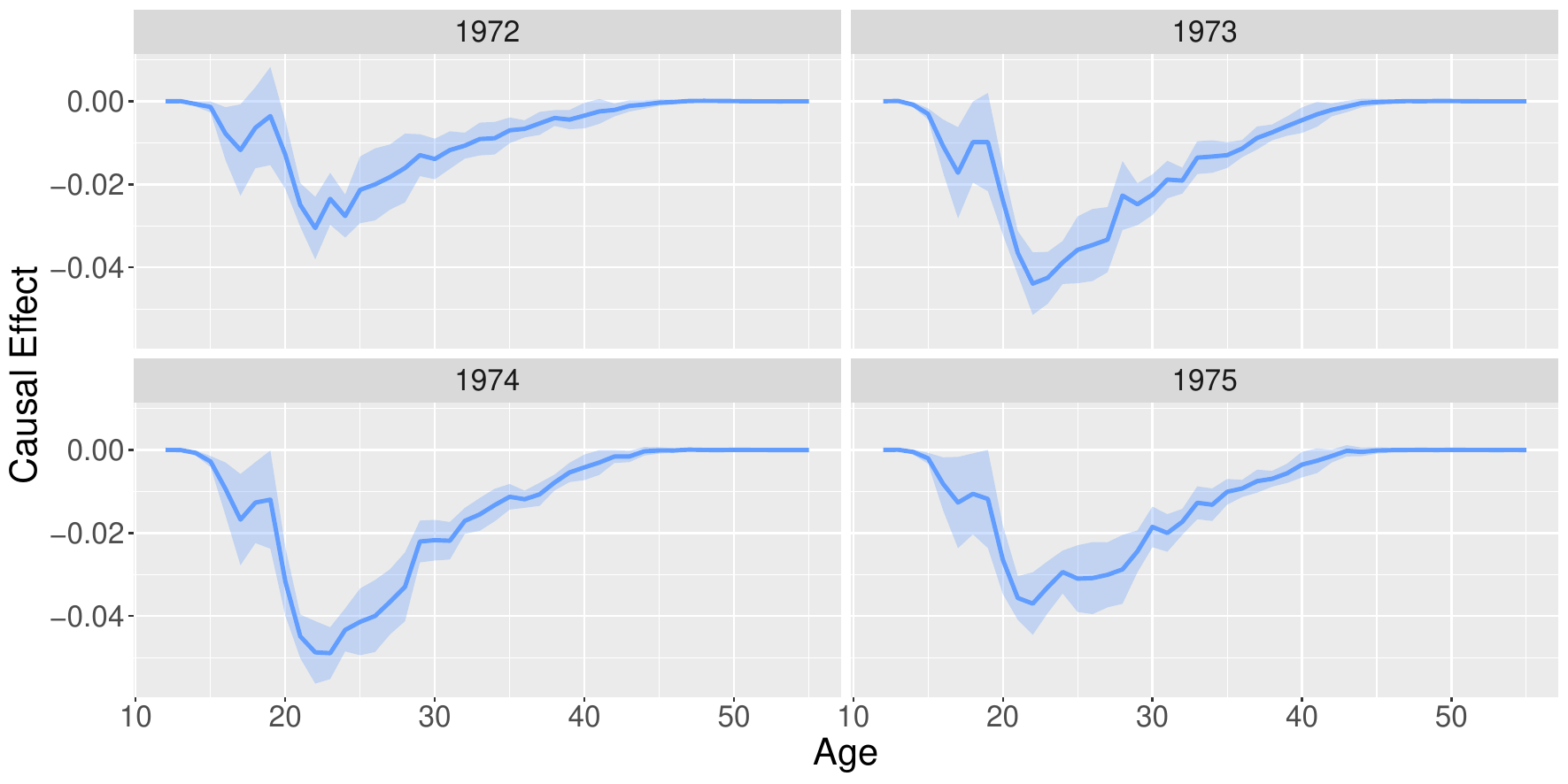}
    \caption{Estimates and pointwise 90\% prediction bands for the causal effects
${Y}_{1t}^I - {Y}_{1t}^N, t = 1972, 1973, 1974, 1975$
based on the ridge augmented FSC.}
    \label{fig:prediction_asfr}
\end{figure}

\begin{table}[t]
\centering
\begin{tabular}{l r r @{\hspace{2cm}} l r r}
\hline
Control country & FSC   & AFSC   & Control country & FSC  & AFSC  \\
\hline \hline
Austria & 0.396  & 0.526 & Hungary & 0.000 & -0.041 \\
Belgium & 0.000  & 0.093 & Ireland & 0.000 & 0.075 \\
Bulgaria & 0.416 & 0.382 & Italy & 0.000 &  -0.144 \\
Canada & 0.000   & -0.090 & Japan &0.000 & -0.063 \\
Switzerland & 0.089 & -0.026 & Netherlands & 0.000 & 0.033 \\
Czechia & 0.188 & 0.158 & Portugal & 0.000 & -0.119 \\
Denmark & 0.000 & 0.036 & Slovakia & 0.000 & 0.116 \\
Spain & 0.000 & -0.059 & Sweden & 0.000 & 0.077 \\
Finland & 0.000 & 0.004 & United Kingdom & 0.000 & 0.088  \\
France & 0.000 & 0.099 & United States & 0.000 & -0.144 \\
\hline
\end{tabular}
\caption{Control unit weights for the ASFR data obtained using the FSC and ridge augmented FSC (AFSC).}
\label{table:asfr}
\end{table}

\subsection{Analysis of the Effect of the Collapse of the Soviet Union on Mortality}
\label{subsec:mortality}
Next, we reanalyze the effect of the 1991 collapse of the Soviet Union on its age-at-death distributions, which was previously analyzed by \cite{kurisu2025geodesic}. The Human Mortality Database (\url{https://www.mortality.org}) provides annual life tables, i.e., histograms of death counts by age, for various countries. From these histograms, we derive the age-at-death distributions for each country.
Specifically, for each country, the density functions of the age-at-death distributions are obtained by smoothing the histograms. For this processing step, we use the \texttt{CreateDensity} function in the \texttt{frechet} package \citep{chen2023frechet}.
We treat these distributions as elements of the Wasserstein space (see Example \ref{exm:distribution}).

The collapse of the Soviet Union in 1991 had a profound impact on population dynamics, with life expectancy in Russia experiencing dramatic declines. 
We aim to study the causal effect of the collapse on mortality patterns in Russia by applying the proposed methods, using age-at-death distributions as the outcomes. 
Russia is selected as the treated unit, while the control group comprises 17 Western European countries: Austria, Belgium, Denmark, Finland, France, Germany, Iceland, Ireland, Italy, Luxembourg, the Netherlands, Norway, Portugal, Spain, Sweden, Switzerland, and the United Kingdom. 
We define 1970 to 1990 as the pre-treatment periods and 1991 to 1999 as the post-treatment periods. 
Figures \ref{fig:data_aad_quantile} and \ref{fig:data_aad_density} in the Appendix show the quantile and density functions of the age-at-death distributions for Russia and the control countries from 1970 to 1999.
In this application, we define the differences between the quantile functions $Y_{1t}^I - Y_{1t}^N, t=1991, 1992, \ldots, 1999$ as the causal effects of interest. 
The proposed approach was implemented to obtain the FSC and ridge augmented FSC estimators. For the choice of basis functions $\{\varphi_k\}_{k=1}^K$ in the ridge augmented FSC, we employed cubic B-splines and set $K = 50$. 
For the modification of the augmented estimator $\hat{Y}_{1T}^{N, \text{aug}}$, we employed the rearrangement method (see Remark \ref{rmk:rearrangement}).

Figure~\ref{fig:synthtic_outcome_aad_quantile} compares the quantile functions of the observed age-at-death distributions for Russia with the corresponding quantile functions constructed using the FSC and augmented FSC methods.
To conserve space, we report results only for the period from 1985 to 1999; results for 1970–1984 are provided in Figure~\ref{fig:synthetic_aad_until1984} in the Appendix. See also Figure~\ref{fig:synthtic_outcome_aad_density} in the Appendix for the corresponding comparison based on density functions.
The pre-treatment fit achieved by the FSC is relatively poor, and as a result, the augmentation leads to substantial adjustments in this application. Specifically, the pre-treatment fit of the FSC is 0.2092, whereas that of the augmented FSC is 0.0634, representing an improvement of approximately $69.7\%$ due to augmentation.
Table \ref{table:aad} presents the weights of the control units. The weights obtained using the FSC assign a total weight of 1 to Portugal. The augmented weights can be negative, and these weights differ substantially from those obtained using the FSC.

Figure~\ref{fig:prediction_aad_1} presents the estimates and $90\%$ prediction bands for the causal effects $Y_{1t}^I - {Y}_{1t}^N$ for $t = 1991, 1992, \ldots, 1999$ based on the augmented FSC. The results indicate that the collapse of the Soviet Union led to pronounced leftward shifts in the age-at-death distributions in the years following 1992, reflecting sharply increased mortality risks across nearly all age groups in Russia. These findings are consistent with those reported by \cite{kurisu2025geodesic}.
To evaluate the statistical significance of these effects, we conducted a placebo permutation test based on the augmented FSC. The resulting $p$-values for the null hypothesis of no causal effect in each post-treatment year (1991–1999) are 0.389, 0.056, 0.056, 0.056, 0.056, 0.056, 0.056, 0.056, and 0.056, respectively. These results suggest that the immediate effect in 1991 is not statistically significant, whereas the increases in mortality become statistically significant in the years after 1992.
Detailed results of the permutation test are provided in Figure~\ref{fig:placebo_aad} in the Appendix.

\begin{figure}[t]
    \centering
    \includegraphics[width=\linewidth]{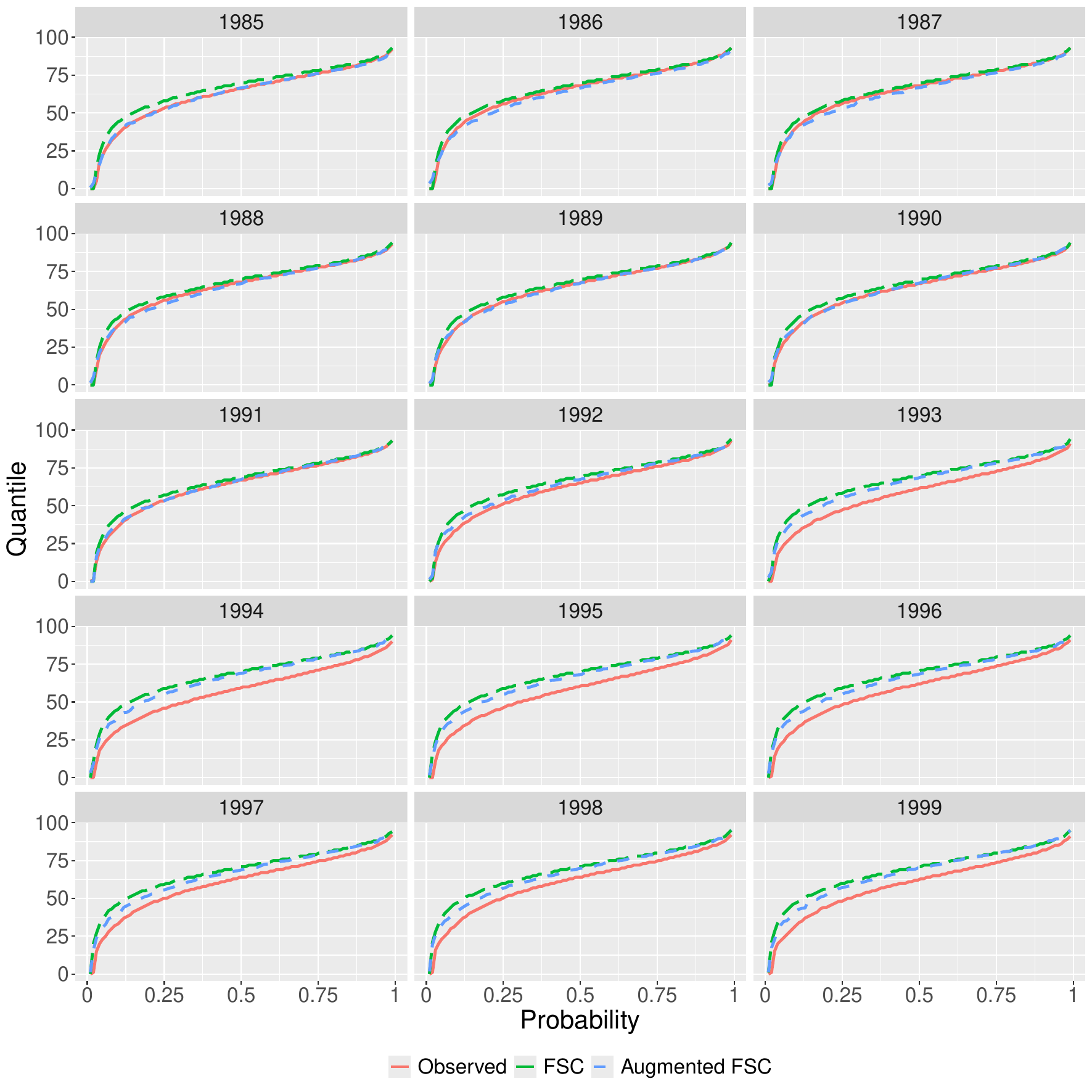}
    \caption{The quantile functions of the observed age-at-death distributions for Russia and the corresponding quantile functions obtained from the FSC and ridge augmented FSC units. The priods from 1985 to 1990 are  pre-treatment priods, while 1991 to 1999 are the post-treatment periods.}
    \label{fig:synthtic_outcome_aad_quantile}
\end{figure}

\begin{table}[t]
\centering
\begin{tabular}{l r r @{\hspace{2cm}} l r r}
\hline
Control country & FSC   & AFSC   & Control country & FSC  & AFSC  \\
\hline \hline
Austria & 0.000  & 0.046 & Luxembourg & 0.000 & 0.071 \\
Belgium & 0.000  & -0.116 & Netherlands & 0.000 & -0.334 \\
Denmark & 0.000 & 0.282 & Norway & 0.000 &  -0.448 \\
Finland & 0.000 & 0.461 & Portugal & 1.000 & 0.898 \\
France & 0.000 & 1.258 & Spain & 0.000 & -0.943 \\
Germany & 0.000 & 0.034 & Sweden & 0.000 & -0.076 \\
Iceland & 0.000 & 0.147 & Switzerland & 0.000 & -0.380 \\
Ireland & 0.000 & 0.134 & United Kingdom & 0.000 & 0.034  \\
Italy & 0.000 & -0.068 &  &  &   \\
\hline
\end{tabular}
\caption{Control unit weights for the mortality data obtained using the FSC and ridge augmented FSC (AFSC).}
\label{table:aad}
\end{table}

\begin{figure}[t]
    \centering
    \includegraphics[width=\linewidth]{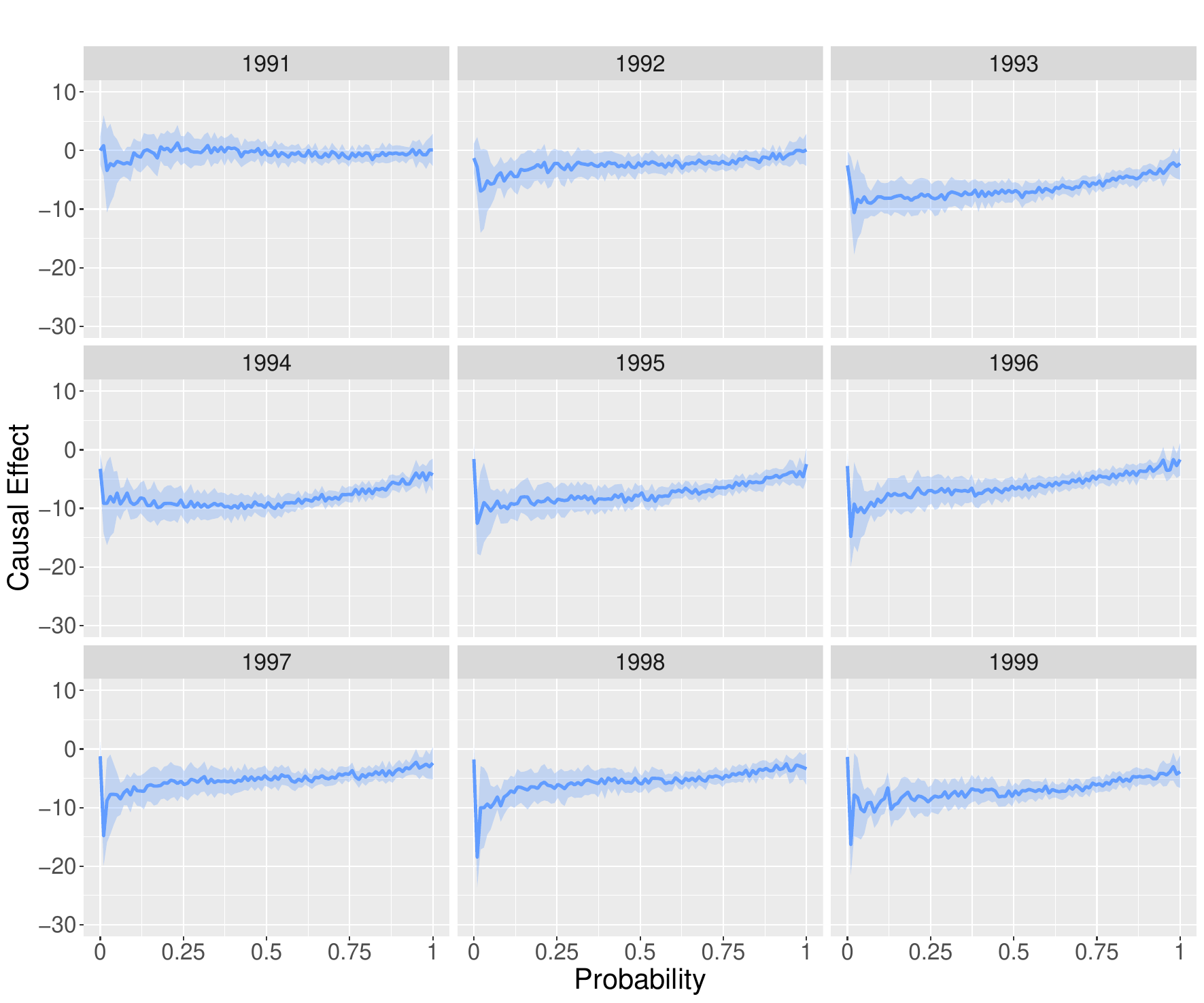}
    \caption{Estimates and pointwise $90\%$ prediction bands for the causal effects ${Y}_{1t}^I - {Y}_{1t}^N, t=1991, 1992, \ldots , 1999$ based on the ridge augmented FSC.}
    \label{fig:prediction_aad_1}
\end{figure}

\subsection{Analysis of the Impact of the Brexit Announcement on U.K. Services Trade}
\label{subsec:service}
Finally, we analyze the impact of the Brexit announcement on the United Kingdom (UK) services trade. 
On June 23, 2016, the UK narrowly voted to leave the European Union.
This decision may substantially affect the UK’s ability to maintain frictionless trade with its largest trading partner.
We assess the impact of this decision on UK services trade using the proposed framework, with covariance matrices as outcomes.

UN Trade and Development (\url{https://unctad.org}) provides quarterly trade data for various services and countries.
Considering data availability, we focus on the following nine service categories: SC (transport), SD (travel), SE (construction), SG (financial services), SH (intellectual property charges), SI (telecommunications, computer and information services), SJ (other business services), SK (personal, cultural and recreational services), and SL (government goods and services).
We designate the UK as the treated unit and construct a control group consisting of the following 22 countries: Australia, Austria, Canada, Czechia, Estonia, Finland, France, Greece, Hungary, Ireland, Iceland, Italy, Japan, Korea, Luxembourg, Latvia, New Zealand, Portugal, Slovenia, Sweden, T\"{u}rkiye, and the United States.
We define the periods from 2009 Q1 to 2016 Q1 as the pre-treatment periods to eliminate the effects of the 2008 financial crisis, and those from 2016 Q2 to 2018 Q2 as the post-treatment periods.
For each country, time period, and service category, we quantify trade by the total value of imports and exports, expressed in millions of US dollars.
This construction yields a 9-dimensional vector of service-specific trade volumes, denoted by $Q_{it} = (Q_{it1}, \ldots, Q_{it9})'$, for country $i$ in period $t$.
Following \cite{dubey2020functional}, we construct a $(9 \times 9)$ covariance matrix of service trade volumes, $Y_{it} = (Q_{it} - \bar{Q}_{t})(Q_{it} - \bar{Q}_{t})'$, 
where $\bar{Q}_t = (\bar{Q}_{t1}, \ldots, \bar{Q}_{t9})'$ denotes the cross-sectional mean of $Q_{it}$ across the 23 countries.
These covariance matrices serve as the outcomes and are viewed as elements of the space $\mathrm{Sym}_9^+$ equipped with the Frobenius metric $d_F$ (see Example \ref{exm:covmat}).
In this application, we define the differences between the covariance matrices $Y_{1t}^I - Y_{1t}^N, t=\text{2016 Q1}, \text{2016 Q2}, \ldots, \text{2018 Q2}$ as the causal effects of interest.
The proposed approach was implemented to obtain the FSC and ridge augmented FSC estimators.
For the modification of the augmented estimator $\hat{Y}_{1t}^{N, \text{aug}}$, we employed the \texttt{NearPD} function in the \texttt{Matrix} package \citep{bates2025matrix}.

Figure \ref{fig:service_difference_fsc} displays the heatmaps of the differences between the observed trade covariance matrices for the UK and the corresponding matrices for the FSC units. Figure \ref{fig:service_difference_afsc} presents the corresponding results for the augmented FSC units.
For brevity, we report results only for the periods from 2015 Q1 to 2017 Q2; results for the other periods are provided in Figures \ref{fig:service_difference_fsc_pre}, \ref{fig:service_difference_afsc_pre}, \ref{fig:service_difference_fsc_after}, and \ref{fig:service_difference_afsc_after} in the Appendix. 
We observe that the pre-treatment fit was improved by the augmentation. Specifically, the pre-treatment fit of the FSC is 39.3429, whereas that of the augmented FSC is 20.0639, indicating an improvement of approximately $49.0\%$ due to the augmentation. 
We also find that the estimated causal effects in the post-treatment periods are overall negative, suggesting that the variances and covariances of UK service trade volumes, $(Q_{1tj} - \bar{Q}_{tj})(Q_{1tk} - \bar{Q}_{tk}), j, k = 1, \ldots 9$,  decreased following the Brexit announcement.
This finding can be explained as follows.
UK service trade volumes are generally larger than the cross-sectional average (i.e., $Q_{1tj} > \bar{Q}_{tj}$). 
After the Brexit announcement, these trade volumes declined, bringing $Q_{1tj}$  closer to the average $\bar{Q}_{tj}$. As a result, the magnitudes of deviations from the mean were reduced, leading to the decreases in the variances and covariances.
Table \ref{table:service} reports the weights assigned to the control units. The weights obtained using the FSC are sparse, with only France, Greece, and the United States receiving positive weights. In contrast, the augmented FSC yields substantially different weights from those of the FSC.

Figures \ref{fig:service_prediction_lower} and \ref{fig:service_prediction_upper} presents the lower and upper bounds of the $90\%$ prediction bands for the causal effects $Y_{1t}^I - Y_{1t}^N, t= \text{2016 Q2}, \ldots, \text{2018 Q2}$ based on the  augmented FSC.
For example, the upper bounds for the variances of SD in the periods 2017 Q3, 2017 Q4 and 2018 Q1 are negative, indicating that the causal effects for the variances of SD in these periods are significantly negative. 
To evaluate the statistical significance of the causal effects, we conducted a placebo permutation test based on the augmented FSC. The resulting $p$-values for the null hypothesis of no causal effect in each post-treatment period (2016 Q2--2018 Q2) are $0.130, 0.130, 0.130, 0.087, 0.130, 0.130, 0.130, 0.130$, and $0.130$, respectively. 
Detailed results of the permutation test are provided in Figure \ref{fig:placebo_service} in the Appendix.

\begin{figure}
    \centering
    \includegraphics[width=\linewidth]{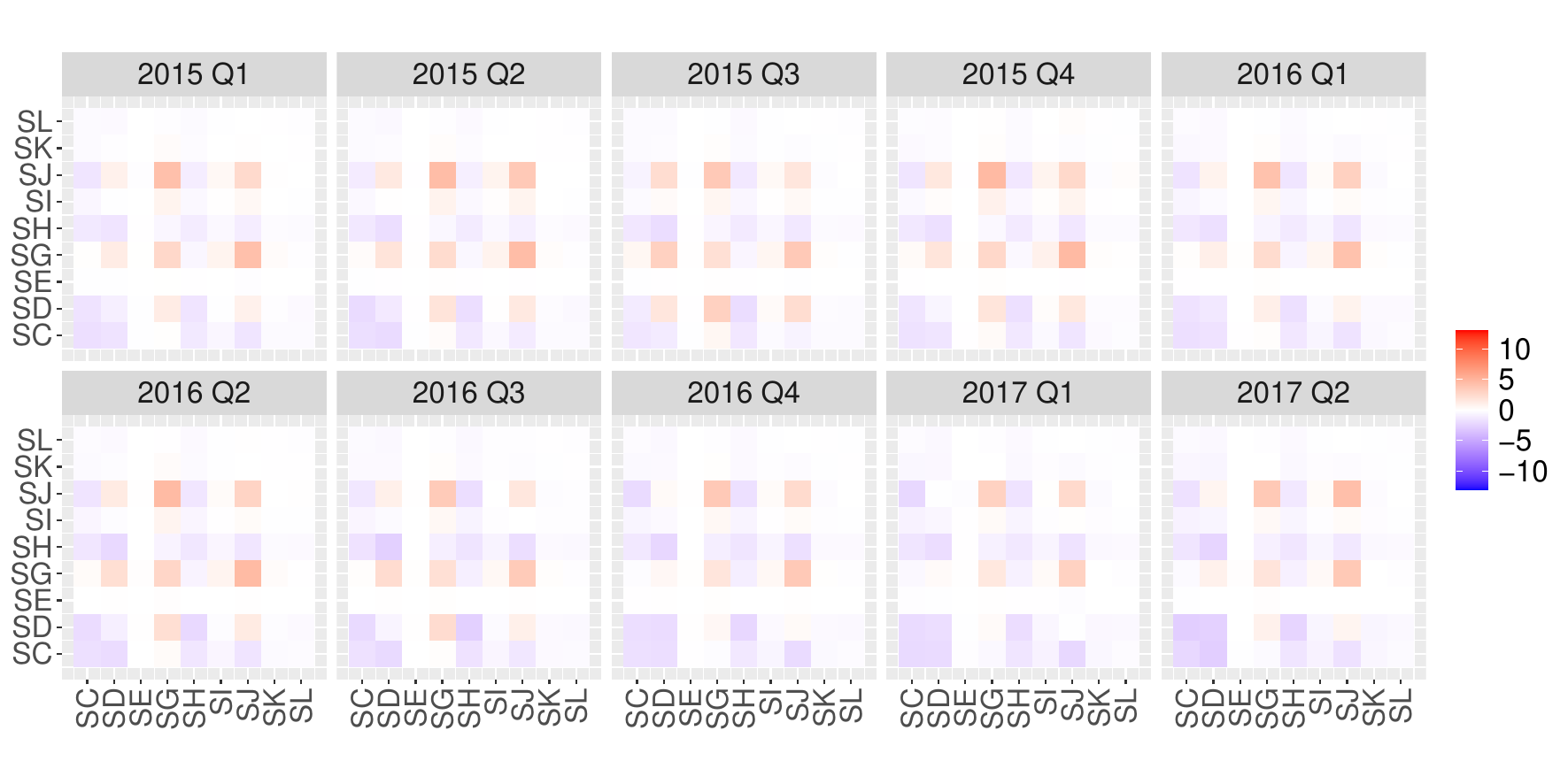}
    \caption{Heatmaps of the differences between the observed trade covariance matrices for the UK and the corresponding matrices for the FSC units during pre-treatment periods (2015 Q1--2016 Q1) and post-treatment periods (2016 Q2--2017 Q2).}
    \label{fig:service_difference_fsc}
\end{figure}

\begin{figure}
    \centering
    \includegraphics[width=\linewidth]{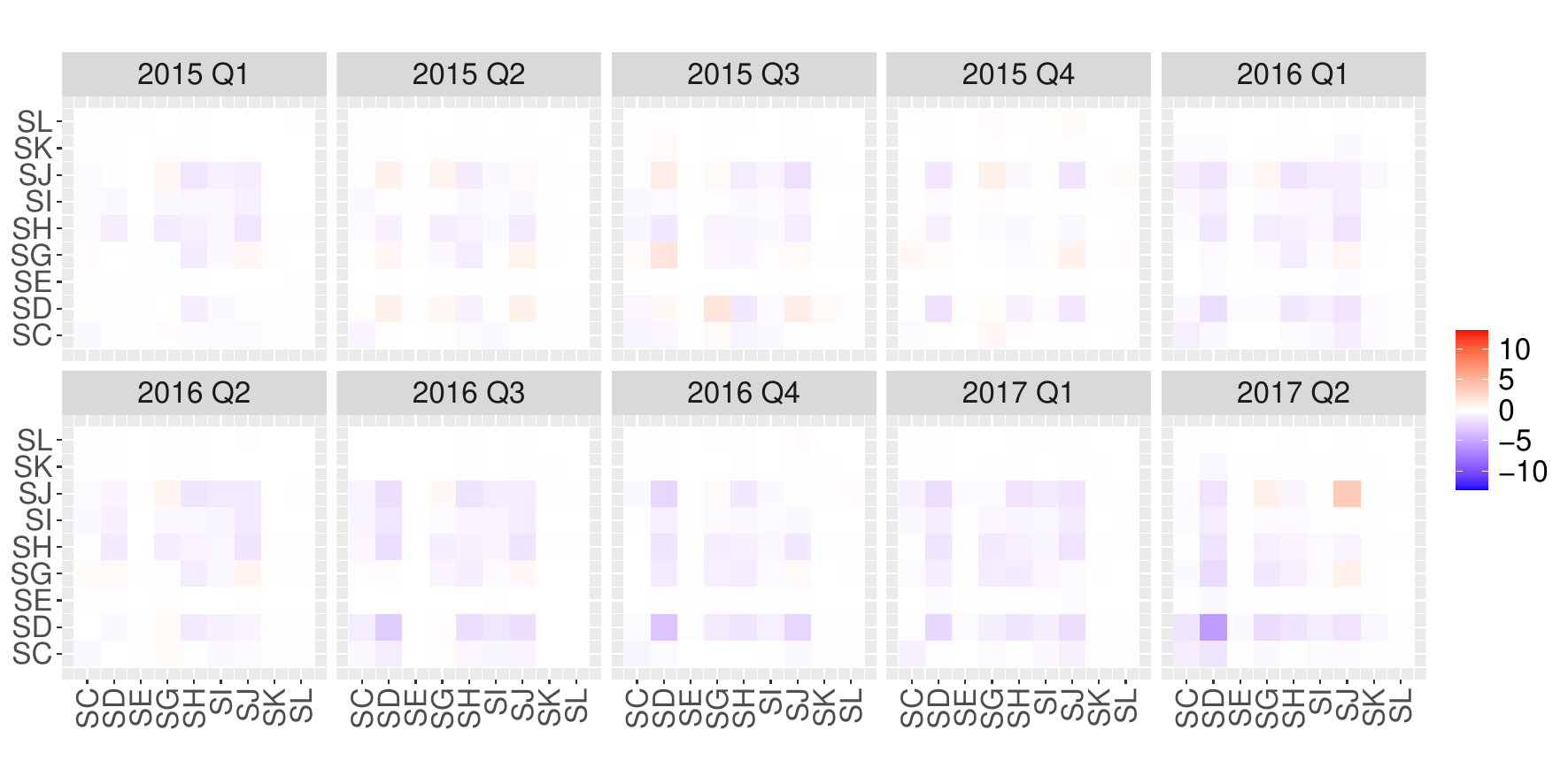}
    \caption{Heatmaps of the differences between the observed trade covariance matrices for the UK and the corresponding matrices for the ridge augmented FSC units during pre-treatment periods (2015 Q1--2016 Q1) and post-treatment periods (2016 Q2--2017 Q2).}
    \label{fig:service_difference_afsc}
\end{figure}

\begin{table}[t]
\centering
\begin{tabular}{l r r @{\hspace{2cm}} l r r}
\hline
Control country & FSC   & AFSC   & Control country & FSC  & AFSC  \\
\hline \hline
Australia & 0.000 & 1.335  & Italy & 0.000 & -0.858  \\
Austria & 0.000 & 1.152  & Japan & 0.000 & -0.777  \\
Canada & 0.000 & -12.621  & Korea & 0.000 & 0.658  \\
Czechia & 0.000 & 0.307  & Luxembourg & 0.000 & 0.243 \\
Estonia & 0.000 & 42.746  & Latvia & 0.000 & -42.588  \\
Finland & 0.000 & -0.783  & New Zealand & 0.000 & 5.887  \\
France & 0.596 & 0.415 & Portugal & 0.000 & -13.758 \\
Greece & 0.290 & 5.699  & Slovenia & 0.000 & -3.581  \\
Hungary & 0.000 & 16.669  & Sweden & 0.000 & -0.093  \\
Ireland & 0.000 & -0.420  & T\"{u}rkiye & 0.000  & 2.182  \\
Iceland & 0.000 & -0.771  & United States & 0.115   & -0.043 \\
\hline
\end{tabular}
\caption{Control unit weights for the service trade data obtained using the FSC and ridge augmented FSC (AFSC).}
\label{table:service}
\end{table}

\begin{figure}
    \centering
    \includegraphics[width=\linewidth]{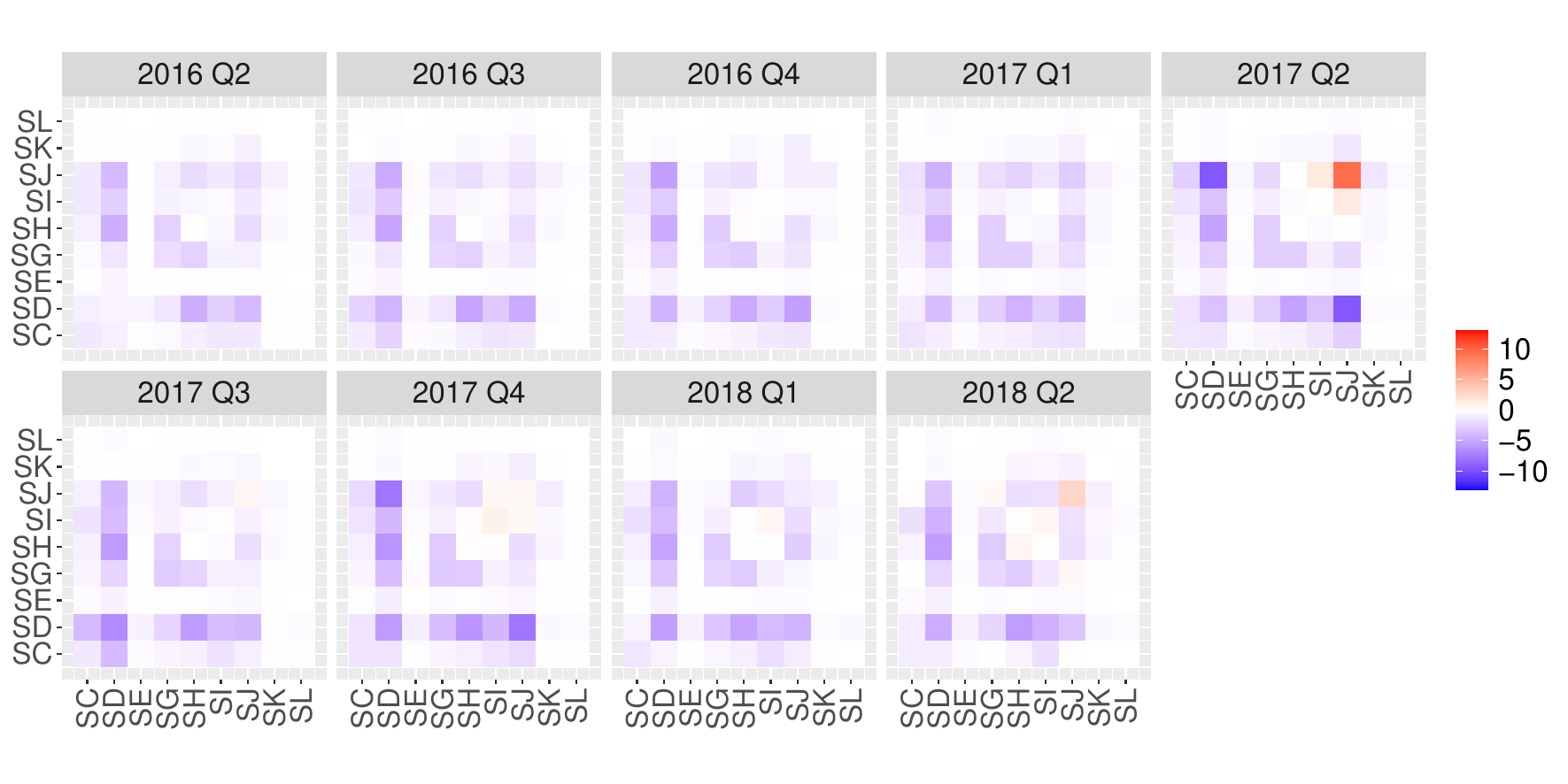}
    \caption{Lower bounds of the $90\%$ prediction bands for the causal effects $Y_{1t}^I - Y_{1t}^N, t= \text{2016 Q2}, \ldots, \text{2018 Q2}$ based on the ridge augmented FSC. }
    \label{fig:service_prediction_lower}
\end{figure}

\begin{figure}
    \centering
    \includegraphics[width=\linewidth]{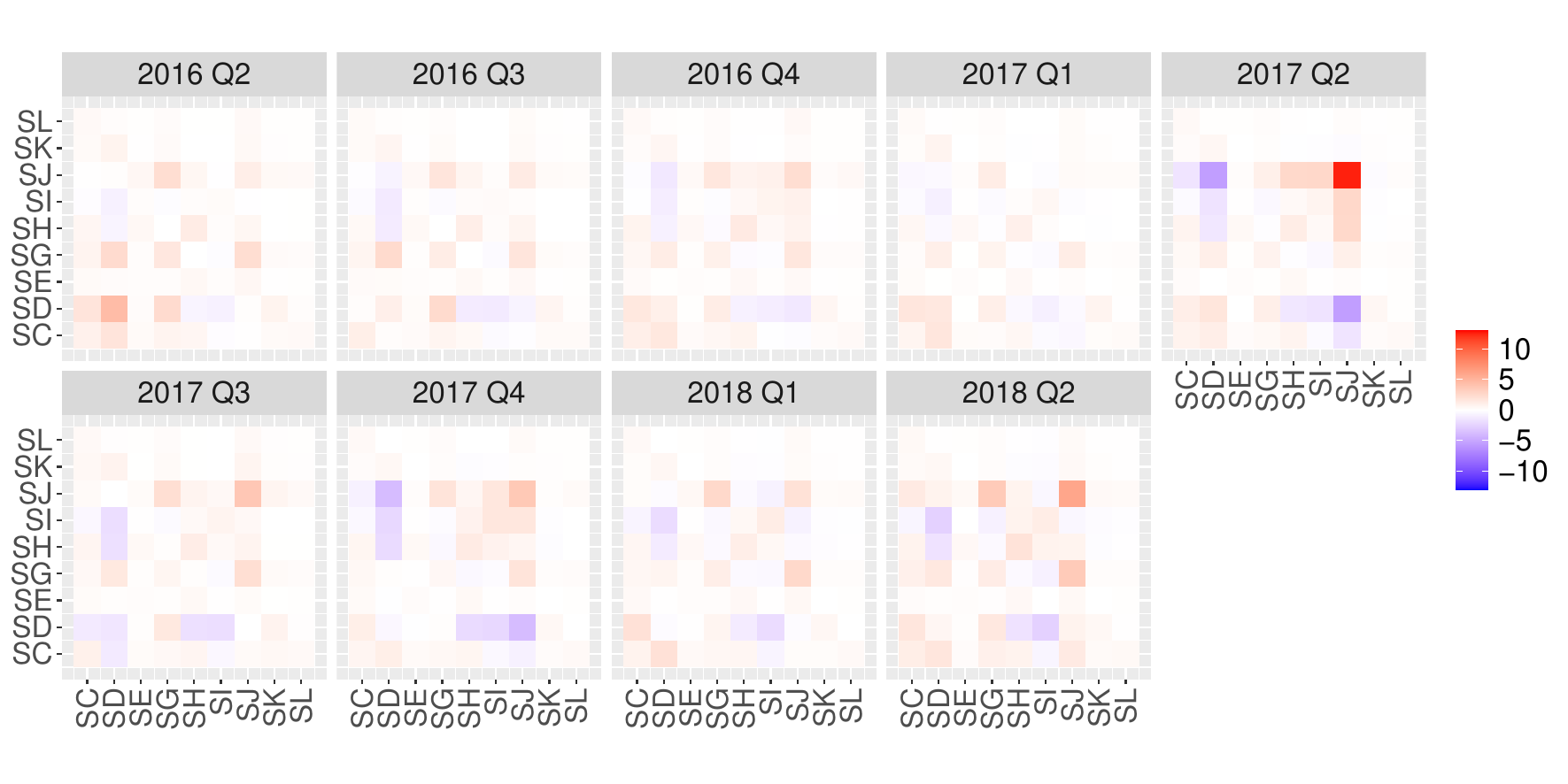}
    \caption{Upper bounds of the $90\%$ prediction bands for the causal effects $Y_{1t}^I - Y_{1t}^N, t= \text{2016 Q2}, \ldots, \text{2018 Q2}$ based on the ridge augmented FSC.}
    \label{fig:service_prediction_upper}
\end{figure}

\section{Concluding Remarks}
\label{sec:conclusion}
We have proposed the functional synthetic control (FSC) method as an extension of the SCM that accommodates metric space–valued outcomes, including functions, distributions, covariance matrices, and compositional data. Unlike existing frameworks, our approach leverages isometric embeddings of metric spaces into Hilbert spaces, thereby providing theoretical guarantees for estimation and enabling the construction of prediction sets for causal effects. We expect that the proposed FSC framework will serve as a useful addition to the toolbox for causal inference with complex outcomes.

There are several directions for future research. \cite{arkhangelsky2021synthetic} proposed the synthetic difference-in-differences (SDID) method, which combines desirable features of the SCM and difference-in-differences. \cite{kurisu2025geodesic} extended this method to outcomes residing in geodesic metric spaces. An important avenue for future work is to extend our FSC framework to SDID and investigate the theoretical properties of such an extension. 
Another promising direction is to extend the proposed framework to the case of a staggered adoption setting \citep{ben2022synthetic}. 

\paragraph{Acknowledgments.}
The authors thank seminar and conference participants for many valuable comments. 
Ryo Okano was partially supported by JSPS Grant-in-Aid for JSPS Fellows (25KJ0140).
Daisuke Kurisu was partially supported by JSPS KAKENHI Grant Number 23K12456.

\clearpage
\bibliographystyle{apalike}
\bibliography{ref}

\clearpage
\appendix

\counterwithin{lem}{section}
\counterwithin{thm}{section}
\counterwithin{prp}{section}
\counterwithin{dfn}{section}
\counterwithin{ass}{section}
\counterwithin{cor}{section}
\counterwithin{figure}{section}

\section{Projection onto Closed Convex Subsets of Hilbert Spaces}
\label{sec:projection}
The following lemma outlines the properties of projections onto closed convex subsets of Hilbert spaces. Statement (a) is a well-known result in functional analysis \citep[e.g.,][]{deutsch2001best}. For the sake of completeness, we include a proof in Section~\ref{sec:proof}.

\begin{lem}\label{lem:proj}
Assume that $C$ is a nonempty closed convex subset of a Hilbert space $\mathcal{H}$. For any $y \in \mathcal{H}$, define the projection of $y$ onto $C$ as $\pi_C(y) \in \mathrm{argmin}_{h \in C} \|y-h\|_\mathcal{H}$. Then we have the following:
\begin{enumerate}
    \item $\pi_C(y)$ exists and is unique.
    \item For any $u,v \in \mathcal{H}$, $\|\pi_C(u) - \pi_C(v)\|_{\mathcal{H}} \leq \|u-v\|_{\mathcal{H}}$.
    In particular, if $v \in \mathcal{H}$, $\|\pi_C(u) - v\|_{\mathcal{H}} \leq \|u-v\|_{\mathcal{H}}$. 
\end{enumerate}
\end{lem}

\section{Auxiliary Covariates}
\label{sec:covariates}
In this section, we consider the case with auxiliary covariates.
We first define the FSC and ridge augmented FSC estimators with covariates in Section \ref{sec:ridge_aug_fsc_cov}.
After examining the properties of the ridge augmented FSC with covariates in Section \ref{sec:prop_weights_cov}, we derive finite-sample error bounds under an autoregressive model with covariates in Section \ref{sec:est_error_auto_cov} and under a latent factor model with covariates in Section \ref{sec:est_error_factor_cov}.

\subsection{FSC and Ridge Augmented FSC Methods with Covariates}
\label{sec:ridge_aug_fsc_cov}
Suppose that each unit $i$ has auxiliary covariates
$Z_i = (Z_{i1}, \ldots, Z_{ip})' \in \mathbb{R}^p$.
We assume that the number of the covariates $p$ is small, and that the control units' covariates are centered, i.e.,   $\sum_{i=2}^n Z_i = 0$.
Let $Z_0$ denote the $(N-1) \times p$ matrix of covariates for the control units.
We incorporate the covariates in two ways:
(i) into the balance objective of the FSC, and (ii) into the outcome model $\hat{m}_{it}$, which is used in the construction of the augmented estimator. 

For the former, we extend  the FSC problem \eqref{eq:optim_fscm}  to choose the weights $\hat{\gamma}^{\text{cscm}} \in \Delta^{N-1}$ as the solution to 
\begin{equation}
    \min_{\gamma \in \Delta^{N-1}} \sum_{t=1}^{T_0} \left\| Y_{1t} - \sum_{i=2}^N \gamma_i Y_{it} \right\|_{\mathcal{H}}^2 + w \left\|Z_1 - \sum_{i=2}^N \gamma_i Z_i \right\|_2^2.
\end{equation}
Here, $w\ge 0$ is a user-specific hyperparameter that determines the importance of balancing the auxiliary covariates. 
Based on the weights $\hat{\gamma}^{\text{cscm}}$, for $t=T_0+1, ..., T$, the estimators for the counterfactual outcomes $Y_{1t}^N$ and $\nu_{1t}^N$ are defined as 
\begin{equation}
    \hat{Y}_{1t}^{N, \text{cscm}} = \sum_{i=2}^N  \hat{\gamma}_i^{\text{cscm}} Y_{it}, \quad 
    \hat{\nu}_{1t}^{N, \text{cscm}} = \Psi^{-1}(\hat{Y}_{1t}^{N, \text{cscm}})
\end{equation}
respectively. 

For the latter,  we extend the estimate \eqref{eq:ridge_estimate_approx} to 
\begin{equation}
    \hat{m}_{it}^{\text{cov}}(x)
    =
    \hat{\eta}_0(x) + \sum_{s=1}^{T_0}\hat{\theta}_s(x)'r_{is} + \hat{\delta}(x)'Z_i, \quad x \in \mathcal{X},
\end{equation}
where for each $x \in \mathcal{X}$, 
\begin{align}
    &\left\{\hat{\eta}_0(x), \hat{\theta}_1(x), \ldots, \hat{\theta}_{T_0}(x), \hat{\delta}(x)\right\}  \\
    & =
    \argmin_{\eta_0, \theta_1, \ldots, \theta_{T_0}, \delta} \sum_{i=2}^N\left\{Z_{it}(x) - \left(\eta_0 + \sum_{s=1}^{T_0}\theta_s'r_{is} + \delta'Z_i \right)\right\}^2 +\lambda \sum_{s=1}^{T_0} \|\theta_s\|_2^2. 
    \label{eq:regression_prob_cov}
\end{align}
Here, $\lambda > 0$ is a regularization parameter.
The ridge augmented FSC estimator for $Y_{1t}^N$ with covariates is then
\begin{equation}
    \hat{Y}_{1t}^{N, \text{cov}}
    =\sum_{i=2}^N \hat{\gamma}_i^{\text{scm}} Y_{it}
    +
    \left\{
    \hat{m}_{1t}^{\text{cov}} - \sum_{i=2}^N \hat{\gamma}_i^{\text{scm}} \hat{m}_{it}^{\text{cov}}
    \right\},
    \label{eq:ridge_aug_cov}
\end{equation}
and its modified version $\tilde{Y}_{1t}^{N, \text{cov}}$
is defined analogously to  \eqref{eq:aug_fscm_modif}.

The next lemma shows that the estimator \eqref{eq:ridge_aug_cov} is itself a weighting estimator. This result is analogous to Lemma 4 in \cite{ben2021augmented}. 
For $i=1, \ldots, N$, let $\check{r}_{i\cdot} \in \mathbb{R}^{KT_0}$ be the residual components of the linear regression of $r_{i\cdot}$ on the control auxiliary covariates $Z_i$, and let $\check{r}_{0\cdot} \in \mathbb{R}^{(N-1) \times (KT_0)}$ be the matrix of control units' residual components. In other words, 
\begin{equation}
    \check{r}_{i\cdot} = r_{i\cdot} - r_{0\cdot}'Z_{0}'(Z_{0}'Z_{0})^{-1} Z_i, \quad 
    \check{r}_{0\cdot}
    =\begin{pmatrix}
    \check{r}_{2\cdot}' \\ 
    \vdots \\ 
    \check{r}_{N\cdot}'
\end{pmatrix}.
\end{equation}

\begin{lem}
    For any $t=T_0+1, \ldots,T$, the ridge augmented FSC estimator with covariates \eqref{eq:ridge_aug_cov} is expressed as 
    \begin{equation}
        \hat{Y}_{1t}^{N, \mathrm{cov}}
        =
        \sum_{i=2}^N \hat{\gamma}_i^{\mathrm{cov}} Y_{it},
    \end{equation}
    where the weights $\hat{\gamma}^{\mathrm{cov}} = (\hat{\gamma}^{\mathrm{cov}}_i)_{i=2}^N$ are given by 
    \begin{equation}
        \hat{\gamma}_i^{\mathrm{cov}}
        =
        \hat{\gamma}_i^{\mathrm{scm}}
        +
        (\check{r}_{1\cdot} - \check{r}_{0\cdot}' \hat{\gamma}^{\mathrm{scm}})' 
        (\check{r}_{0\cdot}'\check{r}_{0\cdot} + \lambda I_{KT_0})^{-1} \check{r}_{i}
        +
        (Z_1 - Z_{0}'\hat{\gamma}^{\mathrm{scm}})'(Z_{0}'Z_{0})^{-1}Z_i.
        \label{eq:weights_cov}
    \end{equation}
    \label{lem:weighting_expression_cov}
\end{lem}

\subsection{Properties of Weights from Ridge Augmented FSC with Covariates}
\label{sec:prop_weights_cov}
We investigate the properties of the weights $\hat{\gamma}^{\text{cov}} = (\hat{\gamma}_i^{\text{cov}})_{i=2}^N$ obtained from the ridge augmented FSC with covariates. In particular, we analyze how the covariate fit $\|Z_1 - \sum_{i=2}^N \hat{\gamma}_i^{\text{cov}}Z_i\|_2$, the pre-treatment fit 
$
\sqrt{\sum_{t=1}^{T_0}\left\|Y_{1t} - \sum_{i=2}^N \hat{\gamma}_i^{\mathrm{cov}}Y_{it}\right\|_{\mathcal{H}}^2},
$
and the norm $\|\hat{\gamma}^{\text{cov}}\|_2$ depend on the regularization parameter $\lambda$. 

To make the dependence on the number of the orthonormal vectors $K$ explicit,  we denote $\check{r}_{0\cdot}$ and $\hat{\gamma}^{\text{cov}}$ by $\check{r}_{0\cdot}^{(K)}$ and $\hat{\gamma}^{\text{cov}(K)}$, respectively.
For any positive integer $K$, let $\check{m}(K)$ be the rank of the matrix $\check{r}_{0\cdot}^{(K)}$, and let $\check{d}_{\text{max}}^{(K)}$, $\check{d}_{\text{min}}^{(K)}$ be the maximum and minimum singular values of $\check{r}_{0\cdot}^{(K)}$, respectively. 
Note that since $\check{r}_{0\cdot}^{(K)}$ is an $(N-1) \times (KT_0)$ matrix, we have $\check{m}(K) \le N-1$ for all $K$. 

\begin{ass}
    There exist constants $C_2 > 0$ and $c_2 > 0$ such that $\check{d}_{\mathrm{max}}^{(K)} \le C_2$ and $\check{d}_{\mathrm{min}}^{(K)} \ge c_2$ hold for any positive integer $K$.
    \label{ass:singular_values_cov}
\end{ass}

Under this assumption, we obtain the following result.

\begin{lem} 
    \begin{enumerate}
        \item For any positive integer $K$, the weights $\hat{\gamma}^{\mathrm{cov}(K)}$ with any regularization parameter $\lambda > 0$  exactly balance the auxiliary covariates: 
        $$
        Z_1 - \sum_{i=2}^N\hat{\gamma}_i^{\mathrm{cov(K)}}Z_i = 0. 
        $$
        \item Suppose  Assumption \ref{ass:singular_values_cov} holds. Then, for any positive integer $K$, the weights $\hat{\gamma}^{\mathrm{cov}(K)}$ with regularization parameter $\lambda > 0$ satisfy 
        \begin{align}
        & \sqrt{\sum_{t=1}^{T_0}\left\|Y_{1t} - \sum_{i=2}^N \hat{\gamma}_i^{\mathrm{cov}(K)}Y_{it}\right\|_{\mathcal{H}}^2 }
        \le F_1(\lambda)
        +
        R_5^{(K)}.
        \label{eq:prefit_bound_cov}
    \end{align}
    Here,
    \begin{align}
        F_1(\lambda) &= \frac{\sqrt{\check{m}(K)}\lambda}{(\check{d}_{\mathrm{min}}^{(K)})^2+\lambda}\sqrt{\sum_{t=1}^{T_0}\left\|Y_{1t} - \sum_{i=2}^N \hat{\gamma}_i^{\mathrm{scm}} Y_{it}\right\|_{\mathcal{H}}^2} \\
        & \phantom{\le} +
        \frac{\sqrt{\check{m}(K)}\lambda}{(\check{d}_{\mathrm{min}}^{(K)})^2+\lambda}\sqrt{\sum_{t=1}^{T_0}\sum_{i=2}^N\|Y_{it}\|_{\mathcal{H}}^2}
        \|Z_0(Z_0'Z_0)^{-1}(Z_1 - Z_0'\hat{\gamma}^{\mathrm{scm}})\|_2,
        \label{eq:func_def_1}
    \end{align}
    and $R_5^{(K)} \to 0$ as $K \to \infty$.
    \item Suppose  Assumption \ref{ass:singular_values_cov} holds. Then, for any positive integer $K$, the weights $\hat{\gamma}^{\mathrm{cov}(K)}$ with regularization parameter $\lambda > 0$ satisfy 
    \begin{align}
       \|\hat{\gamma}^{\mathrm{cov}(K)}\|_2 
        \le 
        F_2(\lambda)+R_6^{(K)}.
        \label{eq:norm_bound_cov}
    \end{align}
    Here,
    \begin{align}
        F_2(\lambda)& = \|\hat{\gamma}^{\mathrm{scm}}\|_2 
        +
        \frac{\sqrt{\check{m}(K)}\check{d}_{\mathrm{max}}^{(K)}}{(\check{d}_{\mathrm{min}}^{(K)})^2 + \lambda} \sqrt{\sum_{t=1}^{T_0}\left\|Y_{1t} - \sum_{i=2}^N \hat{\gamma}_i^{\mathrm{scm}} Y_{it}\right\|_{\mathcal{H}}^2} \\
        & \phantom{\le}+ 
        \frac{\sqrt{\check{m}(K)}\check{d}_{\mathrm{max}}^{(K)}}{(\check{d}_{\mathrm{min}}^{(K)})^2 + \lambda} 
\sqrt{\sum_{t=1}^{T_0}\sum_{i=2}^N\|Y_{it}\|_{\mathcal{H}}^2}
        \|Z_0(Z_0'Z_0)^{-1}(Z_1 - Z_0'\hat{\gamma}^{\mathrm{scm}})\|_2 \\ 
        & \phantom{\le} + 
        \|Z_0(Z_0'Z_0)^{-1}(Z_1 - Z_0'\hat{\gamma}^{\mathrm{scm}})\|_2,
        \label{eq:func_def_2}
    \end{align}
    and $R_6^{(K)} \to 0$ as $K \to \infty$.
    \end{enumerate}
    \label{lem:cov_weight_prop}
\end{lem}

The main term $F_1(\lambda)$ of the bound in \eqref{eq:prefit_bound_cov} decreases as $\lambda$ decreases, and in particular, it vanishes as $\lambda \to 0$. This implies that the ridge augmented FSC method with covariates can achieve an almost perfect fit when $\lambda$ is close to zeto (assuming $K$ is large). In contrast, the main term $F_2(\lambda)$ of the bound in \eqref{eq:norm_bound_cov} increases as $\lambda$ decreases, indicating that the norm of the weights can become large when $\lambda$ is small.

\subsection{Estimation Error under Autoregressive Model with Covariates}
\label{sec:est_error_auto_cov}
In the rest of this section, we assume that $T = T_0+1$.
As in Section \ref{sec:est_error},  to clarify our discussion, we define a general weighting estimator $\hat{Y}_{1T}^N = \sum_{i=2}^N \hat{\gamma}_i Y_{iT}$, where the weights $(\hat{\gamma}_i)_{i=2}^N$ are not dependent on the post-treatment outcomes $Y_{1T}, \ldots, Y_{NT}$, and satisfy $\sum_{i=2}^N \hat{\gamma}_i=1$. The ridge augmented FSC estimator with covariates, $\hat{Y}_{1T}^{N, \text{cov}}$, takes this form, as shown in Lemma \ref{lem:weighting_expression_cov}. 
We also define its modification,  
$
\tilde{Y}_{1T}^{N}
    =
    \argmin_{y \in \mathcal{Y}} \| y - \hat{Y}_{1T}^{N}\|_{\mathcal{H}},
$
and construct an estimator for $\nu_{1T}^N$ as 
$\hat{\nu}_{1T}^N = \Psi^{-1}(\tilde{Y}_{1T}^{N})$. In what follows, we focus on bounding the estimation error $d(\nu_{1T}^N, \hat{\nu}_{1T}^N)$.

We begin by assuming the following data-generating process that includes covariates.
\begin{ass}[Autoregressive model with covariates]
\label{ass_auto_cov}
    For each unit $i=1, \ldots, N$,
    we assume that the post-treatment control potential outcome $Y_{1T}^N \in \mathcal{Y}$ is generated as
    \begin{equation}
        Y_{iT}^{N}(x)
        =
        \sum_{t=1}^{T_0} \langle\beta_t(x, \cdot), Y_{it}^{N}  \rangle 
        +
        \sum_{\ell=1}^p\eta_{\ell}(x)Z_{i\ell}
         + \varepsilon_{iT}(x), \quad x \in \mathcal{X}.
         \label{eq:autregressive_cov}
    \end{equation}
    Here, $\beta_t(\cdot, \cdot): \mathcal{X}^2 \to \mathbb{R}, t=1, \ldots, T_0$ and $\eta_\ell: \mathcal{X} \to \mathbb{R}, \ell=1, \ldots, p$ are coefficient functions, and $\varepsilon_{iT} \in \mathcal{H}, i=1, \ldots, N$ are independent mean-zero errors. Furthermore, we assume that there exits a constant $\sigma > 0$ such that  $\|\varepsilon_{iT}\|_{\mathcal{H}} \le \sigma$ holds almost surely for all $i=1, \ldots, N$.
\end{ass}

Under this assumption, we derive the following finite-sample error bound. 
\begin{thm}
    Suppose Assumption \ref{ass_auto_cov} holds. Then, for any $\delta > 0$, the generic estimator $\hat{\nu}_{1T}^N$ satisfies 
    \begin{align} 
        d(\nu_{1T}^N, \hat{\nu}_{1T}^N)
        &\le \sqrt{\sum_{t=1}^{T_0} \|\beta_t\|_{\mathcal{H} \times \mathcal{H}}^2} \sqrt{\sum_{t=1}^{T_0} \left\|Y_{1t} - \sum_{i=2}^N \hat{\gamma}_iY_{it}\right\|_{\mathcal{H}}^2} 
        +
        \sqrt{\sum_{\ell=1}^p\|\eta_\ell\|_{\mathcal{H}}^2}\left\|Z_1 - \sum_{i=2}^N \hat{\gamma}_i Z_i\right\|_2 \\ 
        & \phantom{\le} +
        \delta \sigma (1+\|\hat{\gamma}\|_2)
        \label{eq:error_bound_auto_cov}
    \end{align}
    with probability at least $1 - 2e^{-\delta^2/2}$.
    \label{thm:est_error_auto_cov}
\end{thm}
From \eqref{eq:error_bound_auto_cov},
we observe that the estimation error of the generic estimator is governed by the pre-treatment fit $\sqrt{\sum_{t=1}^{T_0} \left\|Y_{1t} - \sum_{i=2}^N \hat{\gamma}_i Y_{it}\right\|_{\mathcal{H}}^2}$, the covariate fit $\|Z_1-\sum_{i=2}^N \hat{\gamma}_i Z_i\|_2$ and the norm of the weights $\|\hat{\gamma}\|_2$.
In contrast to the bound in Theorem \ref{thm:est_error_auto}, the bound in Theorem \ref{thm:est_error_auto_cov} additionally involves the covariate fit. 
This suggests that achieving accurate estimation under Assumption \ref{ass_auto_cov} requires balancing both the pre-treatment fit and the covariate fit.

Combining this result with Lemmas \ref{lem:aug_weight_prop} and \ref{lem:cov_weight_prop}, we obtain the following bounds for the ridge augmented FSC estimator without covariates, $\hat{\nu}_{1T}^{N, \text{aug}}$, and with covariates, $\hat{\nu}_{1T}^{N, \text{cov}}$.

\begin{cor}
\label{cor:est_error_auto_cov}
    \begin{enumerate}
        \item Suppose Assumptions \ref{ass:singular_values} and \ref{ass_auto_cov} hold. Then, for any positive integer $K$ and $\delta > 0$, the estimator $\hat{\nu}_{1T}^{N, \mathrm{aug}}$ with regularization parameter $\lambda > 0$ satisfies 
    \begin{align}
         d(\nu_{1T}^N, \hat{\nu}_{1T}^{N, \mathrm{aug}(K)})
        &\le  
         \sqrt{\sum_{t=1}^{T_0} \|\beta_t\|_{\mathcal{H}\times \mathcal{H}}^2} 
        F_3(\lambda)  
       +\sqrt{\sum_{\ell=1}^p\|\eta_\ell\|_{\mathcal{H}}^2}
        \left\|Z_1 - \sum_{i=2}^N \hat{\gamma}_i^{\mathrm{aug}(K)} Z_i\right\|_2 \\
        &\phantom{\le}
        +
        \delta \sigma( 1+ F_4(\lambda)) + R_7^{(K)}
        \label{eq:error_bound_ridge_auto_withoutcov}
    \end{align}
    with probability at least $1- 2e^{-\delta^2/2}$. Here,  
    \begin{equation}
        F_3(\lambda)
        =
        \frac{\sqrt{m(K)}\lambda}{(d_{\mathrm{min}}^{(K)})^2+\lambda}
        \sqrt{\sum_{t=1}^{T_0} \left\|Y_{1t} - \sum_{i=2}^N \hat{\gamma}_i^{\mathrm{scm}}Y_{it}\right\|_{\mathcal{H}}^2}
        \label{eq:func_def_3}
    \end{equation}
    and 
    \begin{equation}
        F_4(\lambda)
        =
    \|\hat{\gamma}^{\mathrm{scm}}\|_2 +
        \frac{\sqrt{m(K)}d_{\mathrm{max}}^{(K)}}{(d_{\mathrm{min}}^{(K)})^2 + \lambda}\sqrt{\sum_{t=1}^{T_0} \left\|Y_{1t} - \sum_{i=2}^N \hat{\gamma}_i^{\mathrm{scm}}Y_{it}\right\|_{\mathcal{H}}^2},
        \label{eq:func_def_4}
    \end{equation}
    and $R_7^{(K)} \to 0$ as $K \to 0$.
    \item Suppose Assumptions \ref{ass:singular_values_cov} and \ref{ass_auto_cov} hold. Then, for any positive integer $K$ and $\delta > 0$, the estimator $\hat{\nu}_{1T}^{N, \mathrm{cov}}$ with regularization parameter $\lambda > 0$ satisfies  
    \begin{align}
         d(\nu_{1T}^N, \hat{\nu}_{1T}^{N, \mathrm{cov}(K)})
        \le  \sqrt{\sum_{t=1}^{T_0} \|\beta_t\|_{\mathcal{H} \times \mathcal{H}}^2} F_1(\lambda)  + \delta \sigma(1+F_2(\lambda)) + R_8^{(K)}
        \label{eq:error_bound_ridge_auto_cov}
    \end{align}
    with probability at least $1- 2e^{-\delta^2/2}$. Here, $F_1(\lambda)$ and $F_2(\lambda)$ are defined in \eqref{eq:func_def_1} and \eqref{eq:func_def_2}, respectively, and $R_8^{(K)} \to 0$ as $K \to \infty$.
    \end{enumerate}
\end{cor}
We observe that the regularization parameter $\lambda$ in the estimators $\hat{\nu}_{1T}^{N, \text{aug}}$ and $\hat{\nu}_{1T}^{N, \text{cov}}$ control the trade-off between the components of the bounds. Specifically, the terms $\sqrt{\sum_{t=1}^{T_0} \|\beta_t\|_{\mathcal{H}\times \mathcal{H}}^2} F_3(\lambda)$ in \eqref{eq:error_bound_ridge_auto_withoutcov} and $\sqrt{\sum_{t=1}^{T_0} \|\beta_t\|_{\mathcal{H} \times \mathcal{H}}^2} F_1(\lambda) $ in \eqref{eq:error_bound_ridge_auto_cov} are increasing in $\lambda$, whereas the terms $\delta \sigma( 1+ F_4(\lambda))$ in \eqref{eq:error_bound_ridge_auto_withoutcov} and $\delta \sigma(1+F_2(\lambda))$ in \eqref{eq:error_bound_ridge_auto_cov} are decreasing in $\lambda$. Comparing the bounds in \eqref{eq:error_bound_ridge_auto_withoutcov} and \eqref{eq:error_bound_ridge_auto_cov}, we see that none of the terms in \eqref{eq:error_bound_ridge_auto_cov} corresponds to the covariate fit, while the term $\sqrt{\sum_{\ell=1}^p\|\eta_\ell\|_{\mathcal{H}}^2}
        \left\|Z_1 - \sum_{i=2}^N \hat{\gamma}_i^{\mathrm{aug}(K)} Z_i\right\|_2$ in \eqref{eq:error_bound_ridge_auto_withoutcov}
corresponds to the covariate fit. 
This implies that when the covariates are predictive of the potential outcomes, i.e., when $\sqrt{\sum_{\ell=1}^p\|\eta_\ell\|_{\mathcal{H}}^2}$ is large, incorporating the covariates can improve the estimation accuracy.

\subsection{Estimation Error under Latent Factor Model with Covariates}
\label{sec:est_error_factor_cov}
We next assume the following data-generating process. 
\begin{ass}[Latent factor model with covariates]
\label{ass_latent_cov}
Suppose that there are $J$ unknown, latent  factors at each time $t=1, \ldots, T$, denoted by $\mu_t = (\mu_{jt})_{j=1}^J$, where each $\mu_{jt}$ is in $\mathcal{H}$.
    For each $x \in \mathcal{X}$,
    we define the vectors of pre-treatment factors, $\mu_t(x) \in \mathbb{R}^J, t=1, \ldots, T_0$, and the matrix $\mu(x) \in \mathbb{R}^{T_0 \times J}$ as 
    \[
    \mu_t(x) = \begin{pmatrix}
    \mu_{1t}(x) \\ 
    \vdots \\ 
    \mu_{Jt}(x)
    \end{pmatrix}, \quad 
    \mu(x) = 
    \begin{pmatrix}
        \mu_1(x)' \\ 
        \vdots \\ 
        \mu_{T_0}(x)'
    \end{pmatrix}.
    \]
    We assume that there exist constant $M_1 > 0$  such that $|\mu_{jt}(x)| \le M_1$ for all $j, t, x$. Furthermore, 
    we assume that there exists a constant $M_2 > 0$ such that, for any $x$, the minimum eigenvalue of the matrix  $\mu(x)'\mu(x)$, denoted by $\xi_{\text{min}}(x)$, satisfies   $\xi_{\text{min}}(x) \ge M_2$.  
    In addition,  suppose that each unit $i$ has a vector of unknown factor loadings $\phi_i = (\phi_{ij})_{j=1}^J \in \mathbb{R}^J$.  We assume that for each unit $i = 1, \ldots, N$, the control potential outcome at time period $t=1, \ldots, T$ is generated as 
    \begin{equation}
        Y_{it}^{N}(x)
        =
        \sum_{j=1}^{J}  \phi_{ij} \mu_{jt}(x)
         + 
         \sum_{\ell=1}^p\eta_{\ell t}(x)Z_{i\ell}
         +
         \varepsilon_{it}(x), \quad x \in \mathcal{X},
         \label{eq:latent_factor_cov}
    \end{equation}
    where $\eta_{\ell t}: \mathcal{X} \to \mathbb{R}, \ell=1, \ldots p, t=1, \ldots T$ are coefficient functions, and $\varepsilon_{it} \in \mathcal{H}, i=1, \ldots, N, t=1, \ldots, T$ are independent, zero-mean errors.  We further assume that there exists a constant \( \sigma > 0 \) such that \( \| \varepsilon_{it} \|_{\mathcal{H}} \leq \sigma \) almost surely for every \( i \) and \( t \).
\end{ass}

Under this assumption, we derive the following finite-sample error bound for the generic estimator.

\begin{thm}
Suppose Assumption \ref{ass_latent_cov} holds. Then, for any $\delta > 0$, the the generic weighting estimator satisfies
    \begin{align}
        d(\nu_{1T}^N, \hat{\nu}_{1T}^N) 
        &\le \frac{M_1^2J^{3/2}}{M_2\sqrt{T_0}} \sqrt{\sum_{t=1}^{T_0} \left\|Y_{1t} - \sum_{i=2}^N \hat{\gamma}_iY_{it}\right\|_{\mathcal{H}}^2} \\
        & \phantom{\le}
        +
        \sqrt{2}\max\left\{1,\frac{M_1^2J^{3/2}}{M_2\sqrt{T_0}}\right\} \sqrt{\sum_{t=1}^T \sum_{\ell=1}^p \|\eta_{\ell t}\|_{\mathcal{H}}^2}\left\|Z_1 - \sum_{i=2}^N \hat{\gamma}_i Z_i\right\|_2 \\
        & \phantom{\le}
        +
        \frac{2\sigma M_1^2J^{3/2}}{M_2}\|\hat{\gamma}\|_1
        +
        \delta \sigma(1+\|\hat{\gamma}\|_2)
        \label{eq:est_error_factor_cov}
    \end{align}
with probability at least $1 -  2e^{-\delta^2/2}$.
    \label{thm:est_error_factor_cov}
\end{thm}

From \eqref{eq:est_error_factor_cov}, we see that the estimation error of the generic estimator $\hat{\nu}_{1T}^N$ is governed by the pre-treatment fit $\sqrt{\sum_{t=1}^{T_0} \left\|Y_{1t} - \sum_{i=2}^N \hat{\gamma}_i Y_{it}\right\|_{\mathcal{H}}^2}$, the covariate fit $\|Z_1 - \sum_{i=2}^N \hat{\gamma}_i Z_i\|_2$ and the norms of the weights $\|\hat{\gamma}\|_1, \|\hat{\gamma}\|_2$.
In contrast to Theorem \ref{thm:est_error_factor}, the bound in this theorem additionally involves the covariate fit. 
This suggests that achieving accurate estimation under Assumption \ref{ass_latent_cov} requires balancing both the pre-treatment fit and the covariate fit.

Combining this result with Lemmas \ref{lem:aug_weight_prop} and \ref{lem:cov_weight_prop}, we obtain the following bounds for the estimators $\hat{\nu}_{1T}^{N, \text{aug}}$ and $\hat{\nu}_{1T}^{N, \text{cov}}$.

\begin{cor}
\begin{enumerate}
    \item Suppose Assumptions \ref{ass:singular_values} and \ref{ass_latent_cov} hold. Then, for any positive integer $K$ and $\delta > 0$, the estimator $\hat{\nu}_{1T}^{N, \mathrm{aug}}$ with regularization parameter $\lambda > 0$ satisfies 
    \begin{align}
        d(\nu_{1T}^N, \hat{\nu}_{1T}^{N, \mathrm{aug}(K)})  
        &\le \frac{M_1^2J^{3/2}}{M_2\sqrt{T_0}}F_3(\lambda) \\
        & \phantom{\le} 
        +
        \sqrt{2}\max\left\{1, \frac{M_1^2J^{3/2}}{M_2\sqrt{T_0}}\right\} \sqrt{\sum_{t=1}^T \sum_{\ell=1}^p \|\eta_{\ell t}\|_{\mathcal{H}}^2}\left\|Z_1 - \sum_{i=2}^N \hat{\gamma}_i^{\mathrm{aug}(K)} Z_i\right\|_2 \\
        & \phantom{\le} 
        +
        \sigma\left(\frac{2\sqrt{N-1} M_1^2J^{3/2}}{M_2}+\delta\right)F_4(\lambda) + \delta \sigma + R_9^{(K)}
    \end{align}
    with probability at least $1 - 2e^{-\delta^2/2}$. Here,  $F_3(\lambda)$ and $F_4(\lambda)$ are defined in \eqref{eq:func_def_3} and \eqref{eq:func_def_4}, respectively, and  $R_9^{(K)} \to 0$ as $K \to \infty$.
    \item Suppose Assumptions \ref{ass:singular_values_cov} and \ref{ass_latent_cov} hold. Then, for any positive integer $K$ and $\delta > 0$, the estimator $\hat{\nu}_{1T}^{N, \mathrm{cov}}$ with regularization parameter $\lambda > 0$ satisfies 
    \begin{align}
         d(\nu_{1T}^N, \hat{\nu}_{1T}^{N, \mathrm{cov}(K)}) 
         &\le 
         \frac{M_1^2J^{3/2}}{M_2\sqrt{T_0}} F_1(\lambda) \\
         & \phantom{\le} + 
         \sigma\left(\frac{2\sqrt{N-1} M_1^2J^{3/2}}{M_2}+\delta\right)F_2(\lambda) + \delta \sigma + R_{10}^{(K)}
    \end{align}
    with probability at least $1 - 2e^{-\delta^2/2}$. Here,  $F_1(\lambda)$ and $F_2(\lambda)$ are defined in \eqref{eq:func_def_1} and \eqref{eq:func_def_2}, respectively, and  $R_{10}^{(K)} \to 0$ as $K \to \infty$.
\end{enumerate}
    \label{cor:est_error_factor_cov}
\end{cor}

\section{Proofs}
\label{sec:proof}
\subsection{Proofs for Section \ref{sec:method}}
\paragraph*{Proof of Lemma \ref{lem:closed_form_RFASCM}.}
Fix $t \in \{T_0+1, \ldots,  T\}$.
Since $X_{is}, i=2, \ldots, N$ are centered for each $s=1, \ldots, T_0$, the intercept $\hat{\eta}_0(x)$ and 
coefficients
$\hat{\theta}(x) = (\hat{\theta}_1(x)', \ldots, \hat{\theta}_{T_0}(x)')' \in \mathbb{R}^{KT_0}$ are given by 
    \begin{equation}
        \hat{\eta}_0(x) = \frac{1}{N-1}\sum_{i=2}^N Y_{it}(x), \quad \hat{\theta}(x) = (r_{0\cdot}'r_{0\cdot} + \lambda I_{KT_0})^{-1} \sum_{i=2}^N r_{i\cdot}Y_{it}(x).
    \end{equation}
    Therefore, we have
    \begin{align}
        \hat{m}_{1t}^{\text{rid}}(x) - \sum_{i=2}^N \hat{\gamma}_i^{\text{scm}}\hat{m}^{\text{rid}}_{it}(x) 
        &=\hat{\eta}_0(x) + \sum_{s=1}^{T_0} \hat{\theta}_s(x)'r_{1s}
        -
        \sum_{i=2}^N \hat{\gamma}_i^{\text{scm}}\left\{ \hat{\eta}_0(x) + \sum_{s=1}^{T_0}\hat{\theta}_s(x)'r_{is} \right\} \\ 
        &= \sum_{s=1}^{T_0} \hat{\theta}_s(x)'\left(r_{1s} - \sum_{i=2}^N \hat{\gamma}^{\text{scm}}_ir_{is}\right) \\ 
        &= (r_{1\cdot} - r_{0\cdot}'\hat{\gamma}^{\text{scm}})'\hat{\theta}(x) \\ 
        &= \sum_{i=2}^N (r_{1\cdot} - r_{0\cdot}'\hat{\gamma}^{\text{scm}})' (r_{0\cdot}'r_{0\cdot} + \lambda I_{KT_0})^{-1} r_{i\cdot}Y_{it}(x).  
    \end{align}
Substituting this expression into \eqref{eq:est_ridge_y} gives
\begin{align}
    \hat{Y}_{1t}^{N, \text{aug}}(x)
    &=
    \sum_{i=2}^N \hat{\gamma}_i^{\text{scm}}Y_{it}(x)
    +
    \left\{\hat{m}_{1t}^{\text{rid}}(x) - \sum_{i=2}^N \hat{\gamma}_i^{\text{scm}}\hat{m}_{it}^{\text{rid}}(x)\right\} \\
    &= \sum_{i=2}^N \{\hat{\gamma}_i^{\text{scm}} +(r_{1\cdot} - r_{0\cdot}'\hat{\gamma}^{\text{scm}})' (r_{0\cdot}'r_{0\cdot} + \lambda I_{KT_0})^{-1} r_{i\cdot} \} Y_{it}(x),
\end{align}
which implies that
\[
\hat{Y}_{1t}^{N, \mathrm{aug}} = \sum_{i=2}^N \hat{\gamma}_i^{\mathrm{aug}}Y_{it},
\]
where the weights are given by
\[
\hat{\gamma}_i^{\mathrm{aug}} = \hat{\gamma}_i^{\mathrm{scm}} + (r_{1\cdot} - r_{0\cdot}'\hat{\gamma}^{\text{scm}})'(r_{0\cdot}'r_{0\cdot} + \lambda I_{KT_0})^{-1} r_{i \cdot}.
\]
The fact that $\hat{\gamma}^{\text{aug}}$ solves the penalized SCM problem \eqref{eq:const_optim_prob} can be verified by directly solving the constrained optimization.

\paragraph*{Proof of Lemma \ref{lem:aug_weight_prop}.}
We prove part (a) after establishing part (b).
For part (b), fix a positive integer $K$. By Lemma \ref{lem:closed_form_RFASCM}, 
\begin{align}
    \hat{\gamma}^{\text{aug}(K)}
    &=
    \hat{\gamma}^{\text{scm}} + r_{0\cdot}^{(K)}(r_{0\cdot}^{(K)}{'}r_{0\cdot}^{(K)} + \lambda I_{KT_0})^{-1}(r_{1\cdot}^{(K)} - r_{0\cdot}^{(K)}{'}\hat{\gamma}^{\text{scm}}). 
    \label{eq:weight_aug_rep}
\end{align}
Let $r_{0\cdot}^{(K)} = U^{(K)} D^{(K)} V^{(K)}{'}$ be the singular value decomposition of the matrix $r_{0\cdot}^{(K)}$. 
Here, $U^{(K)}$ and $V^{(K)}$ are $(N-1) \times m(K)$ and $(KT_0) \times m(K)$ orthogonal matrices, respectively.  
$D^{(K)} = \text{diag}\{d_1^{(K)}, \ldots, d_{m(K)}^{(K)}\}$  is an $m(K) \times m(K)$ diagonal matrix, where $d_1^{(K)}, \ldots, d_{m(K)}^{(K)}$ are the singular values of $r_{0\cdot}^{(K)}$ such that $d_{1}^{(K)} \ge \cdots \ge d_{m(K)}^{(K)}$. In Section \ref{sec:est_error},  we denote $d_{1}^{(K)}$ and $d_{m(K)}^{(K)}$ as $d_{\text{max}}^{(K)}$ and $d_{\text{min}}^{(K)}$, respectively. 
Using this decomposition, we have \begin{align}
    r_{0\cdot}^{(K)}{'}r_{0\cdot}^{(K)} + \lambda I _{KT_0}
    =
    V^{(K)}(D^{(K)})^{2}V^{(K)}{'} + \lambda I_{KT_0}
    =
    V^{(K)}\text{diag}\{(d_1^{(K)})^2+\lambda, \ldots, (d_{m(K)}^{(K)})^2+\lambda\} V^{(K)}{'},
\end{align}
which implies 
\begin{align}
    (r_{0\cdot}^{(K)}{'}r_{0\cdot}^{(K)} + \lambda I _{KT_0})^{-1}
&=
V^{(K)} \text{diag}\{\{(d_1^{(K)})^2 + \lambda\}^{-1}, \ldots, \{(d_{m(K)}^{(K)})^2 + \lambda\}^{-1}\} V^{(K)}{'} \\ 
&= V^{(K)} E^{(K)} V^{(K)}{'},
\label{eq:inv_formula}
\end{align}
where we denote $E^{(K)} = \text{diag}\{\{(d_1^{(K)})^2 + \lambda\}^{-1}, \ldots, \{(d_{m(K)}^{(K)})^2 + \lambda\}^{-1}\}$. Combining \eqref{eq:weight_aug_rep} and \eqref{eq:inv_formula},  we have 
\begin{align}
    \hat{\gamma}^{\text{aug}(K)}
    &= 
    \hat{\gamma}^{\text{scm}} 
    +
    r_{0\cdot}^{(K)}V^{(K)}E^{(K)}V^{(K)}{'} (r_{1\cdot}^{(K)} - r_{0\cdot}^{(K)}{'}\hat{\gamma}^{\text{scm}}) \\ 
    &= U^{(K)}D^{(K)} E^{(K)} V^{(K)}{'} (r_{1\cdot}^{(K)} - r_{0\cdot}^{(K)}{'}\hat{\gamma}^{\text{scm}}). 
\end{align}
By the triangle inequality and the sub-multiplicativity of the Frobenius norm, 
\begin{align}
    \|\hat{\gamma}^{\mathrm{aug}(K)}\|_2
    &\le 
    \|\hat{\gamma}^{\mathrm{scm}}\|_2
    +
    \|D^{(K)}E^{(K)}\|_F\|r_{1\cdot}^{(K)} - r_{0\cdot}^{(K)}{'}\hat{\gamma}^{\text{scm}}\|_2 \\ 
    &=  
    \|\hat{\gamma}^{\mathrm{scm}}\|_2
    +
    \sqrt{\sum_{j=1}^{m(K)}\left(\frac{d_j^{(K)}}{(d_j^{(K)})^2 + \lambda}\right)^2} \|r_{1\cdot}^{(K)} - r_{0\cdot}^{(K)}{'}\hat{\gamma}^{\text{scm}}\|_2 \\ 
    &\le 
    \|\hat{\gamma}^{\mathrm{scm}}\|_2
    +
     \frac{\sqrt{m(K)}d_{\text{max}}^{(K)}}{(d_{\text{min}}^{(K)})^2 + \lambda} \|r_{1\cdot}^{(K)} - r_{0\cdot}^{(K)}{'}\hat{\gamma}^{\text{scm}}\|_2 \\ 
     & =
     \|\hat{\gamma}^{\mathrm{scm}}\|_2 
     +
     \frac{\sqrt{m(K)}d_{\text{max}}^{(K)}}{(d_{\text{min}}^{(K)})^2 + \lambda} 
     \sqrt{\sum_{t=1}^{T_0} \left\|Y_{1t}^{N} - \sum_{i=2}^N \gamma_i^{\mathrm{scm}}Y_{it}\right\|_{\mathcal{H}}^2}
     +
     R_{1}^{(K)}, 
\end{align}
where 
\[
R_{1}^{(K)}
=
\frac{\sqrt{m(K)}d_{\text{max}}^{(K)}}{(d_{\text{min}}^{(K)})^2 + \lambda}
\left\{ 
\|r_{1\cdot}^{(K)} - r_{0\cdot}^{(K)}{'}\hat{\gamma}^{\text{scm}}\|_2
-
     \sqrt{\sum_{t=1}^{T_0} \left\|Y_{1t}^{N} - \sum_{i=2}^N \gamma_i^{\mathrm{scm}}Y_{it}\right\|_{\mathcal{H}}^2}
     \right\}.
\]
Note that under Assumption \ref{ass:singular_values}, 
\[
\frac{\sqrt{m(K)}d_{\text{max}}^{(K)}}{(d_{\text{min}}^{(K)})^2 + \lambda}
\le 
\frac{\sqrt{N-1}C_1}{c_1^2+\lambda}
\]
Moreover,  as $K \to \infty$, 
\begin{align}
\|r_{1\cdot}^{(K)} -r_{0\cdot}^{(K)}\hat{\gamma}^{\text{scm}}\|_2 
\to
\sqrt{\sum_{t=1}^{T_0}\sum_{k=1}^\infty \left(r_{1, t, k} - \sum_{i=2}^N r_{i, t, k}\right)^2}
=
\sqrt{\sum_{t=1}^{T_0} \left\|Y_{1t}^{N} - \sum_{i=2}^N \gamma_i^{\mathrm{scm}}Y_{it}\right\|_{\mathcal{H}}^2}.
\end{align}
Hence, we have that $R_1^{(K)} \to 0$ as $K \to \infty$.

For part (a), fix a positive integer $K$. Observe that 
\begin{align*}
     \sum_{t=1}^{T_0} \left\|Y_{1t} - \sum_{i=2}^N \gamma_i^{\mathrm{aug}(K)}Y_{it}\right\|_{\mathcal{H}}^2
     &=
     \sum_{t=1}^{T_0} \left\|X_{1t} - \sum_{i=2}^N \gamma_i^{\mathrm{aug}(K)}X_{it}\right\|_{\mathcal{H}}^2 \\
     &
     =\sum_{t=1}^{T_0} \sum_{k=1}^\infty \left(r_{1,t,k} - \sum_{i=2}^N \hat{\gamma}_i^{\mathrm{aug}(K)}r_{i,t,k} \right)^2 \\ 
     &= 
     \| r_{1\cdot}^{(K)} - r_{0\cdot}^{(K)}{'}\hat{\gamma}^{\mathrm{aug}(K)}\|_2^2 
     +
     \sum_{t=1}^{T_0} \sum_{k=K+1}^\infty \left(r_{1,t,k} - \sum_{i=2}^N \hat{\gamma}_i^{\mathrm{aug}(K)}r_{i,t,k} \right)^2,
\end{align*}
which  implies 
\begin{equation}
    \sqrt{\sum_{t=1}^{T_0} \left\|Y_{1t}^{N} - \sum_{i=2}^N \gamma_i^{\mathrm{aug}(K)}Y_{it}\right\|_{\mathcal{H}}^2}
    \le 
    \| r_{1\cdot}^{(K)} - r_{0\cdot}^{(K)}{'}\hat{\gamma}^{\mathrm{aug}(K)}\|_2
     +
     \sqrt{\sum_{t=1}^{T_0} \sum_{k=K+1}^\infty \left(r_{1,t,k} - \sum_{i=2}^N \hat{\gamma}_i^{\mathrm{aug}(K)}r_{i,t,k} \right)^2}.
     \label{eq:fit_aug_decompose}
\end{equation}

We now bound the first term in \eqref{eq:fit_aug_decompose}. From the expression \eqref{eq:weight_aug_rep},
\begin{align}
    r_{1\cdot}^{(K)} - r_{0\cdot}^{(K)}{'} \hat{\gamma}^{\text{aug}(K)}
    &=
    r_{1\cdot}^{(K)} - r_{0\cdot}^{(K)}{'}\{\hat{\gamma}^{\text{scm}} + r_{0\cdot}^{(K)}(r_{0\cdot}^{(K)}{'}r_{0\cdot}^{(K)} + \lambda I _{KT_0})^{-1}(r_{1\cdot}^{(K)} - r_{0\cdot}^{(K)}{'} \hat{\gamma}^{\text{scm}})\} \\
    &= \lambda(r_{0\cdot}^{(K)}{'}r_{0\cdot}^{(K)} + \lambda I _{KT_0})^{-1}(r_{1\cdot}^{(K)} - r_{0\cdot}^{(K)}{'} \hat{\gamma}^{\text{scm}}).
\end{align}
Combining this with \eqref{eq:inv_formula},  we obtain
\begin{align}
    \| r_{1\cdot}^{(K)} - r_{0\cdot}^{(K)}{'}\hat{\gamma}^{\mathrm{aug}(K)}\|_2^2
    & =
    ( r_{1\cdot}^{(K)} - r_{0\cdot}^{(K)}{'}\hat{\gamma}^{\mathrm{aug}(K)})'( r_{1\cdot}^{(K)} - r_{0\cdot}^{(K)}{'}\hat{\gamma}^{\mathrm{aug}(K)}) \\ 
    &= \{\lambda V^{(K)} E^{(K)} V^{(K)}{'} (r_{1\cdot}^{(K)} - r_{0\cdot}^{(K)}\hat{\gamma}^{\text{scm}})\}'\{\lambda V^{(K)} E^{(K)} V^{(K)}{'} (r_{1\cdot}^{(K)} - r_{0\cdot}^{(K)}\hat{\gamma}^{\text{scm}})\} \\ 
    &= \{\lambda E^{(K)}V^{(K)}(r_{1\cdot}^{(K)} - r_{0\cdot}^{(K)}\hat{\gamma}^{\text{scm}}) \}' \{ \lambda E^{(K)}V^{(K)}(r_{1\cdot}^{(K)} - r_{0\cdot}^{(K)}\hat{\gamma}^{\text{scm}}) \} \\ 
    &=
    \| \lambda E^{(K)}V^{(K)}(r_{1\cdot}^{(K)} - r_{0\cdot}^{(K)}\hat{\gamma}^{\text{scm}}) \|_2^2.
    \label{eq:aug_fit_transform}
\end{align}
By the sub-multiplicativity of the Frobenius norm, 
\begin{align}
    \|\lambda E^{(K)} V^{(K)} (r_{1\cdot}^{(K)} - r_{0\cdot}^{(K)}\hat{\gamma}^{\text{scm}})\|_2^2 
    &\le 
    \|\lambda E^{(K)}\|_F^2 \|r_{1\cdot}^{(K)} -r_{0\cdot}^{(K)}{'}\hat{\gamma}^{\text{scm}}\|_2^2 \\
    & =
    \sum_{j=1}^{m(K)}\left(\frac{\lambda}{(d_j^{(K)})^2 + \lambda} \right)^2  \|r_{1\cdot}^{(K)} -r_{0\cdot}^{(K)}{'}\hat{\gamma}^{\text{scm}}\|_2^2 \\
    & \le 
    m(K)\left(\frac{\lambda}{(d_{\text{min}}^{(K)})^2 + \lambda}\right)^2\|r_{1\cdot}^{(K)} -r_{0\cdot}^{(K)}{'}\hat{\gamma}^{\text{scm}}\|_2^2. 
    \label{eq:norm_bound_1}
\end{align}
Combining \eqref{eq:aug_fit_transform} and \eqref{eq:norm_bound_1}, we have 
\begin{equation}
    \| r_{1\cdot}^{(K)} - r_{0\cdot}^{(K)}{'}\hat{\gamma}^{\mathrm{aug}(K)}\|_2
    \le 
    \frac{\sqrt{m(K)}\lambda}{(d_{\text{min}}^{(K)})^2 + \lambda}\|r_{1\cdot}^{(K)} -r_{0\cdot}^{(K)}{'}\hat{\gamma}^{\text{scm}}\|_2.
    \label{eq:fit_aug_bound_1}
\end{equation}

Next, we bound the second term in \eqref{eq:fit_aug_decompose}. Part (b) implies that there exists a constant $B_1 > 0$ such that $\|\hat{\gamma}^{\text{aug}(K)}\|_2 \le B_1$ for any $K$. Then, for each $k > K+1$, the Cauchy–Schwarz inequality gives 
\begin{align}
    \sum_{t=1}^{T_0}\left(r_{1,t,k} - \sum_{i=2}^N \hat{\gamma}_i^{\text{aug}(K)}r_{i,t,k}\right)^2
    &=
    \sum_{t=1}^{T_0}\left\{\sum_{i=2}^N\hat{\gamma}_i^{\text{aug}(K)}(r_{1,t,k} - r_{i,t,k})\right\}^2 \\ 
    & \le \sum_{t=1}^{T_0} \|\hat{\gamma}^{\text{aug}(K)}\|_2^2 \sum_{i=2}^N (r_{1,t,k} - r_{i,t,k})^2 \\ 
    & \le 
    B_1^2 \sum_{t=1}^{T_0} \sum_{i=2}^N (r_{1,t,k} - r_{i,t,k})^2.
\end{align}
Hence,  it follows that 
\begin{align}
    \sqrt{\sum_{t=1}^{T_0} \sum_{k=K+1}^\infty \left(r_{1,t,k} - \sum_{i=2}^N \hat{\gamma}_i^{\mathrm{aug}(K)}r_{i,t,k} \right)^2}
    \le 
    B_1
    \sqrt{\sum_{k=K+1}^\infty \sum_{t=1}^{T_0} \sum_{i=2}^N (r_{1,t,k} - r_{i,t,k})^2}.
    \label{eq:norm_bound_2}
\end{align}
Combining \eqref{eq:fit_aug_decompose}, \eqref{eq:fit_aug_bound_1} and \eqref{eq:norm_bound_2}, it holds that  
\begin{align}
    \sqrt{\sum_{t=1}^{T_0} \left\|Y_{1t} - \sum_{i=2}^N \gamma_i^{\mathrm{aug}(K)}Y_{it}\right\|_{\mathcal{H}}^2}
    &\le 
    \frac{\sqrt{m(K)}\lambda}{(d_{\text{min}}^{(K)})^2 + \lambda}\|r_{1\cdot}^{(K)} -r_{0\cdot}^{(K)}{'}\hat{\gamma}^{\text{scm}}\|_2 \\
    & \phantom{\le}
    +
    B_1
    \sqrt{\sum_{k=K+1}^\infty \sum_{t=1}^{T_0} \sum_{i=2}^N (r_{1,t,k} - r_{i,t,k})^2} \\ 
    & = 
    \frac{\sqrt{m(K)}\lambda}{(d_{\text{min}}^{(K)})^2 + \lambda}
    \sqrt{\sum_{t=1}^{T_0} \left\|Y_{1t} - \sum_{i=2}^N \gamma_i^{\mathrm{scm}}Y_{it}\right\|_{\mathcal{H}}^2} + R_2^{(K)},
\end{align}
where 
\begin{align}
    R_2^{(K)}
&=
\frac{\sqrt{m(K)}\lambda}{(d_{\text{min}}^{(K)})^2 + \lambda}
\left\{
\|r_{1\cdot}^{(K)} -r_{0\cdot}^{(K)}{'}\hat{\gamma}^{\text{scm}}\|_2
-\sqrt{\sum_{t=1}^{T_0} \left\|Y_{1t} - \sum_{i=2}^N \gamma_i^{\mathrm{scm}}Y_{it}\right\|_{\mathcal{H}}^2}
\right\} \\
& \phantom{=}
+
B_1
    \sqrt{\sum_{k=K+1}^\infty \sum_{t=1}^{T_0} \sum_{i=2}^N (r_{1,t,k} - r_{i,t,k})^2}.
\end{align}
Note that under Assumption \ref{ass:singular_values}, 
\[
\frac{\sqrt{m(K)}\lambda}{d_{\text{min}}^{(K)} + \lambda}
\le 
\frac{\sqrt{N-1}\lambda}{c_1^2 + \lambda}, 
\]
Furthermore, as $K \to \infty$, 
\[
\|r_{1\cdot}^{(K)} -r_{0\cdot}^{(K)}{'}\hat{\gamma}^{\text{scm}}\|_2 \to 
\sqrt{\sum_{t=1}^{T_0} \left\|Y_{1t} - \sum_{i=2}^N \gamma_i^{\mathrm{scm}}Y_{it}\right\|_{\mathcal{H}}^2}, \quad 
\sum_{k=K+1}^\infty \sum_{t=1}^{T_0} \sum_{i=2}^N (r_{1,t,k} - r_{i,t,k})^2 \to 0.
\]
Hence, we have $R_2^{(K)} \to 0$ as $K \to \infty$.

\subsection{Proofs for Section \ref{sec:est_error}}
\paragraph*{Proof of Theorem \ref{thm:est_error_auto}.}
In the proof, we aim to bound the estimation error $\|{Y}_{1T}^N - \hat{Y}_{1T}^N\|_\mathcal{H}$ of the generic weighting estimator $\hat{Y}_{1T}^N$. 
Note that, by Lemma \ref{lem:proj}, we have $\|{Y}_{1T}^N - \tilde{Y}_{1T}^N\|_{\mathcal{H}} \le \|{Y}_{1T}^N - \hat{Y}_{1T}^N\|_{\mathcal{H}}$. 
Moreover, by the distance-preserving property of the map $\Psi$, we have $d({\nu}_{1T}^N, \tilde{\nu}_{1T}^N) = \|{Y}_{1T}^N - \tilde{Y}_{1T}^N\|_{\mathcal{H}}$, where  $\hat{\nu}_{1T}^N = \Psi^{-1}(\tilde{Y}_{1T}^N)$. 
Together, these imply that $d({\nu}_{1T}^N, \hat{\nu}_{1T}^N) \le \|{Y}_{1T}^N - \hat{Y}_{1T}^N\|_{\mathcal{H}}$.
Hence, once a bound on $\|{Y}_{1T}^N - \hat{Y}_{1T}^N\|_{\mathcal{H}}$ is established, a corresponding bound for $d({\nu}_{1T}^N, \hat{\nu}_{1T}^N)$ immediately follows. 
Recall that $\mathcal{H}$ is the space of real-valued squared integrable functions on the measure space $(\mathcal{X}, \mathcal{A}, \mu)$. 

The difference between $Y_{1T}^N$ and $\hat{Y}_{1T}^N$ can be decomposed as 
\[
Y_{1T}^N(x) - \hat{Y}_{1T}^N(x) 
    =
    \sum_{t=1}^{T_0} \left\langle \beta_t(x, \cdot), Y_{1t} - \sum_{i=2}^N\hat{\gamma}_i Y_{it} \right\rangle_{\mathcal{H}} + \varepsilon_{1T}(x) - \sum_{i=2}^N \hat{\gamma}_i \varepsilon_{iT}(x), \quad  x \in \mathcal{X}.
\]
This yields the bound
\begin{equation}
    \|Y_{1T}^N - \hat{Y}_{1T}^N\|_{\mathcal{H}}
    \le 
    \|\Delta_{1}\|_{\mathcal{H}} + \left\|\varepsilon_{1T} - \sum_{i=2}^N \hat{\gamma}_i \varepsilon_{iT}\right\|_{\mathcal{H}},
    \label{eq:auto_bound}
\end{equation}
where $\Delta_{1}: \mathcal{X} \to \mathbb{R}$ is defined by 
\begin{equation}
\Delta_{1}(x) = 
\sum_{t=1}^{T_0} \left\langle \beta_t(x, \cdot), Y_{1t} - \sum_{i=2}^N\hat{\gamma}_i Y_{it} \right\rangle_{\mathcal{H}}.
\label{eq:delta_1_definition}
\end{equation}

We now bound $\|\Delta_{1}\|_{\mathcal{H}}$. 
For each $x \in \mathcal{X}$, the Cauchy–Schwarz inequality gives
\[
|\Delta_1(x)| 
\le 
\sqrt{\sum_{t=1}^{T_0}\|\beta_t(x, \cdot)\|_{\mathcal{H}}^2}
\sqrt{\sum_{t=1}^{T_0}\left\|Y_{1t} - \sum_{i=2}^N \hat{\gamma}_iY_{it}\right\|_{\mathcal{H}}^2}.
\]
Hence, we obtain 
\begin{align}
    \|\Delta_{1}\|_{\mathcal{H}} 
    &= \sqrt{\int_\mathcal{X}|\Delta_{1}(x)|^2 d\mu(x)} \\
    & \le 
    \sqrt{\int_{\mathcal{X}} \sum_{t=1}^{T_0}\|\beta_t(x, \cdot)\|_{\mathcal{H}}^2 d\mu(x)} 
\sqrt{\sum_{t=1}^{T_0}\left\|Y_{1t} - \sum_{i=2}^N \hat{\gamma}_iY_{it}\right\|_{\mathcal{H}}^2} \\ 
&= 
\sqrt{\sum_{t=1}^{T_0}\|\beta_t\|_{\mathcal{H}\times \mathcal{H}}^2}
\sqrt{\sum_{t=1}^{T_0}\left\|Y_{1t} - \sum_{i=2}^N \hat{\gamma}_iY_{it}\right\|_{\mathcal{H}}^2}.
    \label{eq:bound_B}
\end{align}

Next, we bound $\|\varepsilon_{1T} - \sum_{i=2}^N \hat{\gamma}_i \varepsilon_{iT}\|_{\mathcal{H}}$. 
Note that $\varepsilon_{1T}, -\hat{\gamma}_2 \varepsilon_{2T}, \ldots, -\hat{\gamma}_N\varepsilon_{NT}$ are independent mean-zero random variables satisfying $\|\varepsilon_{1T}\|_\mathcal{H} \le \sigma$ and $\|-\hat{\gamma}_i\varepsilon_{iT}\|_\mathcal{H} \le |\hat{\gamma}_i|\sigma (i=2, \ldots, N)$, almost surely. Then, Theorem 1.2 in \cite{Pinelis1991inequalities} implies that 
\begin{align}
    P\left(\left\|\varepsilon_{1T} - \sum_{i=2}^N \hat{\gamma}_i\varepsilon_{iT}\right\|_{\mathcal{H}} \ge r\right) \le 2 \exp\left\{-\frac{r^2}{2\sigma^2(1+\|\hat{\gamma}\|_2^2)}\right\}
\end{align}
holds for any $r > 0$.
Setting $r = \delta \sigma (1+\|\hat{\gamma}\|_2)$ and using the inequality $1+\|\hat{\gamma}\|_2^2 \le (1+\|\hat{\gamma}\|_2)^2$, we obtain 
\begin{equation}
    \left\|\varepsilon_{1T} - \sum_{i=2}^N \hat{\gamma}_i\varepsilon_{iT}\right\|_{\mathcal{H}}
    \le \delta \sigma (1+\|\hat{\gamma}\|_2)
    \label{eq:error_term_bound}
\end{equation}
with probability at least $1 - 2e^{-\delta^2/2}$. 
Combining
\eqref{eq:auto_bound}, \eqref{eq:bound_B} and 
\eqref{eq:error_term_bound}, we obtain the desired result. 

\paragraph*{Proof of Corollary \ref{cor:error_auto}.}
For any positive integer $K$ and $\delta > 0$,
Theorem \ref{thm:est_error_auto} implies
\begin{align}
    d(\nu_{1T}^N, \hat{\nu}_{1T}^{\text{aug}(K)})
        &\le \sqrt{\sum_{t=1}^{T_0} \|\beta_t\|_{\mathcal{H} \times \mathcal{H}}^2} \sqrt{\sum_{t=1}^{T_0} \left\|Y_{1t} - \sum_{i=2}^N \hat{\gamma}_i^{\text{aug}(K)}Y_{it}\right\|_{\mathcal{H}}^2} 
        +
        \delta \sigma(1+\|\hat{\gamma}^{\text{aug}(K)}\|_2)
\end{align}
with probability at least $1  - 2e^{-\delta^2/2}$. 
Combining this inequality with Lemma \ref{lem:aug_weight_prop}, we obtain
\begin{align}
         d(\nu_{1T}^N, \hat{\nu}_{1T}^{\text{aug}(K)})
         & \le \frac{\sqrt{m(K)}\lambda}{(d_{\mathrm{min}}^{(K)})^2 + \lambda}\sqrt{\sum_{t=1}^{T_0} \|\beta_t\|_{\mathcal{H} \times \mathcal{H}}^2} \sqrt{\sum_{t=1}^{T_0} \left\|Y_{1t} - \sum_{i=2}^N \hat{\gamma}_i^{\mathrm{scm}}Y_{it}\right\|_{\mathcal{H}}^2} \\ 
        & \phantom{\le}+ \delta \sigma \left\{ 1+
        \|\hat{\gamma}^{\mathrm{scm}}\|_2 +
        \frac{\sqrt{m(K)}d_{\mathrm{max}}^{(K)}}{(d_{\mathrm{min}}^{(K)})^2 + \lambda}\sqrt{\sum_{t=1}^{T_0} \left\|Y_{1t} - \sum_{i=2}^N \hat{\gamma}_i^{\mathrm{scm}}Y_{it}\right\|_{\mathcal{H}}^2}\right\} \\ 
        & \phantom{\le}+ R_3^{(K)}
    \end{align}
with probability at least $1  - 2e^{-\delta^2/2}$. 
Here, 
\[
R_3^{(K)}
=
\sqrt{\sum_{t=1}^{T_0} \|\beta_t\|_{\mathcal{H} \times \mathcal{H}}^2}R_1^{(K)}
+
\delta \sigma R_2^{(K)},
\]
and since $R_1^{(K)} \to 0$ and $R_2^{(K)} \to 0$ as $K \to \infty$, it follows that $R_3^{(K)} \to 0$ as $K \to \infty$.

\paragraph*{Proof of Theorem \ref{thm:est_error_factor}.}
We begin by stating a lemma that expresses the vector of factor loadings in terms of post-treatment outcomes and errors.

\begin{lem}
    Suppose Assumption \ref{ass_latent} holds. Then, for each $i=1, \ldots, N$, 
    \[
    \phi_i  =  \{\mu(x)'\mu(x)\}^{-1}\mu(x)'\{Y_{i\cdot}(x) - \varepsilon_{i\cdot}(x)\}
\]
for any $x \in \mathcal{X}$,
where $Y_{i\cdot}(x) = (Y_{i1}(x), \ldots, Y_{iT_0}(x))' \in  \mathbb{R}^{T_0}$ and $\varepsilon_{i\cdot}(x) = (\varepsilon_{i1}(x), \ldots, \varepsilon_{iT_0}(x))' \in  \mathbb{R}^{T_0}$. 
\label{lem:factor_expression}
\end{lem}
\begin{proof}
    Under the latent factor model \eqref{eq:latent_factor}, we have 
    $
        Y_{i\cdot}(x) = \mu(x)\phi_i + \varepsilon_{i\cdot}(x). 
    $
    Multiplying both sides by $\mu(x)'$ gives $\mu(x)'Y_{i\cdot}(x) = \mu(x)'\mu(x)\phi_i + \mu(x)'\varepsilon_{i\cdot}(x)$, which immediately gives the desired result. 
\end{proof}

We now prove Theorem \ref{thm:est_error_factor}. As in the proof of Theorem \ref{thm:est_error_auto}, we aim to bound $\|{Y}_{1T}^N - \hat{Y}_{1T}^N\|_\mathcal{H}$. 
Under the latent factor model \eqref{eq:latent_factor} and applying Lemma \ref{lem:factor_expression}, we can decompose the difference between
$Y_{1T}^N$ and $\hat{Y}_{1T}^N$ as 
\begin{align}
Y_{1T}^N(x) - \hat{Y}_{1T}^N(x)
&=
\left\{\phi_1(x) - \sum_{i=2}^N \hat{\gamma}_i \phi_i(x)\right\}'\mu_T(x)
+
\varepsilon_{1T}(x) - \sum_{i=2}^N \hat{\gamma}_i \varepsilon_{iT}(x) \\
& =
\left\{Y_{1\cdot}(x) - \sum_{i=2}^N \hat{\gamma}_i Y_{i\cdot}(x)\right\}{'} \mu(x)\{\mu(x)'\mu(x)\}^{-1}  \mu_T(x) \\
& \phantom{=}
+
\left\{\varepsilon_{1\cdot}(x) - \sum_{i=2}^N \hat{\gamma}_i \varepsilon_{i\cdot}(x)\right\}{'} \mu(x)\{\mu(x)'\mu(x)\}^{-1} \mu_T(x) 
+
\varepsilon_{1T}(x) - \sum_{i=2}^N \hat{\gamma}_i \varepsilon_{iT}(x). 
\end{align}
This yields 
\begin{equation}
    \|Y_{1T}^N - \hat{Y}_{1T}^N\|_{\mathcal{H}}
    \le \|\Delta_{2}\|_{\mathcal{H}} + \|\Delta_{3}\|_{\mathcal{H}} + \left\|\varepsilon_{1T} - \sum_{i=2}^N \hat{\gamma}_i \varepsilon_{iT}\right\|_{\mathcal{H}},
    \label{eq:error_bound_latent}
\end{equation}
where $\Delta_{2}, \Delta_{3}: \mathcal{X} \to \mathbb{R}$ are defined by 
\begin{gather}
    \Delta_{2}(x)
    =
    \left\{Y_{1\cdot}(x) - \sum_{i=2}^N \hat{\gamma}_i Y_{i\cdot}(x)\right\}{'} \mu(x) \{\mu(x)'\mu(x)\}^{-1} \mu_T(x), \\ 
    \Delta_{3}(x)
    =
    \left\{\varepsilon_{1\cdot}(x) - \sum_{i=2}^N \hat{\gamma}_i \varepsilon_{i\cdot}(x)\right\}{'} \mu(x)\{\mu(x)'\mu(x)\}^{-1}  \mu_T(x).
    \label{eq:delta_23_definition}
\end{gather}

We bound $\|\Delta_{2}\|_{\mathcal{H}}$. 
Fix the pre-treatment errors $\varepsilon_{it}, i=1, \ldots, N, t=1, \ldots, T_0$. 
Under Assumption \ref{ass_latent}, the maximum eigenvalue of the $J \times J$ matrix $\{T_0\mu(x)'\mu(x)\}^{-1}$ is $\xi_{\text{min}}^{-1}(x)$, which is bounded above by $M_2^{-1}$. This implies that $\|\{\mu(x)'\mu(x)\}^{-1}\|_{F} \le \sqrt{J}/(T_0M_2)$.
Applying the Cauchy–Schwarz inequality and the sub-multiplicativity of the Frobenius norm, 
\begin{align}
    |\Delta_{2}(x)| &\le \left\|Y_{1\cdot}(x) - \sum_{i=2}^N \hat{\gamma}_i Y_{i\cdot}(x)\right\|_2\|\mu(x)\{\mu(x)'\mu(x)\}^{-1}\mu_T(x)\|_2 \\ 
    & \le 
     \left\|Y_{1\cdot}(x) - \sum_{i=2}^N \hat{\gamma}_i Y_{i\cdot}(x)\right\|_2 \|\mu(x)\|_F \|\{\mu(x)'\mu(x)\}^{-1}\|_F \|\mu_T(x)\|_2 \\ 
     & \le 
     \left\|Y_{1\cdot}(x) - \sum_{i=2}^N \hat{\gamma}_i Y_{i\cdot}(x)\right\|_2 \sqrt{T_0J}M_1\cdot \frac{\sqrt{J}}{T_0M_2} \cdot \sqrt{J}M_1  \\ 
     &=
      \frac{M_1^2J^{3/2}}{M_2\sqrt{T_0}} \left\|Y_{1\cdot}(x) - \sum_{i=2}^N \hat{\gamma}_i Y_{i\cdot}(x)\right\|_2.
\end{align}
so that
\begin{align}
    \|\Delta_{2}\|_{\mathcal{H}} \le \frac{M_1^2J^{3/2}}{M_2\sqrt{T_0}} \sqrt{\int_\mathcal{X} \left\|Y_{1\cdot}(x) - \sum_{i=2}^N \hat{\gamma}_i Y_{i\cdot}(x)\right\|_2^2 d\mu(x)}
    =
    \frac{M_1^2J^{3/2}}{M_2\sqrt{T_0}} \sqrt{\sum_{t=1}^{T_0} \left\|Y_{1t} - \sum_{i=2}^N \hat{\gamma}_iY_{it}\right\|_{\mathcal{H}}^2}.
    \label{eq:bound_B_21}
\end{align}

Next we bound $\|\Delta_{3}\|_{\mathcal{H}}$. Similarly, 
\begin{equation}
    \|\Delta_{3}\|_{\mathcal{H}}
    \le \frac{M_1^2J^{3/2}}{M_2\sqrt{T_0}}\sqrt{\sum_{t=1}^{T_0} \left\|\varepsilon_{1t} - \sum_{i=2}^N \hat{\gamma}_i\varepsilon_{it}\right\|_{\mathcal{H}}^2}.
\end{equation}
For each $t = 1, \ldots, T_0$, applying the triangle inequality, 
\begin{align}
    \left\|\varepsilon_{1t} - \sum_{i=2}^N \hat{\gamma}_i \varepsilon_{it}\right\|_{\mathcal{H}}
    &=
    \left\| \sum_{i=2}^N \hat{\gamma}_i(\varepsilon_{1t} - \varepsilon_{it}) \right\|_{\mathcal{H}} 
    \le 
    \sum_{i=2}^N |\hat{\gamma}_i|\|\varepsilon_{1t} - \varepsilon_{it}\|_{\mathcal{H}} \\ 
    &
    \le 2\sigma \sum_{i=2}^N |\hat{\gamma}_i| 
    = 2\sigma \|\hat{\gamma}\|_1.
\end{align}
Hence, 
\begin{equation}
    \|\Delta_{3}\|_{\mathcal{H}} \le  
    \frac{2\sigma M_1^2J^{3/2}}{M_2\sqrt{T_0}}\|\hat{\gamma}\|_1.
    \label{eq:bound_B_22}
\end{equation}

Since the post-treatment errors $\varepsilon_{iT}, i=1, \ldots, N$ are independent of the pre-treatment errors,  we have (cf. Theorem \ref{thm:est_error_auto}),
\begin{equation}
    \left\|\varepsilon_{1T} - \sum_{i=2}^N \hat{\gamma}_i\varepsilon_{iT}\right\|_{\mathcal{H}}
    \le \delta \sigma (1+\|\hat{\gamma}\|_2)
\label{eq:bound_error_latent}
\end{equation}
with probability at least $1 - 2e^{-\delta^2/2}$. 

Putting together \eqref{eq:error_bound_latent}, \eqref{eq:bound_B_21}, \eqref{eq:bound_B_22} and \eqref{eq:bound_error_latent}, 
we have 
\begin{align}
    \|Y_{1T}^N - \hat{Y}_{1T}^N \|_{\mathcal{H}} 
        &\le \frac{M_1^2J^{3/2}}{M_2\sqrt{T_0}} \sqrt{\sum_{t=1}^{T_0} \left\|Y_{1t} - \sum_{i=2}^N \hat{\gamma}_iY_{it}\right\|_{\mathcal{H}}^2} 
        +
        \frac{2\sigma M_1^2J^{3/2}}{M_2}\|\hat{\gamma}\|_1
        +
        \delta \sigma(1+\|\hat{\gamma}\|_2)
\end{align}
with probability at least $1 - 2e^{-\delta^2/2}$ conditionally on the pre-treatment errors. 
By applying the law of iterated expectations, we obtain the desired result.

\paragraph*{Proof of Corollary \ref{cor:est_error_factor}.}
For any positive integer $K$ and $\delta > 0$, 
Theorem \ref{thm:est_error_factor} and the inequality $\|\hat{\gamma}^{\text{aug}(K)}\|_1 \le \sqrt{N-1}\|\hat{\gamma}^{\text{aug}(K)}\|_2$ imply
\begin{align}
    d(\nu_{1T}^N, \hat{\nu}_{1T}^{N, \text{aug}(K)}) 
        &\le \frac{M_1^2J^{3/2}}{M_2\sqrt{T_0}} \sqrt{\sum_{t=1}^{T_0} \left\|Y_{1t} - \sum_{i=2}^N \hat{\gamma}_iY_{it}^{N, \text{aug}(K)}\right\|_{\mathcal{H}}^2} \\
       & \phantom{\le}
        +
        \sigma\left(\frac{2\sqrt{N-1}M_1^2J^{3/2}}{M_2} + \delta\right)\|\hat{\gamma}^{\text{aug}(K)}\|_2
        +
        \delta \sigma
\end{align}
with probability at least $1 - 2e^{-\delta^2/2}$. Combining this result with Lemma \ref{lem:aug_weight_prop}, we obtain
\begin{align}
     &\|Y_{1T}^N - \hat{Y}_{1T}^{N, \mathrm{aug}(K)} \|_{\mathcal{H}} \\
    &\le \frac{M_1^2J^{3/2}\sqrt{m(K)}\lambda}{M_2\sqrt{T_0}\{(d_{\mathrm{min}}^{(K)})^2 + \lambda\}} \sqrt{\sum_{t=1}^{T_0} \left\|Y_{1t} - \sum_{i=2}^N \hat{\gamma}_i^{\mathrm{scm}}Y_{it}\right\|_{\mathcal{H}}^2} \\ 
    &\phantom{\le}+
    \sigma\left(\frac{2\sqrt{N-1}M_1^2J^{3/2}}{M_2}+\delta\right) \left\{
        \|\hat{\gamma}^{\mathrm{scm}}\|_2 +
        \frac{\sqrt{m(K)}d_{\mathrm{max}}^{(K)}}{(d_{\mathrm{min}}^{(K)})^2 + \lambda}\sqrt{\sum_{t=1}^{T_0} \left\|Y_{1t} - \sum_{i=2}^N \hat{\gamma}_i^{\mathrm{scm}}Y_{it}\right\|_{\mathcal{H}}^2}\right\} \\ 
        &\phantom{\le}+
        \delta \sigma + R_4^{(K)}
\end{align}
with probability at least $1 - 2e^{-\delta^2/2}$. 
Here, 
\begin{align}
    R_4^{(K)} = \frac{M_1^2J^{3/2}}{M_2\sqrt{T_0}} R_1^{(K)} 
    +
     \sigma\left(\frac{2\sqrt{N-1}M_1^2J^{3/2}}{M_2}+\delta\right)R_2^{(K)},
\end{align}
and since $R_1^{(K)} \to 0$ and $R_2^{(K)} \to 0$ as $K \to \infty$, it follows that $R_4^{(K)} \to 0$ as $K \to \infty$.

\subsection{Proofs for Section \ref{sec:projection}}
\paragraph*{Proof of Lemma \ref{lem:proj}.}
For (a), we first show the existence of a minimizer. Define $r := \inf_{h \in C} \|y-h\|_\mathcal{H}<\infty$. %Since $C$ is a nonempty set, $r$ is finite. %Indeed, for any $h \in C$, $r \leq \|y-h\|_\mathcal{H} <\infty$. 
Choose a sequence $\{h_n\}_{n \ge 1} \subset C$ such that $\|y-h_n\|_\mathcal{H} \to r$ as $n \to \infty$. Note that for sufficiently large $n$, we have 
\[
\|h_n\|_\mathcal{H} \leq \|y\|_\mathcal{H} + \|y-h_n\|_\mathcal{H} \leq \|y\|_\mathcal{H} + (r+1)<\infty
\]
and this implies that $\{h_n\}_{n\ge 1}$ has a convergent subsequence $h_{n_k} \to h^\ast \in \mathcal{H}$. Since $C$ is closed, $h^\ast \in C$ and $\|y-h^*\|_\mathcal{H} = \lim_{k\to\infty}\|y-h_{n_k}\|_\mathcal{H} = r$.
This yields that $h^*$ a minimizer, that is, $h^\ast \in \argmin_{h \in C}\|y-h\|_\mathcal{H}$. 

Next, we show the uniqueness of the minimizer. Consider the function $f(h) = \|y-h\|_\mathcal{H}^2$. Note that $f$ is convex on $\mathcal{H}$ since for any $h_1, h_2 \in \mathcal{H}$ and $t \in (0,1)$, one can see
\[
f(th_1+(1-t)h_2) = \|t(y-h_1)+(1-t)(y-h_2)\|_\mathcal{H}^2 \leq t\|y-h_1\|_\mathcal{H}^2 + (1-t)\|y-h_2\|_\mathcal{H}^2.
\]
One can also see that in the above inequality, the equality holds if and only if $h_1 = h_2$. Suppose that $h_1^\ast, h_2^\ast \in C$ are minimizers such that $h_1^\ast \neq h_2^\ast$. From the convexity of $C$, $h_0^\ast=(h_1^\ast+h_2^\ast)/2 \in C$ and the convexity of $f$ yields 
\[
\|y-h_0^\ast\|_\mathcal{H}^2 = f(h_0^\ast) < {1 \over 2}f(h_1^\ast) + {1 \over 2}f(h_2^\ast) = {1 \over 2} \|y-h_1^\ast\|_\mathcal{H}^2 + {1 \over 2} \|y-h_2^\ast\|_\mathcal{H}^2 = r^2. 
\]
This contradicts the fact that $r$ is the minimum value of $\|y-h\|_\mathcal{H}$ over $h \in C$. Therefore, the minimizer is unique.

For (b), we use the following result that will be shown after the proof of (b). 
\begin{lem}\label{lem:inn-prod}
Let $C\subset\mathcal{H}$ be a nonempty closed convex set and $y \in \mathcal{H}$. Then for every $h\in C$,
\[
\langle y-\pi_C(y),\, h-\pi_C(y)\rangle_\mathcal{H} \leq 0.
\]
\end{lem}
From Lemma \ref{lem:inn-prod}, we have
\begin{align}
    \langle u-\pi_C(u), \pi_C(v)-\pi_C(u)\rangle_\mathcal{H} &\leq 0,
    \label{eq:proj-1}\\
    \langle v-\pi_C(v), \pi_C(u)-\pi_C(v)\rangle_\mathcal{H}  &\le 0.
    \label{eq:proj-2}
\end{align}
Combining \eqref{eq:proj-1} and \eqref{eq:proj-2}, we have
\begin{align*}
0&\geq \langle u-\pi_C(u), \pi_C(v)-\pi_C(u)\rangle_\mathcal{H} + \langle v-\pi_C(v), \pi_C(u)-\pi_C(v)\rangle_\mathcal{H}\\
&= -\langle u-\pi_C(u), \pi_C(u)-\pi_C(v)\rangle_\mathcal{H} + \langle v-\pi_C(v), \pi_C(u)-\pi_C(v)\rangle_\mathcal{H}\\
&= \langle \pi_C(u)-u, \pi_C(u)-\pi_C(v)\rangle_\mathcal{H} + \langle v-\pi_C(v), \pi_C(u)-\pi_C(v)\rangle_\mathcal{H}\\
&= \langle \{\pi_C(u) - \pi_C(v)\}-(u-v), \pi_C(u)-\pi_C(v)\rangle_\mathcal{H}. 
\end{align*}
Hence we have
\[
\langle u-v,\, \pi_C(u)-\pi_C(v)\rangle_\mathcal{H} \geq \|\pi_C(u)-\pi_C(v)\|_\mathcal{H}^2.
\]
By Cauchy-Schwarz inequality,
\[
\langle u-v,\, \pi_C(u)-\pi_C(v)\rangle_\mathcal{H} \leq \|u-v\|_\mathcal{H} \|\pi_C(u)-\pi_C(v)\|_\mathcal{H}
\]
and then we have
\[
\|\pi_C(u)-\pi_C(v)\|_\mathcal{H}^2 \leq \|u-v\|_\mathcal{H} \|\pi_C(u)-\pi_C(v)\|_\mathcal{H}.
\]
If $\pi_C(u)\neq \pi_C(v)$, we obtain
\[
\|\pi_C(u)-\pi_C(v)\|_\mathcal{H} \leq \|u-v\|_\mathcal{H}
\]
and the above inequality is trivial if $\pi_C(u)=\pi_C(v)$. Therefore, we obtain the desired result.

Now we show Lemma \ref{lem:inn-prod}. Note that $\pi_C(y) \in C$. Fix $h_0 \in C$ and for $t\in (0,1)$, define $\gamma(t)=(1-t)\pi_C(y) + t h_0\in C$. Then we have
\[
\|y-\pi_C(y)\|_\mathcal{H}^2 \leq \|y-\gamma(t)\|_\mathcal{H}^2 = \|y-\pi_C(y)\|_\mathcal{H}^2 - 2t\langle y-\pi_C(y),\, h_0 -\pi_C(y)\rangle_\mathcal{H} + t^2\|h_0-\pi_C(y)\|_\mathcal{H}^2.
\]
This yields
\[
0 \leq - 2t\langle y-\pi_C(y),\, h_0 -\pi_C(y)\rangle_\mathcal{H} + t^2\|h_0-\pi_C(y)\|_\mathcal{H}^2. 
\]
Recall $t >0$. Therefore, we have
\[
0 \le -2\langle y-\pi_C(y),\, h_0-\pi_C(y)\rangle_\mathcal{H} + t\|h_0-\pi_C(y)\|_\mathcal{H}^2.
\]
Letting $t\downarrow0$ yields $\langle y-\pi_C(y),\, h_0-\pi_C(y)\rangle_\mathcal{H} \leq 0$.

\subsection{Proofs for Section \ref{sec:covariates}}
\paragraph*{Proof of Lemma \ref{lem:weighting_expression_cov}.}
Fix $t \in \{T_0+1, \ldots, T\}$.
We denote $\hat{\theta}(x) = (\hat{\theta}_1(x)', \ldots, \hat{\theta}_{T_0}(x)')' \in \mathbb{R}^{KT_0}$. For each $x \in \mathcal{X}$, the solution to the optimization problem \eqref{eq:regression_prob_cov} is given by 
\begin{gather}
    \hat{\eta}_0(x) = \frac{1}{N-1} \sum_{i=2}^N Y_{it}(x), \\
    \hat{\theta}(x) = (\check{r}_{0\cdot}'\check{r}_{0\cdot} + \lambda I_{KT_0})^{-1} \check{r}_{0\cdot}' Y_{0t}(x), \\
    \hat{\delta}(x) = (Z_0'Z_0)^{-1}Z_0'\{Y_{0t}(x) - \check{r}_{0\cdot} \hat{\theta}(x)\},
    \label{eq:closed_form_solution_cov}
\end{gather}
where $Y_{0t}(x) = (Y_{2t}(x), \ldots, Y_{Nt}(x))' \in \mathbb{R}^{N-1}$.
By definition, 
\begin{align}
    \hat{Y}_{1t}^{N, \text{cov}}(x)
    =
    (\hat{\gamma}^{\text{scm}}){'}Y_{0t}(x)
    +
    (r_{1\cdot} - r_{0\cdot}'\hat{\gamma}^{\text{scm}})'\hat{\theta}(x)
    +
    (Z_1 - Z_0'\hat{\gamma}^{\text{scm}})'\hat{\delta}(x).
    \label{eq:augmented_estimator_decomposition_cov}
\end{align}
Since $\check{r}_{1\cdot} = r_{1\cdot} - r_{0\cdot}'Z_{0}'(Z_0'Z_0)^{-1}Z_1$ and $\check{r}_{0\cdot} = r_{0\cdot} - Z_0 (Z_0'Z_0)^{-1}Z_0'r_{0\cdot}$, we have
\begin{align}
    {r}_{1\cdot} - r_{0\cdot}'\hat{\gamma}^{\text{scm}}
    &=
    \check{r}_{1\cdot} + {r}_{0\cdot}'Z_0'(Z_0'Z_0)^{-1} Z_1
    -
    \{\check{r}_{0\cdot} + Z_0(Z_0'Z_0)^{-1}Z_{0}'r_{0\cdot}\}' \hat{\gamma}^{\text{scm}} \\ 
    &= 
    \check{r}_{1\cdot}  - \check{r}_{0\cdot}'\hat{\gamma}^{\text{scm}}
    +
    r_{0\cdot}'Z_0 (Z_0'Z_0)^{-1}(Z_1 - Z_0' \hat{\gamma}^{\text{scm}}),
\end{align}
which implies
\begin{equation}
    ({r}_{1\cdot} - r_{0\cdot}'\hat{\gamma}^{\text{scm}})'\hat{\theta}(x) = 
    (\check{r}_{1\cdot} - \check{r}_{0\cdot}'\hat{\gamma}^{\text{scm}})'\hat{\theta}(x) 
    +
    (Z_1 - Z_0'\hat{\gamma}^{\text{scm}})' (Z_0'Z_0)^{-1}Z_0'\check{r}_{0\cdot}\hat{\theta}(x).
    \label{eq:augmented_estimator_decomposition_cov_part1}
\end{equation}
Moreover, from \eqref{eq:closed_form_solution_cov}, 
\begin{align}
    (Z_1 - Z_0'\hat{\gamma}^{\text{scm}})'\hat{\delta}(x)
    &=
    (Z_1 - Z_0'\hat{\gamma}^{\text{scm}})'(Z_0'Z_0)^{-1}Z_0' Y_{0t}(x)   \\
    & \phantom{=}
    -
    (Z_1 - Z_0'\hat{\gamma}^{\text{scm}})'(Z_0'Z_0)^{-1}Z_0'\check{r}_{0\cdot} \hat{\theta}(x). 
    \label{eq:augmented_estimator_decomposition_cov_part2}
\end{align}
Combining \eqref{eq:augmented_estimator_decomposition_cov}, \eqref{eq:augmented_estimator_decomposition_cov_part1} and \eqref{eq:augmented_estimator_decomposition_cov_part2}, we obtain 
\begin{align}
    \hat{Y}_{1t}^{N, \text{cov}}(x)
    &=
     (\hat{\gamma}^{\text{scm}}){'}Y_{0t}(x)
    +
    (\check{r}_{1\cdot} - \check{r}_{0\cdot}'\hat{\gamma}^{\text{scm}})'\hat{\theta}(x) 
    +
    (Z_1 - Z_0'\hat{\gamma}^{\text{scm}})'(Z_0'Z_0)^{-1}Z_0' Y_{0t}(x) \\ 
    &= 
    (\hat{\gamma}^{\text{scm}}){'}Y_{0t}(x)
    +
     (\check{r}_{1\cdot} - \check{r}_{0\cdot}'\hat{\gamma}^{\text{scm}})'(\check{r}_{0\cdot}'\check{r}_{0\cdot} + \lambda I_{KT_0})^{-1} \check{r}_{0\cdot}' Y_{0t}(x) \\
     &\phantom{=}
     +
     (Z_1 - Z_0'\hat{\gamma}^{\text{scm}})'(Z_0'Z_0)^{-1}Z_0' Y_{0t}(x). 
\end{align}
This completes the proof of the desired result.

\paragraph*{Proof of Lemma \ref{lem:cov_weight_prop}.}
For notational convenience, we often omit the subscript $K$ throughout the proof.
For instance, we write $r_{0\cdot}^{(K)}$ and $\check{r}_{0\cdot}^{(K)}$ simply as $r_{0\cdot}$ and $\check{r}_{0\cdot}$, respectively.
For part (a), note that for any $K$, 
\[
Z_0' \check{r}_{0\cdot} 
=
Z_0'\{r_{0\cdot}^{} - Z_0(Z_0'Z_0)^{-1}Z_0'r_{0\cdot}^{}\} =
Z_0'r_{0\cdot}^{} - (Z_0'Z_0)(Z_0'Z_0)^{-1}Z_0'r_{0\cdot}^{}
=
0.
\]
Hence, 
\begin{align}
    &Z_1 - Z_0'\hat{\gamma}^{\text{cov}(K)} \\
    &=
    Z_1  
    - Z_0'\{\hat{\gamma}^{\text{scm}} + \check{r}_{0\cdot}^{}(\check{r}_{0\cdot}'\check{r}_{0\cdot}^{} + \lambda I_{KT_0})^{-1}(\check{r}_{1\cdot}^{}- \check{r}_{0\cdot}^{}\hat{\gamma}^{\text{scm}}) +Z_0(Z_0'Z_0)^{-1}(Z_1 - Z_0'\hat{\gamma}^{\text{scm}})\} \\ 
    &= Z_1  
    - Z_0' \hat{\gamma}^{\text{scm}} 
    -
    Z_0'Z_0(Z_0'Z_0)^{-1}(Z_1 - Z_0'\hat{\gamma}^{\text{scm}}) \\ 
    &= 0. 
\end{align}
This implies that $Z_1 - \sum_{i=2}^N \hat{\gamma}_i^{\text{cov}(K)}Z_i = 0$.

We next prove part (c). Fix an arbitrary positive integer $K$. Lemma \ref{lem:weighting_expression_cov} implies that
\begin{equation}
    \hat{\gamma}^{\text{cov}(K)}
    =
    \hat{\gamma}^{\text{scm}} 
    +
    \check{r}_{0\cdot}(\check{r}_{0\cdot}'\check{r}_{0\cdot} + \lambda I _{KT_0})^{-1}(\check{r}_{1\cdot} - \check{r}_{0\cdot}\hat{\gamma}^{\text{scm}}) + Z_0(Z_0'Z_0)^{-1}(Z_1 - Z_0'\hat{\gamma}^{\text{scm}}). 
\end{equation}
Let $\check{r}_{0\cdot} = \check{U}\check{D}\check{V}'$ be the singular value decomposition of the matrix $\check{r}_{0\cdot}$. 
Here, $\check{U}$ and $\check{V}$ are $(N-1) \times \check{m}(K)$ and $(KT_0) \times \check{m}(K)$ orthogonal matrices, respectively. 
$\check{D} = \text{diag}\{\check{d}_1^{(K)}, \ldots, \check{d}_{\check{m}(K)}^{(K)}\}$ is an $\check{m}(K) \times \check{m}(K)$ diagonal matrix, where $\check{d}_1^{(K)}, \ldots , \check{d}_{\check{m}(K)}^{(K)}$ are the singular values of $\check{r}_{0\cdot}$ such that 
$\check{d}_1^{(K)} \ge \ldots \ge \check{d}_{\check{m}(K)}^{(K)}$. 
In Section \ref{sec:covariates}, we denote $\check{d}_1^{(K)}$ and $\check{d}_{m(K)}^{(K)}$ as $\check{d}_{\text{max}}^{(K)}$ and $\check{d}_{\text{min}}^{(K)}$, respectively.
Analogously to the proof of Lemma \ref{lem:aug_weight_prop}, we can show that 
\begin{equation}
    (\check{r}_{0\cdot}'\check{r}_{0\cdot} + \lambda I_{K T_0})^{-1} = \check{V} \check{E} \check{V}',
    \label{eq:rtilde_inv}
\end{equation}
where 
\[
\check{E} = \text{diag}\left\{\left\{(\check{d}_1^{(K)})^2 + \lambda\right\}^{-1}, \ldots, \left\{(\check{d}_{\check{m}(K)}^{(K)})^2 + \lambda\right\}^{-1}\right\}.
\]
Hence, we have 
\begin{align}
    \hat{\gamma}^{\text{cov}(K)} 
    &= \hat{\gamma}^{\text{scm}} + (\check{U}\check{D}\check{V}')\check{V}\check{E}\check{V}'(\check{r}_{1\cdot} - \check{r}_{0\cdot}'\hat{\gamma}^{\text{scm}})
    +
    Z_0(Z_0'Z_0)^{-1}(Z_1 - Z_0'\hat{\gamma}^{\text{scm}}) \\ 
    & = 
    \hat{\gamma}^{\text{scm}} + 
    \check{U}\check{D} \check{E}\check{V}'(\check{r}_{1\cdot} - \check{r}_{0\cdot}'\hat{\gamma}^{\text{scm}})
    +
    Z_0(Z_0'Z_0)^{-1}(Z_1 - Z_0'\hat{\gamma}^{\text{scm}}).
\end{align}
By the  triangle inequality and the sub-multiplicativity of the Frobenius norm, 
\begin{align}
    \|\hat{\gamma}^{\text{cov}(K)}\|_2
    &\le 
    \|\hat{\gamma}^{\text{scm}}\|_2
    +
    \|\check{D}\check{E}\|_{F}\|\check{r}_{1\cdot} - \check{r}_{0\cdot}'\hat{\gamma}^{\text{scm}}\|_2
    +
    \|Z_0(Z_0'Z_0)^{-1}(Z_1 - Z_0'\hat{\gamma}^{\text{scm}})\|_2 \\ 
    & \le
    \|\hat{\gamma}^{\text{scm}}\|_2
    +
    \sqrt{\sum_{j=1}^{\check{m}(K)}\left(\frac{\check{d}_j^{(K)}}{(\check{d}_j^{(K)})^2 + \lambda}\right)^2}\|\check{r}_{1\cdot} - \check{r}_{0\cdot}'\hat{\gamma}^{\text{scm}}\|_2 +
    \|Z_0(Z_0'Z_0)^{-1}(Z_1 - Z_0'\hat{\gamma}^{\text{scm}})\|_2 \\ 
    & \le 
    \|\hat{\gamma}^{\text{scm}}\|_2
    +
    \frac{\sqrt{\check{m}(K)}\check{d}_{\text{max}}^{(K)}}{(\check{d}_{\text{min}}^{(K)})^2 + \lambda} \|\check{r}_{1\cdot}^{(K)} - \check{r}_{0\cdot}^{(K)}{'}\hat{\gamma}^{\text{scm}}\|_2 
    +
    \|Z_0(Z_0'Z_0)^{-1}(Z_1 - Z_0'\hat{\gamma}^{\text{scm}})\|_2.
    \label{eq:gamma_cov_bound}
\end{align}
Note that 
\begin{equation}
    \check{r}_{1\cdot} - \check{r}_{0\cdot}'\hat{\gamma}^{\text{scm}}
    =
    r_{1\cdot} - r_{0\cdot}'\hat{\gamma}^{\text{scm}} 
    -
    r_{0\cdot}' Z_0(Z_0'Z_0)^{-1}(Z_1 -  Z_0'\hat{\gamma}^{\text{scm}}) ,
\end{equation}
which yields
\begin{equation}
    \|\check{r}_{1\cdot}^{(K)} - \check{r}_{0\cdot}^{(K)}{'}\hat{\gamma}^{\text{scm}}\|_2 
    \le 
    \|r_{1\cdot}^{(K)} - r_{0\cdot}^{(K)}{'}\hat{\gamma}^{\text{scm}}\|_2
    +
    \|r_{0\cdot}^{(K)}\|_F\|Z_0(Z_0'Z_0)^{-1}(Z_1 -  Z_0'\hat{\gamma}^{\text{scm}})\|_2.
    \label{eq:eval_rtilde_scm}
\end{equation}
Combining \eqref{eq:gamma_cov_bound} and \eqref{eq:eval_rtilde_scm}, we obtain 
\[
\|\hat{\gamma}^{\mathrm{cov}(K)}\|_2 
        \le 
        F_2(\lambda)+R_6^{(K)},
\]
where 
\begin{align}
        F_2(\lambda)& = \|\hat{\gamma}^{\mathrm{scm}}\|_2 
        +
        \frac{\sqrt{\check{m}(K)}\check{d}_{\text{max}}^{(K)}}{(\check{d}_{\text{min}}^{(K)})^2 + \lambda} \sqrt{\sum_{t=1}^{T_0}\left\|Y_{1t} - \sum_{i=2}^N \hat{\gamma}_i^{\mathrm{scm}} Y_{it}\right\|_{\mathcal{H}}^2} \\
        & \phantom{\le}+ 
        \frac{\sqrt{\check{m}(K)}\check{d}_{\text{max}}^{(K)}}{(\check{d}_{\text{min}}^{(K)})^2 + \lambda} 
\sqrt{\sum_{t=1}^{T_0}\sum_{i=2}^N\|Y_{it}\|_{\mathcal{H}}^2}
        \|Z_0(Z_0'Z_0)^{-1}(Z_1 - Z_0'\hat{\gamma}^{\mathrm{scm}})\|_2 \\ 
        & \phantom{\le} + 
        \|Z_0(Z_0'Z_0)^{-1}(Z_1 - Z_0'\hat{\gamma}^{\mathrm{scm}})\|_2
    \end{align}
and 
\begin{align}
    R_6^{(K)}  
    &= \frac{\sqrt{\check{m}(K)}\check{d}_{\text{max}}^{(K)}}{(\check{d}_{\text{min}}^{(K)})^2 + \lambda}
    \left\{\|r_{1\cdot}^{(K)} - r_{0\cdot}^{(K)}{'}\hat{\gamma}^{\text{scm}}\|_2 -  \sqrt{\sum_{t=1}^{T_0}\left\|Y_{1t} - \sum_{i=2}^N \hat{\gamma}_i^{\mathrm{scm}} Y_{it}\right\|_{\mathcal{H}}^2} \right\} \\
    &\phantom{\le}+
    \frac{\sqrt{\check{m}(K)}\check{d}_{\text{max}}^{(K)}}{(\check{d}_{\text{min}}^{(K)})^2 + \lambda}
    \left\{\|r_{0\cdot}^{(K)}\|_F - \sqrt{\sum_{t=1}^{T_0}\sum_{i=2}^N\|Y_{it}\|_{\mathcal{H}}^2}\right\}\|Z_0(Z_0'Z_0)^{-1}(Z_1 - Z_0'\hat{\gamma}^{\mathrm{scm}})\|_2.
\end{align}
Note that, under Assumption \ref{ass:singular_values_cov}, 
\[
\frac{\sqrt{\check{m}(K)}\check{d}_{\text{max}}^{(K)}}{(\check{d}_{\text{min}}^{(K)})^2 + \lambda}
\le 
\frac{\sqrt{N-1}C_2}{c_2^2+\lambda}.
\]
Moreover, as $K \to \infty$, 
\[
\|r_{1\cdot}^{(K)} - r_{0\cdot}^{(K)}{'}\hat{\gamma}^{\text{scm}}\|_2 \to \sqrt{\sum_{t=1}^{T_0}\left\|Y_{1t} - \sum_{i=2}^N \hat{\gamma}_i^{\mathrm{scm}} Y_{it}\right\|_{\mathcal{H}}^2}, \quad 
\|r_{0\cdot}^{(K)}\|_F \to \sqrt{\sum_{t=1}^{T_0}\sum_{i=2}^N\|Y_{it}\|_{\mathcal{H}}^2}.
\]
Hence, $R_6^{(K)} \to 0$ as $K \to \infty$.

We finally prove part (b). Fix an arbitrary positive integer $K$. Analogously to the proof of Lemma \ref{lem:aug_weight_prop}, we can show that 
\begin{align}
    \sum_{t=1}^{T_0} \left\|Y_{1t} - \sum_{i=2}^N \hat{\gamma}_i^{\text{cov}(K)}Y_{it}\right\|_{\mathcal{H}}^2 
    =
    \| r_{1\cdot}^{(K)} - r_{0\cdot}^{(K)}{'}\hat{\gamma}^{\mathrm{cov}(K)}\|_2^2 
     +
     \sum_{t=1}^{T_0} \sum_{k=K+1}^\infty \left(r_{1,t,k} - \sum_{i=2}^N \hat{\gamma}_i^{\mathrm{cov}(K)}r_{i,t,k} \right)^2,
\end{align}
which implies that
\begin{align}
    \sqrt{\sum_{t=1}^{T_0} \left\|Y_{1t} - \sum_{i=2}^N \hat{\gamma}_i^{\text{cov}(K)}Y_{it}\right\|_{\mathcal{H}}^2} 
    &\le 
    \| r_{1\cdot}^{(K)} - r_{0\cdot}^{(K)}{'}\hat{\gamma}^{\mathrm{cov}(K)}\|_2 \\
    & \phantom{\le}
    +
    \sqrt{\sum_{t=1}^{T_0} \sum_{k=K+1}^\infty \left(r_{1,t,k} - \sum_{i=2}^N \hat{\gamma}_i^{\mathrm{cov}(K)}r_{i,t,k} \right)^2}.
    \label{eq:prefit_cov_bound}
\end{align}
To bound $ \| r_{1\cdot}^{(K)} - r_{0\cdot}^{(K)}{'}\hat{\gamma}^{\mathrm{cov}(K)}\|_2 $, observe that 
\begin{align}
  &r_{1\cdot}^{(K)} - r_{0\cdot}^{(K)}{'}\hat{\gamma}^{\mathrm{cov}(K)} \\
  &=
  r_{1\cdot} - r_{0\cdot}'\{\hat{\gamma}^{\text{scm}} + \check{r}_{0\cdot}(\check{r}_{0\cdot}'\check{r}_{0\cdot} + \lambda I_{KT_0})^{-1}(\check{r}_{1\cdot} - \check{r}_{0\cdot}\hat{\gamma}^{\text{scm}}) + Z_0(Z_0'Z_0)^{-1}(Z_1 - Z_0' \hat{\gamma}^{\text{scm}})\} \\
  &=
  r_{1\cdot} - r_{0\cdot}'\hat{\gamma}^{\text{scm}}  - r_{0\cdot}'\check{r}_{0\cdot} (\check{r}_{0\cdot}'\check{r}_{0\cdot} + \lambda I_{KT_0})^{-1}(\check{r}_{1\cdot} - \check{r}_{0\cdot}'\hat{\gamma}^{\text{scm}}) - r_{0\cdot}'Z_0(Z_0'Z_0)^{-1}(Z_1 - Z_0' \hat{\gamma}^{\text{scm}})\} \\
  &=
  r_{1\cdot} - r_{0\cdot}'Z_0(Z_0'Z_0)^{-1}Z_1 - \{r_{0\cdot}' -r_{0\cdot}'Z_0(Z_0'Z_0)^{-1}Z_0' \}\hat{\gamma}^{\text{scm}} \\
  &\phantom{=}-   r_{0\cdot}'\check{r}_{0\cdot} (\check{r}_{0\cdot}'\check{r}_{0\cdot} + \lambda I_{KT_0})^{-1}(\check{r}_{1\cdot} -\check{r}_{0\cdot}\hat{\gamma}^{\text{scm}}) \\
  & = 
  \check{r}_{1\cdot} - \check{r}_{0\cdot}'\hat{\gamma}^{\text{scm}} - r_{0\cdot}'\check{r}_{0\cdot} (\check{r}_{0\cdot}'\check{r}_{0\cdot} + \lambda I_{KT_0})^{-1}(\check{r}_{1\cdot} -\check{r}_{0\cdot}\hat{\gamma}^{\text{scm}}) \\
  &= 
  \check{r}_{1\cdot} - \check{r}_{0\cdot}'\hat{\gamma}^{\text{scm}} - \check{r}_{0\cdot}'\check{r}_{0\cdot} (\check{r}_{0\cdot}'\check{r}_{0\cdot} + \lambda I_{KT_0})^{-1}(\check{r}_{1\cdot} -\check{r}_{0\cdot}\hat{\gamma}^{\text{scm}})  \\
  &
  =\{I_{KT_0} - \check{r}_{0\cdot}'\check{r}_{0\cdot}(\check{r}_{0\cdot}'\check{r}_{0\cdot} + \lambda I_{KT_0})^{-1}\}(\check{r}_{1\cdot} - \check{r}_{0\cdot}\hat{\gamma}^{\text{scm}}) \\ 
  &= \lambda (\check{r}_{0\cdot}'\check{r}_{0\cdot} + \lambda I_{KT_0})^{-1}(\check{r}_{1\cdot} - \check{r}_{0\cdot}\hat{\gamma}^{\text{scm}}).
\end{align}
Analogously to the proof of Lemma \ref{lem:aug_weight_prop}, combining \eqref{eq:rtilde_inv} and the final expression in the display above, we obtain
\begin{equation}
     \| r_{1\cdot}^{(K)} - r_{0\cdot}^{(K)}{'}\hat{\gamma}^{\mathrm{cov}(K)}\|_2
     =
     \|\lambda \check{E}\check{V}(\check{r}_{1\cdot} - \check{r}_{0\cdot}'\hat{\gamma}^{\text{scm}})\|_2.
     \label{eq:prefit_cov_expression}
\end{equation}
By the sub-multiplicativity of the Frobenius norm, the right-hand side is bounded as
\begin{align}
    \|\lambda \check{E}\check{V}(\check{r}_{1\cdot} - \check{r}_{0\cdot}'\hat{\gamma}^{\text{scm}})\|_2^2
    &\le 
    \|\lambda \check{E}\|_2^2\|\check{r}_{1\cdot} - \check{r}_{0\cdot}'\hat{\gamma}^{\text{scm}}\|_2^2 \\ 
    & =
    \sum_{j=1}^{\check{m}(K)} \left(\frac{\lambda}{(\check{d}_j^{(K)})^2+\lambda}\right)^2\|\check{r}_{1\cdot} - \check{r}_{0\cdot}'\hat{\gamma}^{\text{scm}}\|_2^2 \\
    & \le \check{m}(K)\left(\frac{\lambda}{(\check{d}_{\text{min}}^{(K)})^2+\lambda}\right)^2\|\check{r}_{1\cdot} - \check{r}_{0\cdot}'\hat{\gamma}^{\text{scm}}\|_2^2.
    \label{eq:prefit_cov_expression_bound}
\end{align}
Hence,  
\begin{equation}
    \| r_{1\cdot}^{(K)} - r_{0\cdot}^{(K)}{'}\hat{\gamma}^{\mathrm{cov}(K)}\|_2
    \le 
    \frac{\sqrt{\check{m}(K)}\lambda}{(\check{d}_{\text{min}}^{(K)})^2+\lambda}
    \|\check{r}_{1\cdot} - \check{r}_{0\cdot}'\hat{\gamma}^{\text{scm}}\|_2.
    \label{eq:bound_fit_cov}
\end{equation}
Combining \eqref{eq:prefit_cov_bound} and \eqref{eq:bound_fit_cov} with \eqref{eq:eval_rtilde_scm}, we obtain 
\[
\sqrt{\sum_{t=1}^{T_0}\left\|Y_{1t} - \sum_{i=2}^N \hat{\gamma}_i^{\mathrm{cov}(K)}Y_{it}\right\|_{\mathcal{H}}^2 }
        \le F_1(\lambda)
        +
        R_5^{(K)},
\]
where
\begin{align}
        F_1(\lambda) &= \frac{\sqrt{\check{m}(K)}\lambda}{(\check{d}_{\text{min}}^{(K)})^2+\lambda}\sqrt{\sum_{t=1}^{T_0}\left\|Y_{1t} - \sum_{i=2}^N \hat{\gamma}_i^{\mathrm{scm}} Y_{it}\right\|_{\mathcal{H}}^2} \\
        & \phantom{\le} +
        \frac{\sqrt{\check{m}(K)}\lambda}{(\check{d}_{\text{min}}^{(K)})^2+\lambda}\sqrt{\sum_{t=1}^{T_0}\sum_{i=2}^N\|Y_{it}\|_{\mathcal{H}}^2}
        \|Z_0(Z_0'Z_0)^{-1}(Z_1 - Z_0'\hat{\gamma}^{\mathrm{scm}})\|_2,
\end{align}
and 
\begin{align}
    R_5^{(K)}  
    &= \frac{\sqrt{\check{m}(K)}\lambda}{(\check{d}_{\text{min}}^{(K)})^2 + \lambda}
    \left\{\|r_{1\cdot}^{(K)} - r_{0\cdot}^{(K)}{'}\hat{\gamma}^{\text{scm}}\|_2 -  \sqrt{\sum_{t=1}^{T_0}\left\|Y_{1t} - \sum_{i=2}^N \hat{\gamma}_i^{\mathrm{scm}} Y_{it}\right\|_{\mathcal{H}}^2} \right\} \\
    &\phantom{\le}+
    \frac{\sqrt{\check{m}(K)}\lambda}{(\check{d}_{\text{min}}^{(K)})^2 + \lambda}
    \left\{\|r_{0\cdot}^{(K)}\|_F - \sqrt{\sum_{t=1}^{T_0}\sum_{i=2}^N\|Y_{it}\|_{\mathcal{H}}^2}\right\}\|Z_0(Z_0'Z_0)^{-1}(Z_1 - Z_0'\hat{\gamma}^{\mathrm{scm}})\|_2 \\ 
    &\phantom{\le}+
    \sqrt{\sum_{t=1}^{T_0} \sum_{k=K+1}^\infty \left(r_{1,t,k} - \sum_{i=2}^N \hat{\gamma}_i^{\mathrm{cov}(K)}r_{i,t,k} \right)^2}.
\end{align}
Note that, under Assumption \ref{ass:singular_values_cov}, 
\[
\frac{\sqrt{\check{m}(K)}\lambda}{(\check{d}_{\text{min}}^{(K)})^2 + \lambda}
\le 
\frac{\sqrt{N-1}\lambda}{c_2^2+\lambda}.
\]
Moreover, as $K \to \infty$, 
\[
\|r_{1\cdot}^{(K)} - r_{0\cdot}^{(K)}{'}\hat{\gamma}^{\text{scm}}\|_2 \to \sqrt{\sum_{t=1}^{T_0}\left\|Y_{1t} - \sum_{i=2}^N \hat{\gamma}_i^{\mathrm{scm}} Y_{it}\right\|_{\mathcal{H}}^2}, \quad 
\|r_{0\cdot}^{(K)}\|_F \to \sqrt{\sum_{t=1}^{T_0}\sum_{i=2}^N\|Y_{it}\|_{\mathcal{H}}^2}.
\]
Furthermore, part (c) implies that there exists a constant $B_2 > 0$ such that $\|\hat{\gamma}^{\text{cov}(K)}\|_2 \le B_2$
for any $K$. Hence, analogously to the proof of Lemma \ref{lem:aug_weight_prop}, we can show that as $K \to \infty$, 
\begin{equation}
    \sqrt{\sum_{t=1}^{T_0} \sum_{k=K+1}^\infty \left(r_{1,t,k} - \sum_{i=2}^N \hat{\gamma}_i^{\mathrm{cov}(K)}r_{i,t,k} \right)^2} \to 0.
    \label{eq:reminder_term}
\end{equation}
Hence, $R_5^{(K)} \to 0$ as $K \to \infty$.

\paragraph*{Proof of Theorem \ref{thm:est_error_auto_cov}.}
As in the proof of Theorem \ref{thm:est_error_auto}, we aim to bound $\|Y_{1T}^N - \hat{Y}_{1T}^N\|_{\mathcal{H}}$. Under Assumption \ref{ass_auto_cov}, the difference  between $Y_{1T}^N$ and $\hat{Y}_{1T}^N$ can be decomposed as 
\begin{align}
     Y_{1T}^N(x) - \hat{Y}_{1T}^N(x) 
    &=
    \sum_{t=1}^{T_0} \left\langle \beta_t(x, \cdot), Y_{1t} - \sum_{i=2}^N\hat{\gamma}_i Y_{it} \right\rangle_{\mathcal{H}}
    + 
    \sum_{\ell=1}^p \eta_\ell(x)\left(Z_{1\ell} - \sum_{i=2}^N \hat{\gamma}_iZ_{i\ell}\right) \\ 
    & \phantom{=} +
    \varepsilon_{1T}(x) - \sum_{i=2}^N \hat{\gamma}_i \varepsilon_{iT}(x).
    \label{eq:decompose_auto_cov}
\end{align} 
This yields the bound
\begin{align}
    \|Y_{1T}^N - \hat{Y}_{1T}^N\|_{\mathcal{H}}
    \le 
    \|\Delta_{1}\|_{\mathcal{H}} + \|\Delta_{4}\|_{\mathcal{H}} + 
    \left\|\varepsilon_{1T} - \sum_{i=2}^N \hat{\gamma}_i \varepsilon_{iT} \right\|_{\mathcal{H}},
    \label{eq:decomposed_bound_auto_cov}
\end{align}
where $\Delta_1$ is defined in  \eqref{eq:delta_1_definition}, and $\Delta_4: \mathcal{X} \to \mathbb{R}$ is defined by
\begin{align} 
    \Delta_{4}(x) = \sum_{\ell=1}^p \eta_\ell(x)\left(Z_{1\ell} - \sum_{i=2}^N \hat{\gamma}_iZ_{i\ell}\right).
\end{align}
As shown in the proof of Theorem \ref{thm:est_error_auto}, we have 
\begin{align}
    \|\Delta_{1}\|_{\mathcal{H}} \le 
    \sqrt{\sum_{t=1}^{T_0} \|\beta_t\|_{\mathcal{H}}^2} \sqrt{\sum_{t=1}^{T_0} \left\|Y_{1t} - \sum_{i=2}^N \hat{\gamma}_iY_{it}\right\|_{\mathcal{H}}^2}.
    \label{eq:bound_delta_1_ver2}
\end{align}
Furthermore, it holds that
\begin{equation}
     \left\|\varepsilon_{1T} - \sum_{i=2}^N \hat{\gamma}_i \varepsilon_{iT}\right\|_{\mathcal{H}}
    \le \delta \sigma(1+\|\hat{\gamma}\|_2)
    \label{eq:bound_error_ver2}
\end{equation}
with probability at least $1 - 2e^{-\delta^2/2}$.

We now bound $\|\Delta_4\|_{\mathcal{H}}$. For each $x \in \mathcal{X}$, by the Cauchy–Schwarz inequality, 
\begin{equation}
    |\Delta_4(x)|
    \le 
    \sqrt{\sum_{\ell=1}^p \eta_\ell(x)^2}\left\|Z_1 - \sum_{i=2}^N \hat{\gamma}_i Z_i \right\|_{2}.
\end{equation}
Hence,  
\begin{align}
    \|\Delta_4\|_{\mathcal{H}}
    &=
    \sqrt{\int_\mathcal{X}|\Delta_4(x)|^2 d\mu(x)}
    \le 
    \sqrt{\int_\mathcal{X}\sum_{\ell = 1}^p \eta_\ell(x)^2 d\mu(x) } \left\|Z_1 - \sum_{i=2}^N \hat{\gamma}_i Z_i \right\|_{2} \\
    &=
    \sqrt{\sum_{\ell=1}^p \|\eta_\ell\|_{\mathcal{H}}^2}
    \left\|Z_1 - \sum_{i=2}^N \hat{\gamma}_i Z_i \right\|_{2}.
    \label{bound_delta_4}
\end{align}
Combining \eqref{eq:decomposed_bound_auto_cov}, \eqref{eq:bound_delta_1_ver2}, \eqref{eq:bound_error_ver2} and \eqref{bound_delta_4}, we obtain the desired result. 

\paragraph*{Proof of Corollary \ref{cor:est_error_auto_cov}.}
The proof follows along similar lines as that of Corollary \ref{cor:error_auto}, and is therefore omitted.

\paragraph*{Proof of Theorem \ref{thm:est_error_factor_cov}.}
We begin by stating a lemma that expresses the vector of factor
loadings in terms of the post-treatment outcomes, covairates and post-treatment errors.

\begin{lem}
    Suppose Assumption \ref{ass_latent_cov} holds. Then, for each $i=1, \ldots, N$, 
    \[
    \phi_i  =  \{\mu(x)'\mu(x)\}^{-1}\mu(x)'\{Y_{i\cdot}(x) - \eta(x)Z_i - \varepsilon_{i\cdot}(x)\}
\]
for any $x \in \mathcal{X}$,
where $Y_{i\cdot}(x) = (Y_{i1}(x), \ldots, Y_{iT_0}(x))' \in  \mathbb{R}^{T_0}$ and $\varepsilon_{i\cdot}(x) = (\varepsilon_{i1}(x), \ldots, \varepsilon_{iT_0}(x))' \in  \mathbb{R}^{T_0}$. 
\label{lem:factor_expression_cov}
\end{lem}
\begin{proof}
    Under the latent factor model \eqref{eq:latent_factor_cov}, we have 
    $
        Y_{i\cdot}(x) = \mu(x)\phi_i 
        + 
        \eta(x)Z_i
        +
        \varepsilon_{i\cdot}(x). 
    $
    Multiplying both sides by $\mu(x)'$ gives $\mu(x)'Y_{i\cdot}(x) = \mu(x)'\mu(x)\phi_i + 
    \mu(x)'\eta(x)Z_i
    +
    \mu(x)'\varepsilon_{i\cdot}(x)$, which immediately gives the desired result.
\end{proof}

We now prove Theorem \ref{thm:est_error_factor_cov}.
As in the proof of Theorem \ref{thm:est_error_auto}, we aim to bound $\|Y_{1T}^N - \hat{Y}_{1T}^N\|_{\mathcal{H}}$. 
Under the latent factor model \eqref{eq:latent_factor_cov} and applying Lemma \ref{lem:factor_expression_cov}, we can decompose the difference between $Y_{1T}^N$ and $\hat{Y}_{1T}^N$ as 
 \begin{align}
Y_{1T}^N(x) - \hat{Y}_{1T}^N(x)
&=
\left(\phi_1 - \sum_{i=2}^N \hat{\gamma}_i \phi_i\right){'}\mu_T(x)
+
\sum_{\ell=1}^p \eta_{\ell T}(x)\left(Z_{1\ell} - \sum_{i=2}^N \hat{\gamma}_i Z_{i\ell}\right) \\
& \phantom{=} +
\varepsilon_{1T}(x) - \sum_{i=2}^N \hat{\gamma}_i \varepsilon_{iT}(x) \\ 
& = 
\left\{Y_{1\cdot}(x) - \sum_{i=2}^N \hat{\gamma}_i Y_{i\cdot}(x)\right\}' \mu(x)\{\mu(x)'\mu(x)\}^{-1} \mu_T(x) \\ 
&\phantom{=}
+ 
\left(Z_1 - \sum_{i=2}^N \hat{\gamma}_iZ_i\right)'[\eta_T(x) + \eta(x)'\mu(x)\{\mu(x)'\mu(x)\}^{-1}\mu_T(x)]
\\
& \phantom{=}
+
\left\{\varepsilon_{1\cdot}(x) - \sum_{i=2}^N \hat{\gamma}_i \varepsilon_{i\cdot}(x)\right\}' \mu(x) \{\mu(x)'\mu(x)\}^{-1}\mu_T(x) 
+
\varepsilon_{1T}(x) - \sum_{i=2}^N \hat{\gamma}_i \varepsilon_{iT}(x).
\end{align}
This yields 
\begin{align}
    \|Y_{1T}^{N} - \hat{Y}_{1T}^N\|_{\mathcal{H}}
    \le \|\Delta_2\|_{\mathcal{H}} + \|\Delta_5\|_{\mathcal{H}}+ \|\Delta_3\|_{\mathcal{H}} + \left\|\varepsilon_{1T} - \sum_{i=2}^N \hat{\gamma}_i \varepsilon_{iT}\right\|_{\mathcal{H}}, 
    \label{eq:error_decompose_latent_cov}
\end{align}
where $\Delta_2$ and $\Delta_3$ are defined in \eqref{eq:delta_23_definition}, and $\Delta_5: \mathcal{X} \to \mathbb{R}$ is defined by
\[
\Delta_5(x) = 
\left(Z_1 - \sum_{i=2}^N \hat{\gamma}_iZ_i\right){'}[\eta_T(x) + \eta(x)'\mu(x)\{\mu(x)'\mu(x)\}^{-1}\mu_T(x)].
\]

Fix the pre-treatment errors $\varepsilon_{it}, i=1, \ldots N, t=1, \ldots T_0$. As shown in the proof of Theorem \ref{thm:est_error_factor}, 
\begin{align}
    \|\Delta_{2}\|_{\mathcal{H}} \le 
    \frac{M_1^2J^{3/2}}{M_2\sqrt{T_0}} \sqrt{\sum_{t=1}^{T_0} \left\|Y_{1t} - \sum_{i=2}^N Y_{it}\right\|_{\mathcal{H}}^2}, \quad
    \|\Delta_{3}\|_{\mathcal{H}}
    \le \frac{2\sigma M_1^2J^{3/2}}{M_2 \sqrt{T_0}} \|\hat{\gamma}\|_1,
    \label{eq:bound_delta_23_ver2}
\end{align}
and as shown in the proof of Theorem \ref{thm:est_error_auto}, it holds that 
\begin{equation}
    \left\|\varepsilon_{1T} - \sum_{i=2}^N \hat{\gamma}_i \varepsilon_{iT}\right\|_{\mathcal{H}}
    \le \delta \sigma( 1+ \|\hat{\gamma}\|_2)
    \label{eq:bound_error_ver4}
\end{equation}
with probability at least $1 - 2e^{-\delta ^2/2}$.

We now bound $\|\Delta_5\|_{\mathcal{H}}$. 
For each $x \in \mathcal{X}$, by the Cauchy–Schwarz inequality and the sub-multiplicativity of the Frobenius norm, 
\begin{align}
    |\Delta_5(x)| 
    &\le 
    \left\|Z_1 - \sum_{i=2}^N \hat{\gamma}_i Z_i \right\|_2
    \|\eta_T(x)+ \eta(x)'\mu(x)\{\mu(x)'\mu(x)\}^{-1}\mu_T(x)\|_2 \\ 
    & \le 
    [\|\eta_T(x)\|_2 + \|\eta(x)\|_F\|\mu(x)\|_F \|\{\mu(x)'\mu(x)\}^{-1}\|_F\|\mu_T(x)\|_2 ]
    \left\|Z_1 - \sum_{i=2}^N \hat{\gamma}_i Z_i \right\|_2 \\ 
    & \le 
    \left(\|\eta_T(x)\|_2 + \frac{M_1^2J^{3/2}}{M_2\sqrt{T_0}}\|\eta(x)\|_F \right)
    \left\|Z_1 - \sum_{i=2}^N \hat{\gamma}_i Z_i \right\|_2.
\end{align}
Hence, 
\begin{align}
    \|\Delta_5\|_{\mathcal{H}}
    &\le 
    \sqrt{\int_\mathcal{X}\left\{\|\eta_T(x)\|_2 + \frac{M_1^2J^{3/2}}{M_2\sqrt{T_0}}\|\eta(x)\|_F \right\}^2 d\mu(x)}\left\|Z_1 - \sum_{i=2}^N \hat{\gamma}_i Z_i \right\|_2 \\ 
    & \le 
    \sqrt{\int_\mathcal{X}2\left\{\|\eta_T(x)\|_2^2 + \left(\frac{M_1^2J^{3/2}}{M_2\sqrt{T_0}}\right)^2\|\eta(x)\|_F^2 \right\} d\mu(x)}\left\|Z_1 - \sum_{i=2}^N \hat{\gamma}_i Z_i \right\|_2 \\ 
    & \le 
    \sqrt{2} \max\left\{1, \frac{M_1^2J^{3/2}}{M_2\sqrt{T_0}}\right\} \sqrt{\int_\mathcal{X} \{\|\eta_T(x)\|_2^2 + \|\eta(x)\|_F^2\} d\mu(x)} \left\|Z_1 - \sum_{i=2}^N \hat{\gamma}_i Z_i \right\|_2 \\ 
    &= 
    \sqrt{2} \max\left\{1, \frac{M_1^2J^{3/2}}{M_2\sqrt{T_0}}\right\}\sqrt{\sum_{t=1}^T \sum_{\ell = 1}^p\|\eta_{\ell t}\|_{\mathcal{H}}^2}\left\|Z_1 - \sum_{i=2}^N \hat{\gamma}_i Z_i \right\|_2.
    \label{eq:bound_delta_5}
\end{align}
Combining 
\eqref{eq:error_decompose_latent_cov}, 
\eqref{eq:bound_delta_23_ver2}, \eqref{eq:bound_error_ver4} and \eqref{eq:bound_delta_5}, we have 
\begin{align}
    d(\nu_{1T}^N, \hat{\nu}_{1T}^N) 
        &\le \frac{M_1^2J^{3/2}}{M_2\sqrt{T_0}} \sqrt{\sum_{t=1}^{T_0} \left\|Y_{1t} - \sum_{i=2}^N \hat{\gamma}_iY_{it}\right\|_{\mathcal{H}}^2} \\
        & \phantom{\le}
        +
        \sqrt{2}\max\left\{1, \frac{M_1^2J^{3/2}}{M_2\sqrt{T_0}}\right\} \sqrt{\sum_{t=1}^T \sum_{\ell=1}^p \|\eta_{\ell t}\|_{\mathcal{H}}^2}\left\|Z_1 - \sum_{i=2}^N \hat{\gamma}_i Z_i\right\|_2 \\
        & \phantom{\le}
        +
        \frac{2\sigma M_1^2J^{3/2}}{M_2\sqrt{T_0}}\|\hat{\gamma}\|_1
        +
        \delta \sigma(1+\|\hat{\gamma}\|_2)
\end{align}
with probability at least $1 - 2e^{-\delta^2/2}$ conditionally on the pre-treatment errors. 
By applying the law of iterated expectations, we obtain the desired result. 

\paragraph*{Proof of Corollary \ref{cor:est_error_factor_cov}.}
The proof follows along similar lines as that of Corollary \ref{cor:est_error_factor}, and is therefore omitted.

\section{Supplementary Figures for Empirical Analysis}
\subsection{Supplementary Figures for Section \ref{susec:abortion}}
\label{subsec:supplementary_abortion}
Figure \ref{fig:data_asfr} displays the ASFR curves for East Germany and the control countries from $1956$ to $1975$. 
Figure \ref{fig:synthetic_asfr_until1963} compares the observed ASFR curves of East Germany with the corresponding ASFR curves of the FSC and ridge augmented FSC units over the periods from 1956 to 1963.
Figure \ref{fig:placebo_asfr} presents the detailed results of the placebo permutation test. 

\begin{figure}[t]
    \centering
    \includegraphics[width=\linewidth]{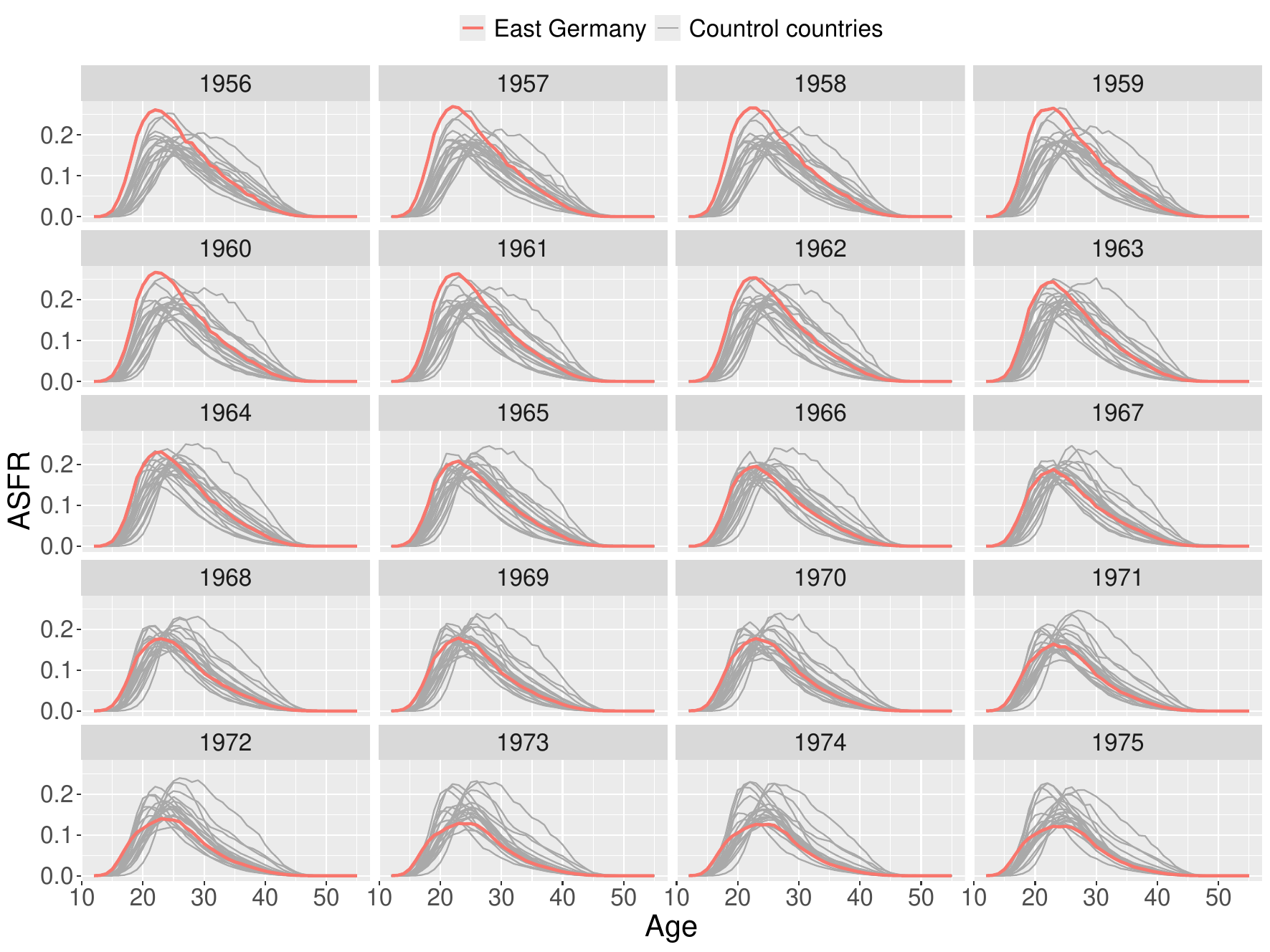}
    \caption{ASFR curves for East Germany and the control countries from 1956 to 1975.}
    \label{fig:data_asfr}
\end{figure}

\begin{figure}[t]
    \centering
    \includegraphics[width=\linewidth]{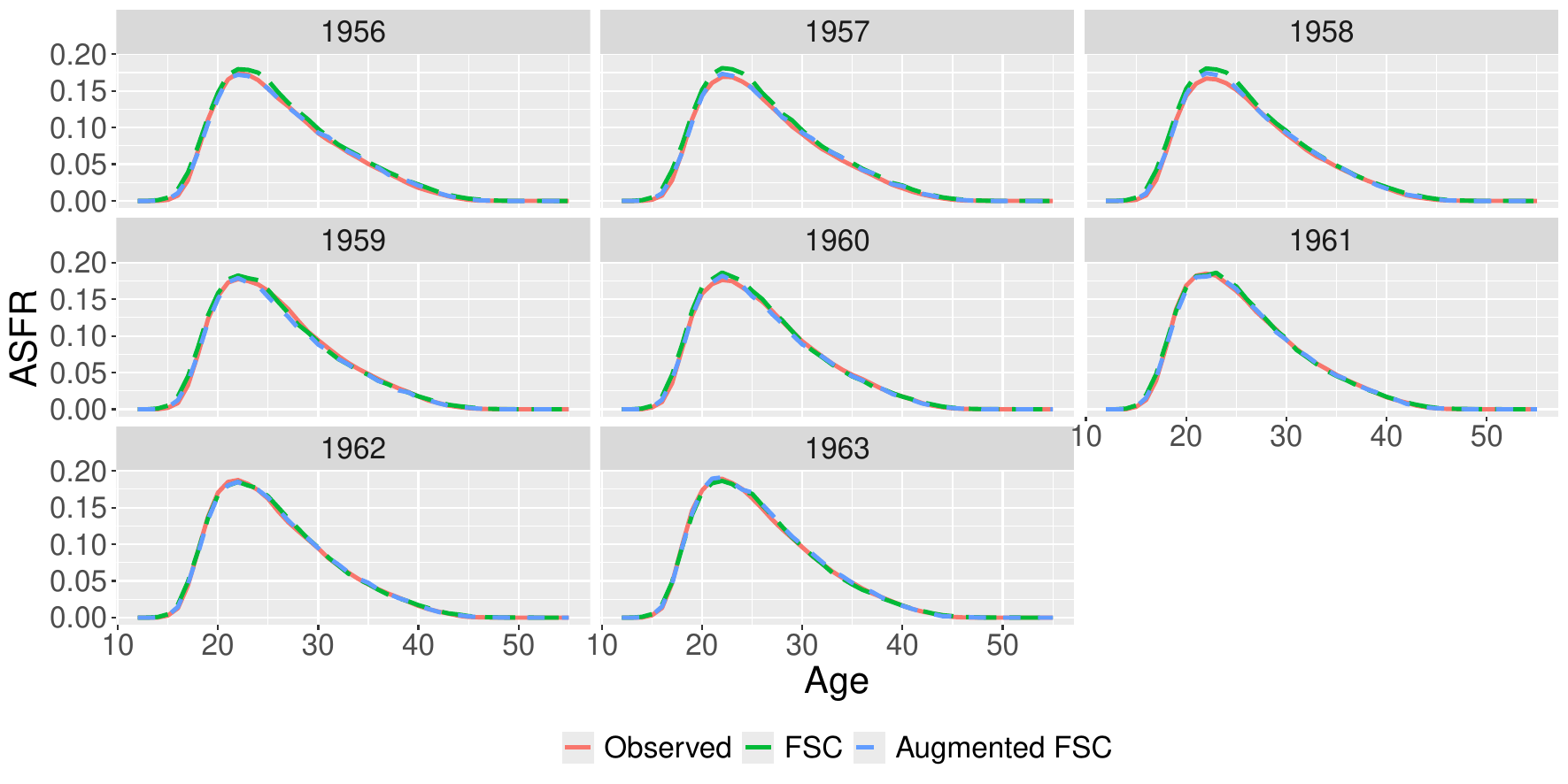}
    \caption{Observed ASFR curves for East Germany and the corresponding ASFR curves for the FSC and ridge augmented FSC units during pre-treatment periods (1956-1963).}
    \label{fig:synthetic_asfr_until1963}
\end{figure}

\begin{figure}[t]
    \centering
    \includegraphics[width=0.8\linewidth]{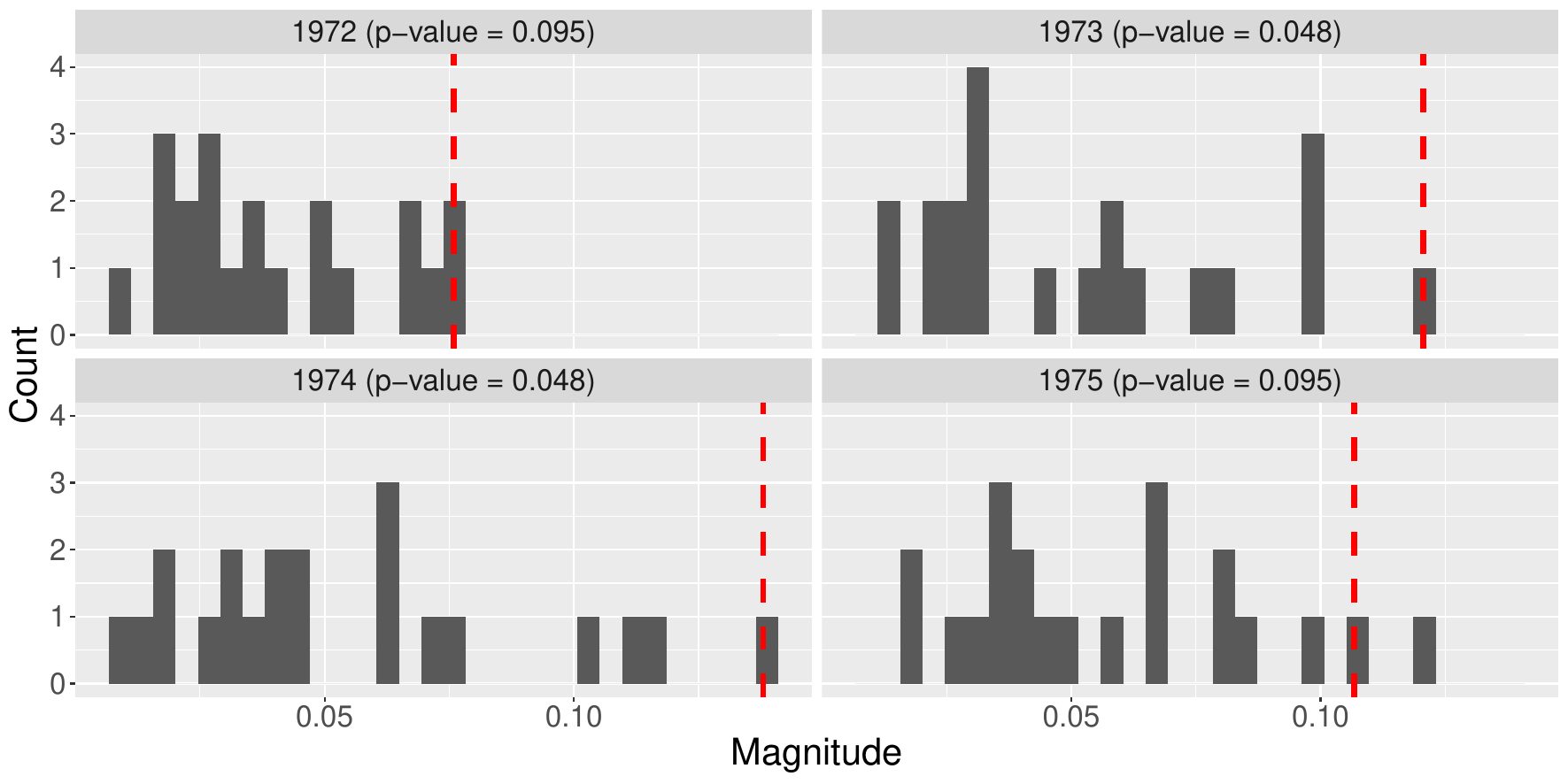}
    \caption{Results of the placebo permutation tests for the ASFR data based on the ridge augmented FSC method. The histograms depict the magnitudes of the causal effects for all units, with the red dashed lines indicating the corresponding magnitudes for the treated unit.}
    \label{fig:placebo_asfr}
\end{figure}

\subsection{Supplementary Figures for Section \ref{subsec:mortality}}
\label{subsec:supplementary_mortality}
Figure \ref{fig:data_aad_quantile} displays the quantile functions of the age-at-death distributions for Russia and the control countries from 1970 to 1999. 
Figure \ref{fig:data_aad_density} shows their density functions. 
Figure \ref{fig:synthetic_aad_until1984} compares the quantile functions of the observed distributions for Russia with the corresponding quantile functions of the FSC and augmented FSC units over the periods from 1970 to 1984.
Figure \ref{fig:placebo_aad} presents the detailed results of the placebo permutation test. 

\begin{figure}[h]
    \centering
    \includegraphics[width=\linewidth]{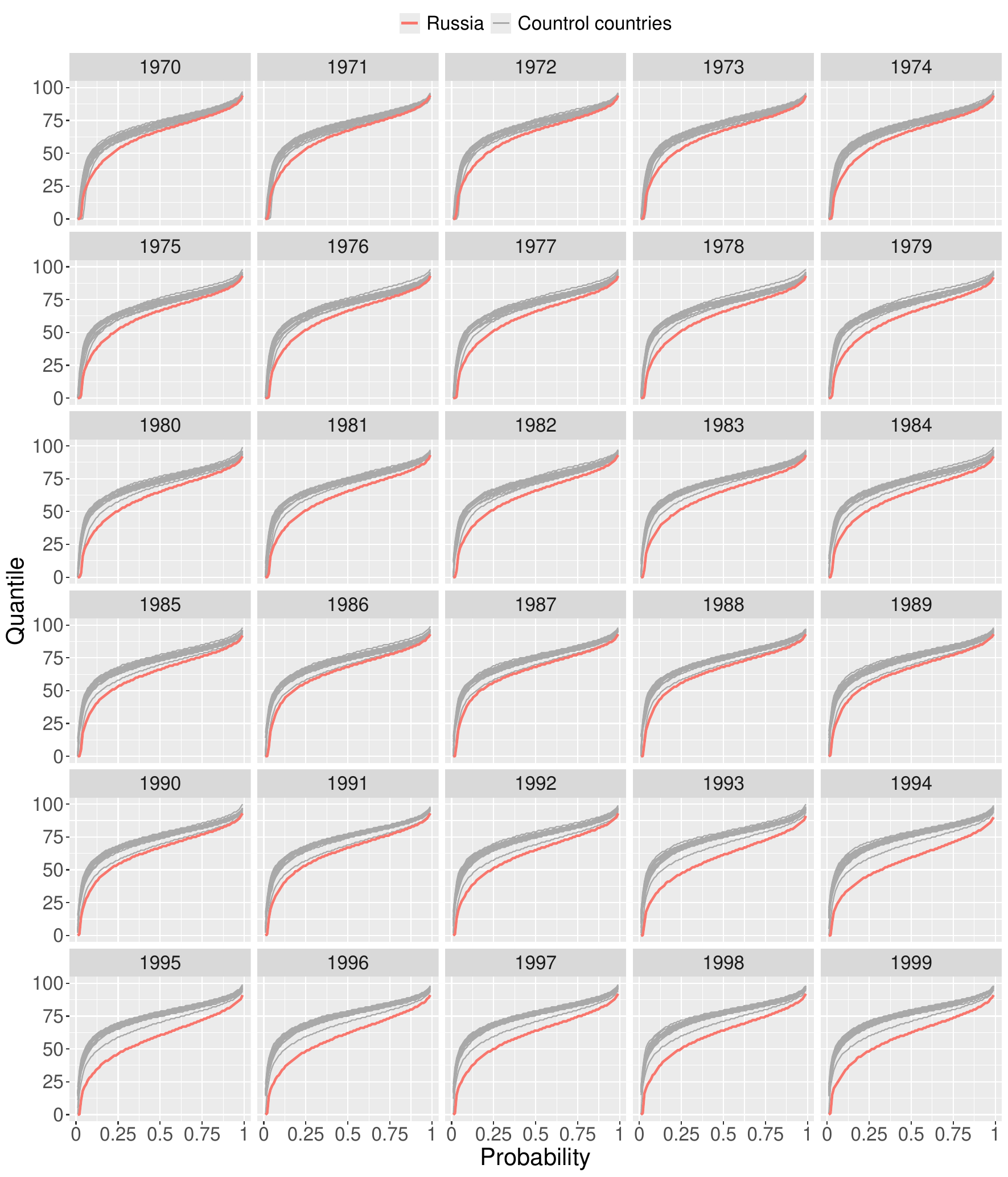}
    \caption{The quantile functions of the age-at-death distributions for Russia and the control countries from 1970 to 1999.}
    \label{fig:data_aad_quantile}
\end{figure}

\begin{figure}[h]
    \centering
    \includegraphics[width=\linewidth]{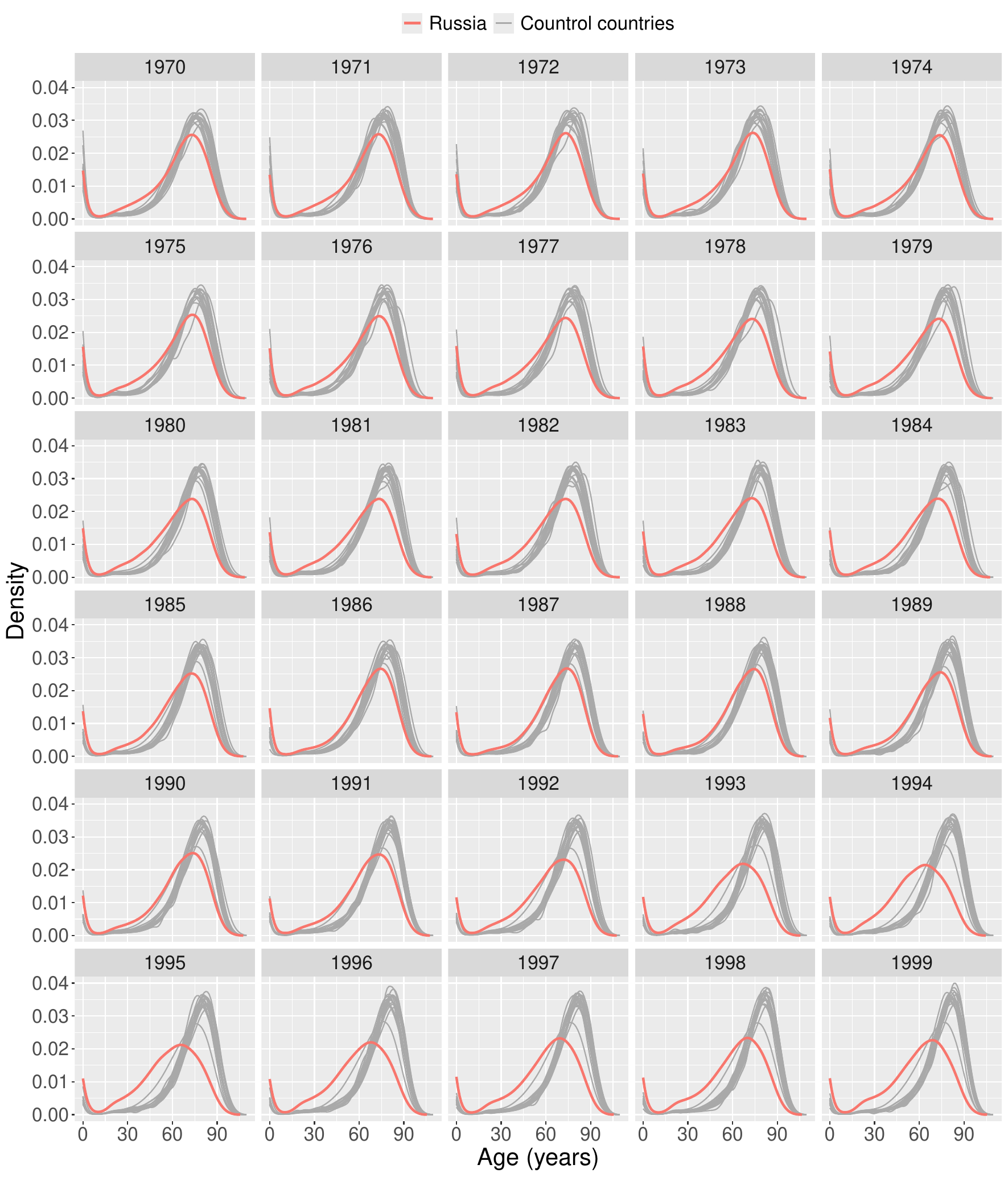}
    \caption{The density functions of age-at-death distributions for Russia and the control countries from 1970 to 1999.}
    \label{fig:data_aad_density}
\end{figure}

\begin{figure}[h]
    \centering
    \includegraphics[width=\linewidth]{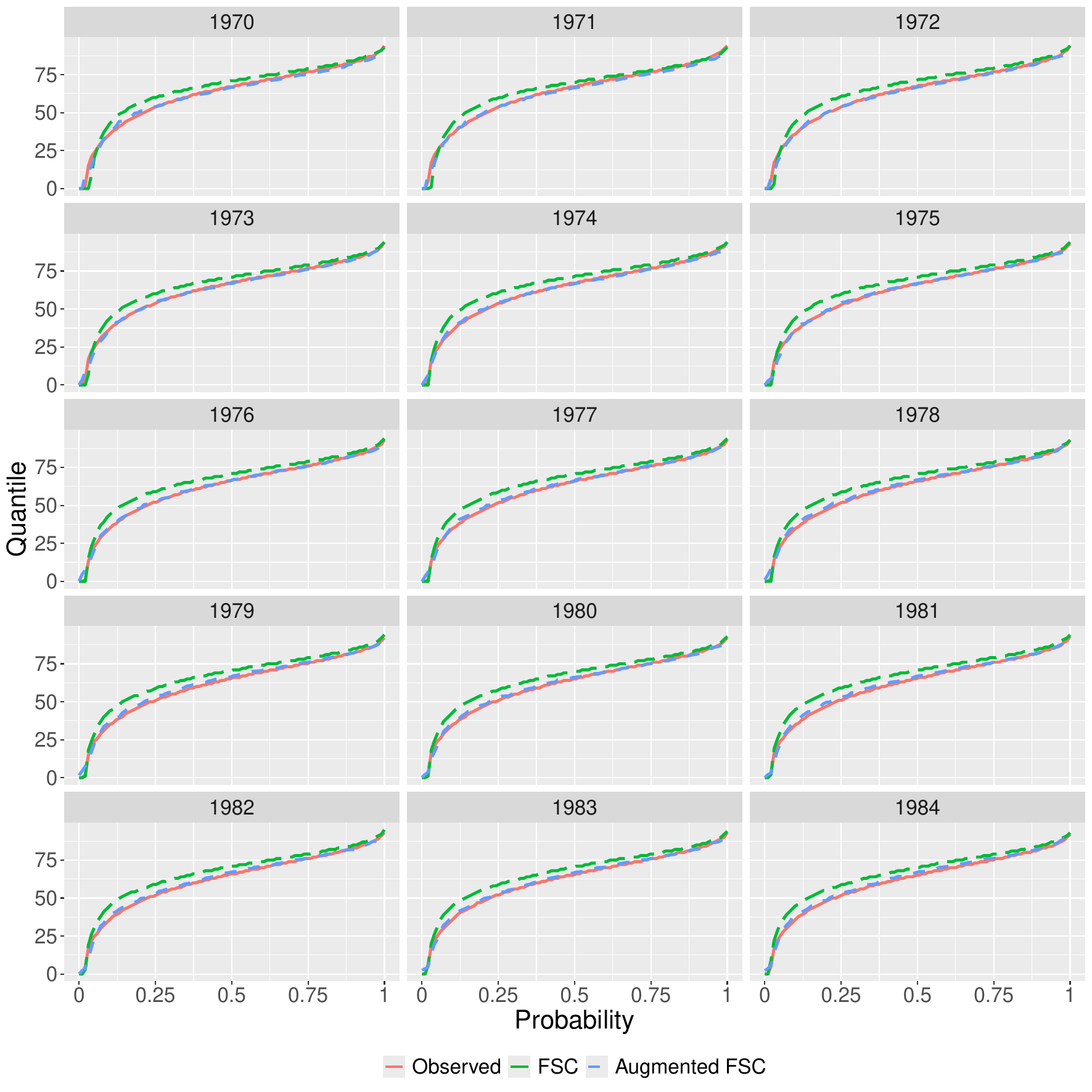}
    \caption{The quantile functions of the observed age-at-death distributions for Russia and the corresponding quantile functions obtained from the FSC and ridge augmented FSC units during pre-treatment periods (1970-1984)}
    \label{fig:synthetic_aad_until1984}
\end{figure}

\begin{figure}[h]
    \centering
    \includegraphics[width=\linewidth]{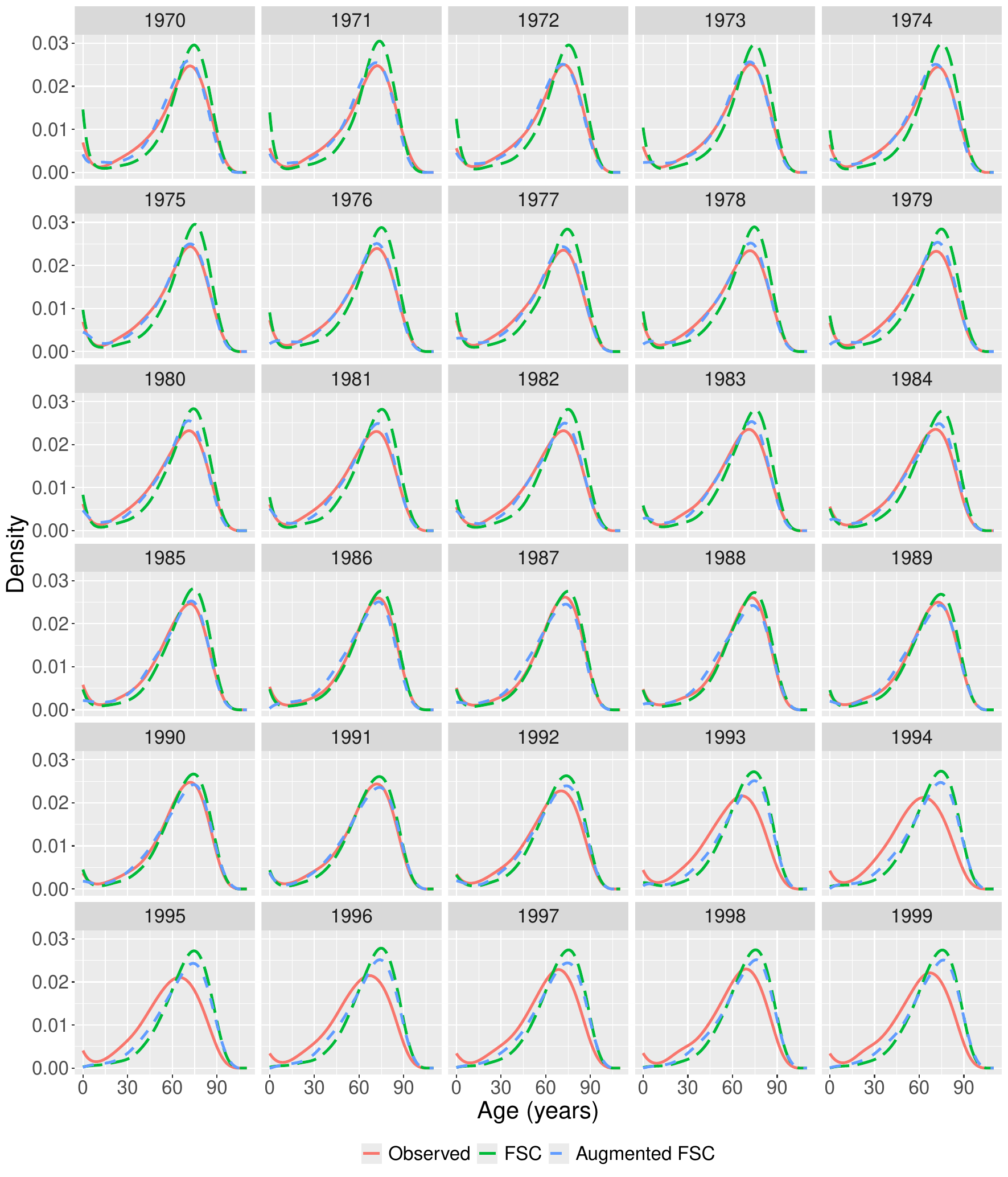}
    \caption{The observed density functions of the age-at-death distributions for Russia and the corresponding density functions obtained from the FSC and ridge augmented FSC units. The periods from 1980 to 1990 are the pre-treatment periods, while 1991 to 1999 are the post-treatment periods.}
    \label{fig:synthtic_outcome_aad_density}
\end{figure}

\begin{figure}[h]
    \centering
    \includegraphics[width=0.9\linewidth]{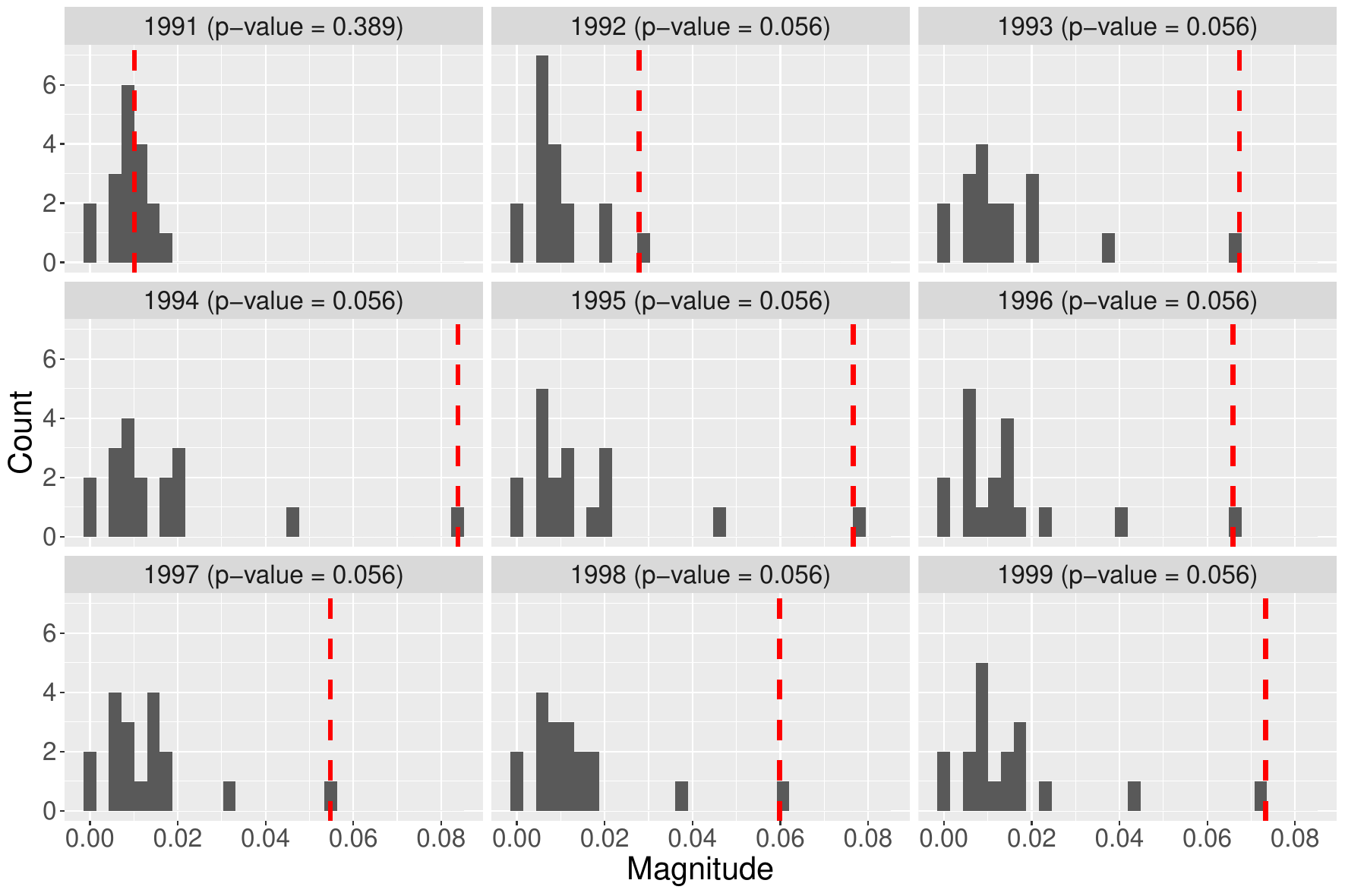}
    \caption{Results of the placebo permutation tests for the mortality data based on the ridge augmented FSC method. The histograms depict the magnitudes of the causal effects for all units, with the red dashed lines indicating the corresponding magnitudes for the treated unit.}
    \label{fig:placebo_aad}
\end{figure}

\subsection{Supplementary Figures for Section \ref{subsec:service}}
\label{subsec:supplementary_service}
Figure \ref{fig:service_difference_fsc_pre} displays heatmaps of the differences between the observed trade covariance matrices for the UK and the corresponding matrices for the FSC units during pre-treatment periods (2009 Q1--2014 Q4).
Figure \ref{fig:service_difference_afsc_pre} presents the corresponding results for the augmented FSC units over the same pre-treatment period (2009 Q1--2014 Q4).
Figures \ref{fig:service_difference_fsc_after} and \ref{fig:service_difference_afsc_after} present the corresponding results for the FSC and augmented FSC units during post-treatment period (2017 Q3--2018 Q2).
Figure \ref{fig:placebo_service} presents the detailed results of the placebo permutation test.

\begin{figure}[h]
    \centering
    \includegraphics[width=\linewidth]{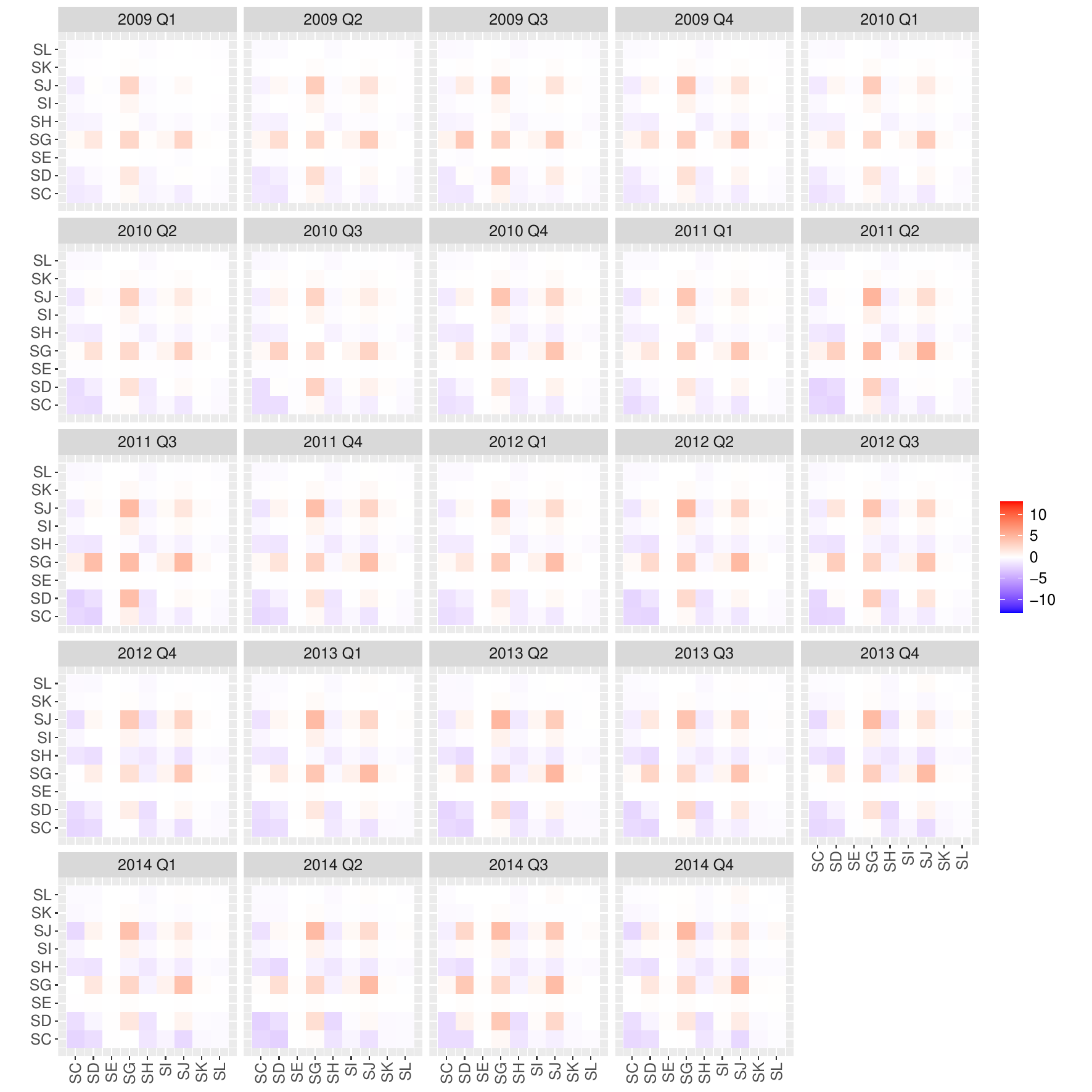}
    \caption{Heatmaps of the differences between the observed trade covariance matrices for the UK and the corresponding matrices for the FSC units during pre-treatment periods (2009 Q1--2014 Q4).}
    \label{fig:service_difference_fsc_pre}
\end{figure}

\begin{figure}[h]
    \centering
    \includegraphics[width=\linewidth]{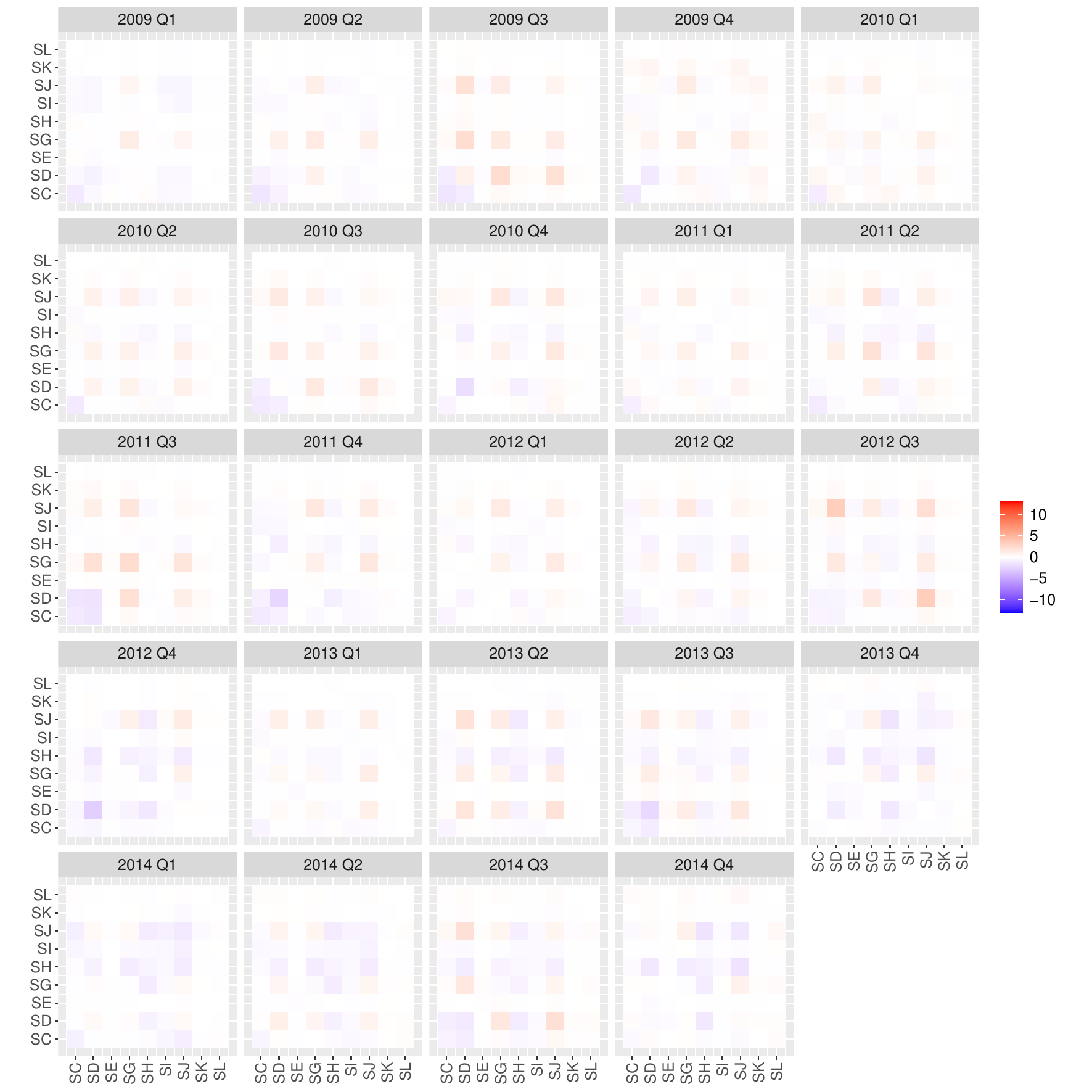}
    \caption{Heatmaps of the differences between the observed trade covariance matrices for the UK and the corresponding matrices for the  ridge augmented FSC units during pre-treatment periods (2009 Q1--2014 Q4).}
    \label{fig:service_difference_afsc_pre}
\end{figure}

\begin{figure}[h]
    \centering
    \includegraphics[width=\linewidth]{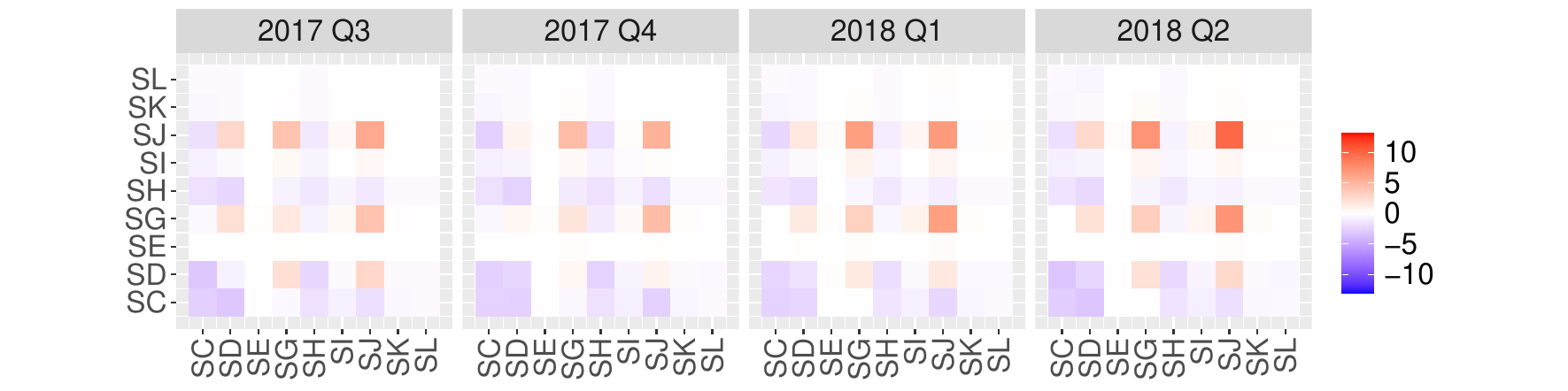}
    \caption{Heatmaps of the differences between the observed trade covariance matrices for the UK and the corresponding matrices for the FSC units during post-treatment periods (2017 Q3--2018 Q2).}
    \label{fig:service_difference_fsc_after}
\end{figure}

\begin{figure}[h]
    \centering
    \includegraphics[width=\linewidth]{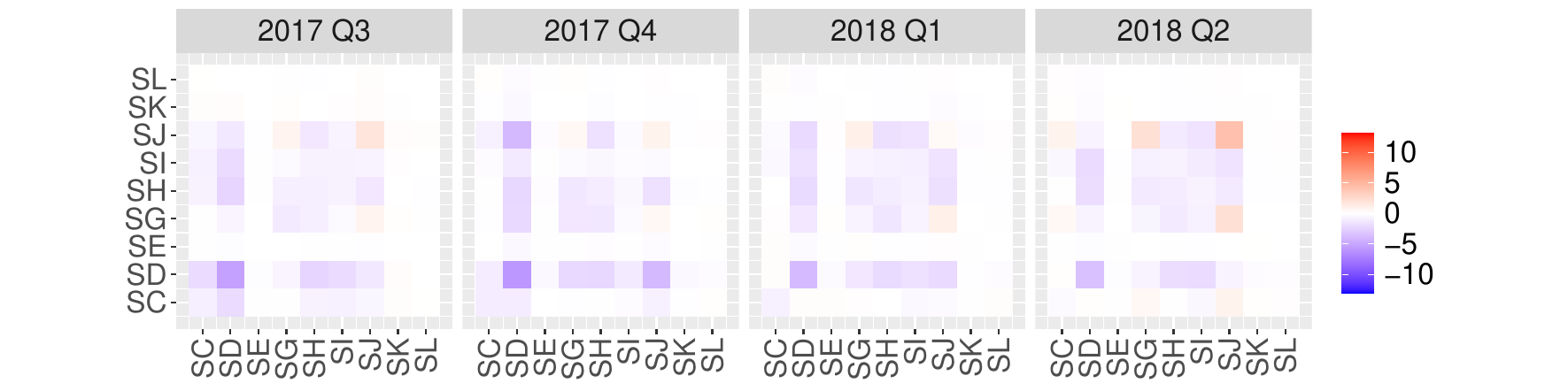}
    \caption{Heatmaps of the differences between the observed trade covariance matrices for the UK and the corresponding matrices for the ridge augmented FSC units during post-treatment periods (2017 Q3--2018 Q2).}
    \label{fig:service_difference_afsc_after}
\end{figure}

\begin{figure}[h]
    \centering
    \includegraphics[width=0.9\linewidth]{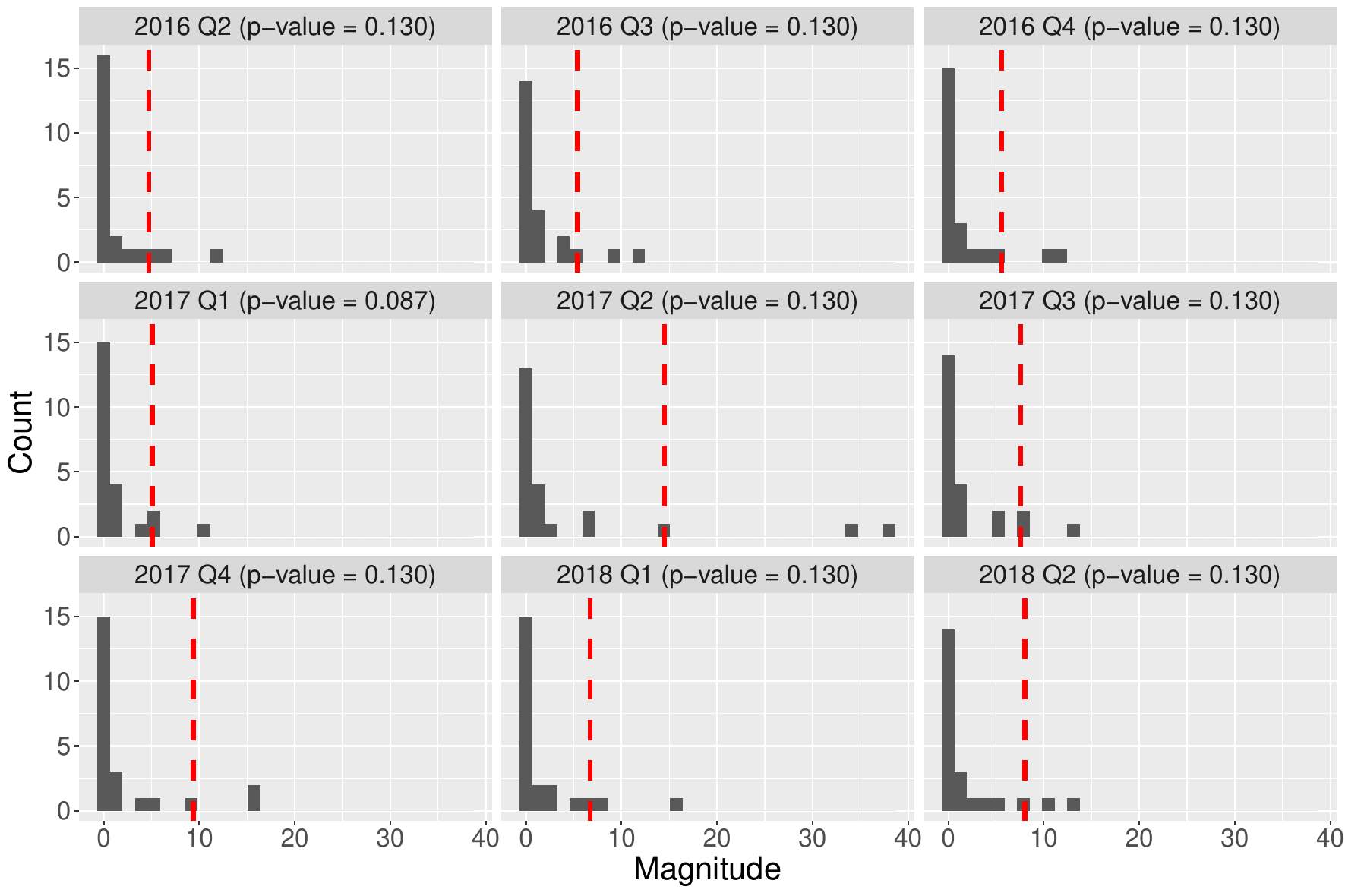}
    \caption{Results of the placebo permutation tests for the service trade data based on the ridge augmented FSC method. The histograms depict the magnitudes of the causal effects for all units, with the red dashed lines indicating the corresponding magnitudes for the treated unit.}
    \label{fig:placebo_service}
\end{figure}

\end{document}